\documentclass[11pt]{amsart}

\usepackage{amsmath, amsthm, amssymb, amsfonts, enumerate}
\usepackage[colorlinks=true,linkcolor=blue,urlcolor=blue]{hyperref}
\usepackage{dsfont}
\usepackage{color}
\usepackage{geometry}
\usepackage{todonotes}
\usepackage{graphicx}
\usepackage{subcaption} 
\usepackage{caption}
\usepackage{float}
\usepackage{mathtools}
\mathtoolsset{showonlyrefs} 
\usepackage{booktabs}

\def \cC{\mathcal{C}}
\def \cD{\mathcal{D}}

\def \cF{\mathcal{F}}
\def \cG{\mathcal{G}}
\def \cH{\mathcal{H}}

\def \cL{\mathcal{L}}

\def \cS{\mathcal{S}}
\def \cT{\mathcal{T}}

\def \bD{\mathbb{D}}
\def \bF{\mathbb{F}}
\def \bG{\mathbb{G}}
\def \bH{\mathbb{H}}
\def \P{\mathsf P}

\def \E{\mathsf E}

\def \N{\mathbb{N}}
\def \R{\mathbb{R}}

\def \ud{\hspace{1pt} \mathrm{d}}
\def \e{\mathrm{e}}
\def \given{\hspace{1.5pt}|\hspace{1.5pt}}
\def \Given{\hspace{1.5pt}\Big|\hspace{1.5pt}}

\newcommand{\eps}{\varepsilon}
\newcommand{\ind}{\mathbf{1}}

\newtheorem{theorem}{Theorem}[section]
\newtheorem{lemma}[theorem]{Lemma}
\newtheorem{corollary}[theorem]{Corollary}
\newtheorem{proposition}[theorem]{Proposition}

\newtheorem{remark}[theorem]{Remark}
\newtheorem{assumption}[theorem]{Assumption}

\title[Optimal Annuitization with Piecewise Deterministic Mortality]{Optimal Annuitization with stochastic mortality:\\
Piecewise Deterministic Mortality Force}
\author[Buttarazzi]{Matteo Buttarazzi}
\author[De Angelis]{Tiziano De Angelis}
\author[Stabile]{Gabriele Stabile}
\address{M.~Buttarazzi: Department of Mathematics, Uppsala University, Box 524,
751 20 Uppsala, Sweden.}
\email{\href{mailto:matteo.buttarazzi@math.uu.se}{matteo.buttarazzi@math.uu.se}}
\address{T.~De Angelis: School of Management and Economics, Dept. ESOMAS, University of Torino, Corso
Unione Sovietica, 218 Bis, 10134, Torino, Italy; Collegio Carlo Alberto, Piazza Arbarello 8, 10122,
Torino, Italy.}
\email{\href{mailto:tiziano.deangelis@unito.it}{tiziano.deangelis@unito.it}}
\address{G.~Stabile: Department MEMOTEF, Sapienza University of Rome, Via del Castro Laurenziano, 9, Rome, Italy}
\email{\href{mailto:gabriele.stabile@uniroma1.it}{gabriele.stabile@uniroma1.it}}

\numberwithin{equation}{section}

\begin{document}

\begin{abstract}
This paper addresses the problem of determining the optimal time for an individual to convert retirement savings into a lifetime annuity. The individual invests their wealth into a dividend-paying fund that follows the dynamics of a geometric Brownian motion, exposing them to market risk. At the same time, they face an uncertain lifespan influenced by a stochastic mortality force. The latter is modelled as a piecewise deterministic Markov process (PDMP), which captures sudden and unpredictable changes in the individual’s mortality force. The individual aims to maximise expected lifetime linear utility from consumption and bequest, balancing market risk and longevity risk under an irreversible, all-or-nothing annuitization decision. The problem is formulated as a three-dimensional optimal stopping problem and, by exploiting the PDMP structure, it is reduced to a sequence of nested one-dimensional problems. 
We solve the optimal stopping problem and find a rich structure for the optimal annuitization rule, which covers all parameter specifications.
Our theoretical analysis is complemented by a numerical example illustrating the impact of a single health shock on annuitization timing, along with a sensitivity analysis of key model parameters.
\end{abstract}

\maketitle
\section{Introduction}
Lifespan uncertainty makes retirement planning particularly challenging, requiring individuals to weigh guaranteed income options against more flexible but riskier alternatives.
A lifetime annuity is an insurance product that pays the annuitant a lifelong guaranteed income in exchange for a premium. It provides a valuable hedging instrument against both longevity risk (i.e., the risk of outliving one's savings) and market risk. However, annuitization is typically an irreversible decision. Once the annuity is purchased, the capital is locked into the contract and, although early withdrawals may be possible, they often incur significant penalties, see \cite{de2024variable} or \cite{landriault2021high}. 
Moreover, annuities do not normally provide residual value to the heirs upon the annuitant's death. Then, they are less appealing to individuals with a strong bequest motive. 

An alternative to annuitization is the so-called do-it-yourself approach, in which individuals manage their retirement wealth by investing in financial assets. This strategy offers greater flexibility and the potential for higher returns or larger bequests. However, it exposes individuals to both market volatility and longevity risk. 
As a result, the decision of whether and when to annuitize involves balancing multiple considerations, including investment risk, lifetime uncertainty, liquidity needs, and the desire to leave wealth to heirs.

In this paper, we study the problem of an individual whose retirement wealth at time zero is invested in the financial market, and it may be converted into an annuity at any time of the individual's choosing. There are two sources of randomness in our model: the financial market and the mortality force. The former is summarised by the evolution of a financial fund according to a geometric Brownian motion and the latter is described by a piecewise deterministic Markov process (PDMP). The optimal annuitization problem is formulated as an optimal stopping problem involving dividend payments before annuitization, a bequest in case of premature death of the individual and the value of the annuity at the time the individual decides to convert their wealth. We solve the problem in closed form, and we show that the annuitization decision exhibits radically different features, depending on the parameter choices. In particular, we observe four different geometries of the stopping set for the optimal stopping problem (cf.\ Theorem \ref{maintheorem}). That provides a rich economic interpretation of the interplay between financial and demographic risk.

Our model for the dynamics of the mortality force is novel, and we draw a detailed comparison to the existing literature in the next subsection. Afterwards, we are going to elucidate the mathematical contribution of the paper.

\subsection{Mortality models}

Mortality modelling is crucial for pricing annuities. 
Since Yaari's seminal contribution \cite{yaari1965uncertain}, a rich literature has analysed the annuitization decision under different mortality assumptions. 
One stream of research models mortality risk using a \textit{constant mortality force} (see, e.g., \cite{gerrard2012choosing}, \cite{liang2014optimal}, \cite{liang2018annuitization} and \cite{stabile2006optimal}). This approach implies that mortality risk does not vary with age or time, which oversimplifies the ageing dynamics, but it is appealing for mathematical tractability.
To better reflect age-related mortality dynamics, another line of research employs a \textit{deterministic time-dependent mortality force} (see \cite{de2019free}, \cite{hainaut2014optimal}, \cite{hainaut2006life}, \cite{milevsky2006optimal}, \cite{milevsky2007annuitization} among others). 
These models capture the increasing risk of death with age by allowing the mortality force to follow a predetermined, time-dependent trajectory. 
\textit{Stochastic mortality models} are better suited to capture uncertainty by incorporating randomness into the evolution of the mortality force. 
The Lee-Carter model \cite{lee1992modeling} pioneered discrete-time stochastic mortality modelling, and it has been extended in several directions (cf.,\ e.g., \cite{renshaw2003lee} and \cite{li2005coherent}).
The first continuous-time stochastic mortality model was introduced by Milevsky and Promislow \cite{milevsky2001mortality}, who employed a mean-reverting Brownian Gompertz process to price mortality-linked derivatives. Dahl \cite{dahl2004stochastic} later proposed a more general diffusion framework, which includes the Milevsky-Promislow model as a special case. 
These foundational contributions have inspired a broad and growing literature on stochastic models.

Incorporating \textit{jumps} into mortality models has been a clear focus of recent research. In a discrete-time setting, Chen and Cox \cite{chen2007modeling} extended the Lee-Carter model by introducing permanent and transitory changes in the mortality force.
Biffis \cite{biffis2005affine} introduced the first continuous-time stochastic mortality framework with jumps, employing an affine jump-diffusion specification to model mortality force. Luciano and Vigna \cite{luciano2008mortality} show that a jump process is more effective than a diffusive component alone in stochastic mortality modelling. Finally, Hainaut and Devolder \cite{hainaut2008mortality} explore pure-jump Lévy processes to model the mortality force. Additional contributions to this very vast literature can be found in \cite{ahmadi2015modeling}, \cite{cox2010mortality}, \cite{cox2006multivariate}, \cite{milidonis2011mortality}, among others.

Our work on optimal annuitization builds on the framework of \cite{de2019free} and \cite{hainaut2014optimal}. In those papers, the insurance company prices the annuity using an {\em objective} mortality force, whereas the individual evaluates their expected future cashflows (before and after annuitization) using a {\em subjective} mortality force. This ambivalence generates complex mechanisms that incentivise/disincentivise annuitization depending on the individual's perceived {\em fairness} of the annuity price. 
Those papers use a deterministic, time-dependent (objective and subjective) mortality force. Instead, we incorporate a stochastic mortality model to capture uncertainty in the {\em subjective} mortality of the individual. This is reasonable, because the {\em objective} mortality force is calculated as an average over the population and therefore it is less affected by idiosyncratic risk.     
We assume that the individual's subjective mortality force evolves in continuous time as a PDMP, as introduced by Davis \cite{davis2018markov}, taking constant (random) values between jumps.
The randomness of the jump times reflects the uncertainty surrounding the occurrence of health-related events, while the distribution of jump sizes captures the uncertainty in the magnitude of their impact on the individual’s health. Together, these ingredients capture a broad spectrum of possible mortality trajectories. For instance, large but infrequent jumps represent significant health shocks, such as severe illnesses or accidents, that substantially increase mortality risk. Alternatively, frequent small jumps represent the gradual, random worsening of life conditions.
For the sake of tractability, we assume that the PDMP has finitely many jumps. From an applied perspective, this assumption carries no loss of generality because the number $N$ of jumps can be chosen arbitrarily large. From a mathematical perspective instead it enables a recursive solution method.

At the modelling level, we depart from \cite{milevsky2007annuitization}, \cite{hainaut2014optimal} and \cite{de2019free}, where the mortality force is deterministic, by introducing a stochastic mortality model. 
Optimal annuitization under stochastic mortality has been studied before (see, e.g., \cite{reichling2015optimal} for a discrete-time model of consumption and portfolio allocation). However, our paper appears to be a first theoretical study in continuous time that casts the question as an optimal timing of an annuity purchase with stochastic mortality. As a first step, the choice of PDMP is mostly for mathematical convenience.
It allows, for example, a reduction of the problem to a nested sequence of one-dimensional optimal stopping problems that would be impossible with a diffusive model for the mortality force. Despite the simplification, we already observe key features of our solution that are not seen in \cite{milevsky2007annuitization}, \cite{hainaut2014optimal} and \cite{de2019free}. Setting aside \cite{milevsky2007annuitization}, which finds deterministic annuitization times, \cite{hainaut2014optimal} and \cite{de2019free} obtain optimal annuitization thresholds as functions of time. In our framework, instead, optimal annuitization thresholds are not functions of time. They are functions $(n,\mu)\mapsto b(n,\mu)$ of the current mortality force $\mu$ and of the number of its past jumps $n$. Time itself is not relevant, and we can study how uncertainty on future health conditions affects the individual's annuitization decisions. Evaluating the optimal boundary along trajectories $(n_t(\omega),\mu_t(\omega))_{t\ge 0}$ of the PDMP yields stochastic optimal annuitization thresholds $t\mapsto b(n_t(\omega),\mu_t(\omega))$ so that the individual must make decisions based on the joint evolution of the demographic risk factors and the financial risk factor (the wealth process). On the contrary, in \cite{hainaut2014optimal} and \cite{de2019free}, the individual at time zero knows the geometry of the optimal boundary $t\mapsto \hat b(t)$ as a function of time and therefore the only uncertainty left in determining the optimal annuitization time is coming from the financial risk factor.

\subsection{Mathematical contribution}
From a mathematical perspective, the problem is initially formulated as a three-dimensional continuous-time optimal stopping problem.
The individual invests their retirement wealth in a dividend-paying fund following a geometric Brownian motion. Under the \textit{all-or-nothing} framework (e.g., \cite{milevsky2007annuitization}, \cite{de2019free}), the individual has the option at any time to irreversibly convert their entire wealth into an immediate lifetime annuity. 
The individual seeks to maximise expected discounted value of dividends and annuity payments. In order to incorporate the individual's bequest motive, we also introduce a linear payoff which is received by the individual's heirs in case of death prior to annuitization. Such a linear term is modulated by a parameter $\nu \in [0,1]$ that captures the strength of the individual’s bequest motive. To justify the presence of the parameter $\nu$ we propose two interpretations. On the one hand, we may think of the problem as a utility maximisation with linear utility, where $\nu$ reflects the individual's preferences. On the other hand, we may think of it as a rational pricing exercise in which $\nu$ plays the role of inheritance tax.

In this paper, we do not investigate the impact of risk aversion, because our aim is to isolate the effect of PDMP mortality on the annuitization decision. In order to derive explicit results, we focus on a risk-neutral individual, similarly to \cite{de2019free} and  \cite{hainaut2014optimal}. Our work complements a stream of the literature that instead explores risk-aversion with somewhat mixed conclusions. In a number of papers, it is argued that  risk aversion is positively correlated with annuity demand (see \cite{horneff2008following} and \cite{inkmann2011deep} among others). The argument there is that annuities provide a steady income for life and hence they protect the retiree against the risk of exhausting her assets. According to other papers instead an increase in risk aversion leads to lower demand for annuities (see e.g. \cite{bommier2014too}, \cite{andre2022impact}). That happens because annuities may be perceived as a gamble: their purchase generates a substantial loss, should the individual die prematurely.

The state dynamics in the optimal stopping problem is three-dimensional because it includes the fund's value $X$, the number of health shocks $n$ that have occurred and the current state of the mortality force $\mu$. However, by exploiting the structure of the piecewise deterministic mortality force process and applying the dynamic programming principle, we reduce the complexity by transforming the original problem into a sequence of nested one-dimensional optimal stopping problems. The transformation enables a recursive formulation of the problem. 
Thanks to that, we need to solve the optimal annuitization problem {\em only} in between any two consecutive jumps of the mortality force: if at the jump time the individual has not annuitized yet, their payoff equals the value function of the same problem with the current fund's value but with a new state for the mortality force. Despite the dimension reduction, the resulting problem is not amenable to explicit solution methods based on a guess-and-verify approach, because the recursive nature of the payoff remains rather implicit. We perform a careful probabilistic study of the value function that leads us to a {\em complete} solution of the problem, and it unveils a rich structure for {\em all} possible optimal annuitization rules. We find four different annuitization regimes which may coexist within the individual's life span (i.e., switching from one state of mortality to another, the individual may change their annuitization strategy across the ones illustrated in Theorem \ref{maintheorem}). In particular, for a given value of the mortality force, optimal annuitization may occur in any of the following, mutually exclusive, situations: (i) when the retirement wealth exceeds an upper threshold, (ii) when it falls below a lower threshold, (iii) never, no matter the retirement wealth or (iv) immediately, no matter the retirement wealth.

We also prove that the value function satisfies a free boundary problem that reflects the nested structure of the payoff, and it satisfies the \emph{smooth fit condition} in the cases with nontrivial continuation and stopping regions.
The theoretical analysis is complemented by a numerical study of our framework, followed by a sensitivity analysis showing how the optimal annuitization thresholds respond to variations in key model parameters.

\subsection{Structure of the paper}

The rest of the paper is structured as follows. In Section \ref{ProblemFormulation}, we introduce the mathematical formulation of the problem and state our main result. Section \ref{numericalanalysis} presents a numerical illustration of the theoretical findings, focusing on the case of a single jump in the mortality force. The theoretical analysis is developed in two stages, enabled by the recursive representation of the value function presented in Proposition \ref{Prop:recursive}. In Section \ref{ConstantForceOfMortality}, we first consider the case of a constant mortality force. Building on that, in Section \ref{OptimalAnnuitizationwithJumps} we extend the analysis to the more general model. Finally, Appendix \ref{app:markovian} contains technical results used throughout the paper.

\section{Problem Formulation}
\label{ProblemFormulation}
In this section, we present a mathematical formulation of the optimal annuitization problem. Our framework incorporates uncertainty in both mortality and financial markets. 
We assume that the individual and the insurance company share the same beliefs regarding the financial market, but may hold different views about demographic risk.
\subsection{Financial and demographic modelling} 
Consider a probability space $(\Omega,\cG,\P)$ equipped with a Brownian motion $(B_t)_{t\ge 0}$. 
Let $\bF^B\coloneqq (\cF^B_t)_{t\geq 0}$ be the right-continuous completion of the Brownian filtration with the $\P$-null sets, with $\cF^B_\infty\coloneqq \vee_{t\geq 0} \cF^B_t$. 
The individual's wealth dedicated to retirement needs is invested in a financial fund, the value of which is denoted by $(X_t)_{t\ge0}$, and it evolves according to 
\begin{align}\begin{aligned}
\ud X_t = (\theta-\alpha)X_t\ud t+\sigma X_t \ud B_t, \ \ \ \ \ \ \ t>0,
\label{Eq:X}
\end{aligned}\end{align}
with $X_0=x$,
where $\theta>0$ is the average continuous return of the financial investment, $\alpha\geq 0$ is the constant dividend rate and $\sigma>0$ is the volatility coefficient. We use $X^x_t$ when we want to stress the starting point $x$ of the dynamics.

We consider an individual of age $\eta\ge 0$ at time $0$. The value of $\eta$ is fixed throughout the paper, and with no loss of generality, time is measured in years. The individual uses a subjective mortality force to assess their survival probability.
We model the subjective mortality force with a Piecewise-Deterministic Markov Process (PDMP), see \cite{davis2018markov}.
The state space for the PDMP must account for the number of jumps that have occurred and the current value of the mortality force. It is denoted $E=\N_0\times[\mu_m,\mu_M]$ for constants $0<\mu_m\leq \mu_M<\infty$ representing the range of possible values of the mortality rate and $\N_0\coloneqq\N\cup\{0\}$.
Let $\mathcal{E}$ denote the Borel $\sigma$-algebra on $E$ and let $\mathcal{P}(E)$ be the set of probability measures on $(E,\mathcal{E})$. In order to describe the dynamics of the PDMP, we consider a mapping $y\mapsto Q(y,\cdot)$ from $E$ to $\mathcal{P}(E)$ with the properties: (i) for each fixed $A\in\mathcal{E}$, the map $y\mapsto Q(y,A)$ is measurable, and (ii)
$Q(y,\{y\})=0$.
Let $(\lambda_n)_{n\in\N}$ be a sequence in $[0,\infty)$ and $(e_n)_{n\in\N}$ be a sequence of independent random variables with $e_n\sim \exp(\lambda_n)$. Define $\tau_0=0$ and
\begin{align}\begin{aligned} \tau_n\coloneqq \sum_{j=1}^n e_j , \ \ \ \ \ n\in\N.\end{aligned}\end{align}
Note that $(e_n)_{n\in\N}$ is the sequence of so-called interarrival times, i.e., the time intervals between consecutive jumps in the mortality force. The time $\tau_n$ is the \textit{n}-th jump time.
Given a deterministic
$\bar \mu_0\in [\mu_m,\mu_M]$, for every $n\in\N$ we let $\bar \mu_n:\Omega \to [\mu_m,\mu_M]$ be a random variable. Then we introduce a process $(\chi_t)_{t\geq0}$, starting at time zero from $\chi_0=(0,\bar \mu_0)$ and defined as
\begin{align}
    \begin{aligned}
        &\chi_t(\omega)\coloneqq (n,\bar \mu_n(\omega)), \hspace{0.5cm} \forall t \in [\tau_n(\omega),\tau_{n+1}(\omega)),
    \end{aligned}
\end{align}
for $n\in\N_0$. The law of the process is specified as
\begin{align}\begin{aligned} \P(\chi_{\tau_{n+1}}\in A\given \sigma(\chi_{\tau_n}))=Q(\chi_{\tau_n},A), \hspace{0.5cm} A \in \mathcal{E}. \end{aligned}\end{align}

\begin{remark}\label{rem:q}
In the general construction described in \cite{davis2018markov}, both coordinates in the process $(\chi_t)_{t\ge 0}$ jump to random values. Instead, in our model, the first coordinate of the process $(\chi_t)_{t\ge 0}$ represents the counter of jumps in the mortality force; therefore, it increases by one after each jump, with probability one. 
In particular, in our setting $Q((n,\mu),\{n+1\}\times\Gamma)=q(n,\mu,\Gamma)$ for any Borel subset $\Gamma$ of $[\mu_m,\mu_M]$, where for each $(n,\mu)$ the function $q(n,\mu,\cdot)$ is a probability measure on the Borel $\sigma$-algebra on $[\mu_m,\mu_M]$. 
\end{remark}

We define $\bH\coloneqq (\cH_t)_{t\ge 0}$ as the natural filtration of $(\chi_t)_{t\geq 0}$ augmented with the $\P$-null sets and with $\cH_\infty\coloneqq \vee_{t\geq 0} \cH_t$. 
We know that $\bH$ is right-continuous, see \cite[Ch.\ 2.25, Thm.\ 25.3, p. 63]{davis2018markov} and $(\chi_t)_{t\geq 0}$ is a homogeneous strong Markov process with respect to $\bH$, see \cite[Ch.\ 2.25, Thm.\ 25.5, p. 64]{davis2018markov}. Since $\bar\mu_0$ is deterministic, $\cH_0=\{\varnothing,\Omega\}$. 

It is useful to consider a pathwise representation of the process $(\chi_t)_{t\geq0}$. Indeed we denote $(\chi_t(\omega))_{t\geq 0}=(n_t(\omega),\mu_t(\omega))_{t\geq 0}$, with
\begin{align}\begin{aligned} 
n_t(\omega)\coloneqq\sum_{i=1}^\infty \ind_{\{t\geq \tau_i(\omega)\}} \quad\text{and}\quad \mu_t(\omega)\coloneqq\bar\mu_{n_t}(\omega)=\bar \mu_0 + \sum_{i=1}^{n_t(\omega)} (\bar \mu_i(\omega)-\bar \mu_{i-1}(\omega)),
\end{aligned}\end{align}
with the convention $\sum_{i=1}^0 = 0$.
Due to the Markovian structure of the process $(\chi_t)_{t\ge 0}$ we allow for an arbitrary initial state $\chi_0=(n,\mu)\in\N_0\times[\mu_m,\mu_M]$ and denote the time of the first jump from such an initial state as
\begin{align}\label{eq:xi}
\begin{aligned} 
\xi_n\coloneqq \inf\{t>0: n_t\neq n\}. 
\end{aligned}\end{align} 
By our construction, the law of $\xi_n$ is exponential with parameter $\lambda_{n+1}$.

We are going to assume independence between financial and mortality risks. Similarly, we assume independence between the jump time $\tau_n$ and the subsequent value of the mortality force $\bar\mu_n$, for $n\in\N$. However, we maintain that $\bar\mu_{n+1}$ may depend on $\bar\mu_n$, because an individual's future health shall be linked to their current one. The above discussion is summarised in the next formal assumption.
\begin{assumption} \label{Ass:indipendence}
The $\sigma$-algebras $\cF^B_\infty$ and $\cH_\infty$ are independent. Moreover, $\xi_n$ and $\bar \mu_{n+1}$ are also independent.
\end{assumption}

Let $\bF\coloneqq (\cF_t)_{t\geq0}$ with $\cF_t=\cF^B_t\vee\cH_t$ for $t\ge 0$ and $\cF_\infty=\vee_{t\geq0}\cF_t$. Under the first part of Assumption \ref{Ass:indipendence} the filtration $\bF$ is right-continuous, see \cite[Ch.\ 1.1.4, Prop.\ 1.12, p.\ 6]{aksamit2017enlargement}. We use a Cox process to model the time of death of the policyholder. Let $\Theta$ be independent of $\cF_\infty$ and exponentially distributed with parameter one. Following the construction as in \cite[Ch.\ 2.3.1]{aksamit2017enlargement}, we define the individual's time of death \begin{align}\begin{aligned} 
\tau_d\coloneqq \inf \{t\geq 0: \Lambda_t\geq \Theta\},
\end{aligned}\end{align}
where $\Lambda_t\coloneqq \int_0^t \mu_s\ud s$. Since $(\Lambda_t)_{t\ge 0}$ is non-decreasing then $\{\tau_d>s\}=\{\Lambda_s<\Theta\}$. Moreover $\Lambda_t<\infty$ for all $t\ge 0$ and $\Lambda_\infty=\infty$. Therefore $\tau_d<\infty$, $\P$-a.s. The process $t\mapsto\ind_{\{\tau_d\le t\}}$ generates a filtration $\bD\coloneqq (\cD_t)_{t\ge 0}$ with
$\cD_t\coloneqq \sigma(\ind_{\{\tau_d\leq s\}}, s\leq t)$ and the usual convention $\cD_\infty=\vee_{t\ge 0}\cD_t$. We introduce the enlarged filtration $\bG\coloneqq (\cG_t)_{t\ge0}$ with $\cG_{t}\coloneqq \cap_{u>t} (\cF_u\vee \cD_u)$ for $t\ge 0$.
Note that $\bG$ is the smallest enlargement of $\bF$, satisfying the usual assumption of right-continuity and completion,  for which the random time $\tau_d$ is a stopping time \cite[Ch.\ VI.3, p. 370]{protter2005stochastic}. 
Moreover, the so-called \textit{Immersion Property} holds, i.e., any $\bF$-martingale is a $\bG$-martingale, see \cite[Ch.\ 2.3.2, Lemma 2.28, p. 43]{aksamit2017enlargement}.

It is well-known that (cf.\ \cite[Ch.\ 2.3.2, Lemma 2.25, p. 42]{aksamit2017enlargement}) 
\begin{align}\begin{aligned} \label{Eq:tau_d_F_infty} \P(\tau_d>s\given \cF_s)=\P(\tau_d>s\given\cF_\infty)=\exp(-\Lambda_s),\quad \text{for all $s\ge 0$}. 
\end{aligned}\end{align}
This result allows us to give an expression for the individual's {\em subjective} probability of survival $_s p_{\eta+t}$ (i.e., the probability at age $\eta+t$ to survive the next $s$ years). The latter is defined for $t,s\ge 0$ by
\begin{align}\label{Def:Subj_prob}
\begin{aligned}
_s p_{\eta+t}&=\P(\tau_d>t+s\given\cF_{t+s},\tau_d>t)\coloneqq \frac{\P (\tau_d>t+s,\tau_d>t \given \cF_{t+s})}{\P(\tau_d>t\given\cF_{t+s})}=\exp\big[-\big(\Lambda_{t+s}-\Lambda_t\big)\big].
\end{aligned}
\end{align}

The insurance company instead relies on a so-called \textit{objective} survival probability function, which is calculated using an \emph{objective} mortality force derived from a demographic analysis of the population. For simplicity, we assume a constant objective mortality force $\hat{\mu}\in [\mu_m,\mu_M]$ so that the objective survival probability reads
\begin{align}\begin{aligned}
_s \hat{p}_{\eta+t}\coloneqq  \e^{-\hat{\mu}s},\quad t,s\ge 0.
\end{aligned}\end{align}
The different survival probabilities adopted by the insurer and the individual account for the imperfect information available to the insurer on the individual's mortality risk profile.

\subsection{Optimization problem} 
The individual invests in the financial market and collects dividends at a rate $\alpha$ on the financial fund. They may choose a random time $\tau:\Omega\to[0,\infty]$ at which they convert their entire wealth into an annuity (we adopt the so-called \textit{all or nothing} framework). 
After annuitization, i.e., for $t\ge \tau(\omega)$ dividend payments stop and the individual receives the annuity payment at a constant annual rate $P_\tau$. For a fixed constant $K\in\R$ the annuity rate is given by
\begin{align}\begin{aligned} \label{Def:annuitypayments}
P_{\tau}\coloneqq \frac{X_\tau-K}{\hat{a}_{\eta+\tau}},
\end{aligned}\end{align}
where
\begin{align}\begin{aligned}
\hat{a}_{\eta+t}\coloneqq  \int_{0}^{\infty} \e^{-\hat{\rho} s} {}_s \hat{p}_{\eta+t}\ud s=\frac{1}{\hat{\rho}+\hat{\mu}},
\end{aligned}\end{align}
is the price that the insurance company charges for a unitary lifetime annuity and $\hat{\rho}>0$ is the interest rate guaranteed by the insurer. The constant $K$ is either an acquisition fee ($K\ge 0$) or an incentive ($K<0$). 
Should the individual pass away before the annuitization time, i.e., on the event $\{\tau\ge \tau_d\}$, they leave a bequest equal to their current wealth.

The individual weighs the cashflows using a linear utility. We assume that the linear utility of dividend and annuity payments has unit slope, whereas the utility from bequest is modulated by a parameter $\nu\in[0,1]$, measuring the strength of bequest motives. Putting all the above considerations together, for a given annuitization time $\tau$ the individual's payoff reads 
\[
\int_{0}^{\tau_d \wedge \tau }\e^{-\rho t}\alpha X_t\ud t  +\ind_{\{\tau_d \leq \tau \}}\e^{-\rho \tau_d} \nu X_{\tau_d} +P_{\tau}\hspace{-3pt}\int_{\tau_d\wedge \tau}^{\tau_d}\e^{-\rho t}\ud t,
\]
where $\rho>0$ is a subjective discount rate.

The goal of the individual is to maximise the expected payoff by choosing optimally the annuitization time $\tau$. 
Thanks to \cite[Ch.\ VI.3, Lemma, p.\ 370]{protter2005stochastic} we know that for any $\tau\in\cT(\bG)$ there is $\sigma\in\cT(\bF)$ such that $\tau\wedge\tau_d=\sigma\wedge\tau_d$, $\P$-a.s. Therefore, for the purpose of the optimization problem at hand, it is equivalent to pick $\bG$-stopping times or $\bF$-stopping times 
(cf.\ also \cite[Appendix A]{de2019free}). The problem reads: 
\begin{align}\begin{aligned}
\label{Def:Value_f_nonmrkv}
V_0\coloneqq\sup_{\tau \in\cT(\bF)} \E \Big[  \int_{0}^{\tau_d \wedge \tau }\e^{-\rho t}\alpha X_t\ud t  +\ind_{\{\tau_d \leq \tau \}}\e^{-\rho \tau_d} \nu X_{\tau_d} +P_{\tau}\hspace{-3pt}\int_{\tau_d\wedge \tau}^{\tau_d}\e^{-\rho t}\ud t \Big]. 
\end{aligned}\end{align}

It is shown in Proposition \ref{prop:markovian} in Appendix \ref{app:markovian} that the above problem can be cast into a Markovian framework. The basic idea is that we can integrate out the independent random variable $\Theta$ that determines $\tau_d$. By doing so we obtain a problem formulation in which the payoff depends explicitly upon the Markov process $(X_t,n_t,\mu_t)_{t\ge 0}$.
In particular, from Proposition \ref{prop:markovian}, we have $V_0=V(X_0,0,\mu_0)$ for a measurable function $V(x,n,\mu)$ defined as 
\begin{align}
    \begin{aligned} \label{Eq:Value_f_mrkv}
    V(x,n,\mu)\coloneqq \sup_{\tau \in\cT(\bF)} \E_{x,n,\mu} \Big[& \int_{0}^{\tau}\e^{-\int_0^t (\rho+\mu_s)\ud s} ( \alpha+\nu \mu_t ) X_t\ud t\\
    &+ \e^{-\int_0^\tau (\rho+\mu_s)\ud s} \hat{f}(n_\tau,\mu_\tau)(X_\tau-K)\Big],
    \end{aligned}
\end{align}
where $\P_{x,n,\mu}(\cdot)=\P(\cdot|X_0=x,n_0=n,\mu_0=\mu)$\footnote{With this notation $\mu_0=\bar \mu_{n_0}=\bar \mu_n$.} and, letting 
\begin{equation} 
\label{Eq: a_eta+t}
 a_{\eta+t}\coloneqq  \E \Big[ \int_{0}^{\infty} \e^{-\rho u} {}_u p_{\eta+t}\ud u \Given \cF_t \Big],
\end{equation}
be the individual's valuation of the unitary annuity, in the second term under expectation, we have 
\begin{align}\begin{aligned} \label{fhatcmu} 
\hat{f}(n_t,\mu_t)= \frac{a_{\eta+t}}{\hat{a}_{\eta+t}}, 
\end{aligned}\end{align} 
for a suitable $\mathcal{E}$-measurable function $\hat f$.
The function $\hat{f}(n,\mu)$ is the so-called ``money's worth'' and it measures the annuity's attractiveness from an individual's perspective. Specifically, it gives the ratio between the individual's valuation of the unitary annuity and its market price. When $\hat{f}(n,\mu)=1$, the annuity is {\em actuarially fair}. Values of $\hat f(n,\mu)$ above/below 1 indicate that the individual perceives the annuity as under/over-priced relative to its actuarial value. 

The next assumption ensures well-posedness of the problem in \eqref{Eq:Value_f_mrkv} (cf.\ Proposition \ref{vofinitinessl}).
\begin{assumption}
\label{mainass}
    $\theta-\alpha-\rho-\mu_m< 0$.
\end{assumption}

\subsection{A recursive formulation via DPP}\label{conludingsection}
The special case of interest in this paper is the one in which the mortality force has finitely many jumps. We express this feature with the following assumption:
\begin{assumption} \label{assnomorejumps}
There is $N\in\N_0$ such that $\bar\mu_{N+m}(\omega)=\bar\mu_N(\omega)$ for all $m\in\N$ and $\omega\in\Omega$ (i.e., $\lambda_{N+1+m}=0$ for all $m\in\N_0$).
\end{assumption}

From an applied perspective, the assumption brings no loss of generality because $N$ can be chosen arbitrarily large. From a mathematical perspective, instead, the assumption will allow us to use a recursive approach to the solution of the problem.
We keep track of the number of jumps in the optimisation problem by relabelling $V(x,n,\mu)=V^N(x,n,\mu)$ and notice that $n\le N$. 
Using the PDMP structure, the problem can be reformulated recursively between consecutive jump times of the mortality force. That leads to a sequence of one-dimensional optimal stopping problems indexed by $(n,\mu)$.
The formula in \eqref{VF:recursive} below has a flavour of dynamic programming principle. Indeed, it says that an individual with initial state $(x,n,\mu)$ chooses a stopping time for the Brownian filtration $\bF^B$ and solves a stopping problem with constant mortality $\mu$ but with an additional term in the running payoff. The latter incorporates the value function of the original problem \eqref{Eq:Value_f_mrkv} after the jump in the mortality force, averaged with respect to the transition kernel $q(n,\mu,\cdot)$. 
The proof of the next proposition is postponed to Appendix \ref{app.proof}.

\begin{proposition} \label{Prop:recursive}
    For $n\in\N_0$, $n\le N$, letting $r_n(\mu)\coloneqq \rho+\mu+\lambda_{n+1}$, it holds
    \begin{align}
        \begin{aligned}\label{VF:recursive}
            V^N(x,n,\mu)&= \!\!\sup_{\tau \in\cT(\bF^B)} \E_{x,n,\mu} \Big[ \e^{-r_n(\mu) \tau}  \hat{f}(n,\mu) (X_\tau\!-\!K)\\
            &\qquad\qquad\qquad+\int_0^\tau\!\! \e^{-r_n(\mu) t} \big[( \alpha\!+\!\nu \mu ) X_t \!+\!\lambda_{n+1} \hat{V}^N(X_t,n\!+\!1;\mu) \big]\ud t  \Big],
    \end{aligned}\end{align}
    where
\begin{align}\begin{aligned} \label{Eq:hatV}\hat{V}^N(x,n+1;\mu)\coloneqq \int_{\mu_m}^{\mu_M} V^N (x,n+1,z)q(n,\mu,\ud z), \quad n\le N-1,
\end{aligned}
\end{align} 
with $q(n,\mu,\cdot)$ introduced in Remark \ref{rem:q} and $\hat V^N (x,N+1 ; \mu)=0$.
\end{proposition}

Proposition \ref{Prop:recursive} yields an optimal stopping problem on the 1-dimensional diffusion $(X_t)_{t\ge 0}$. The value function depends explicitly on $(n,\mu)$ in a parametric way. Thus, we can also replace $\E_{x,n,\mu}$ with the easier $\E_x$. Thanks to the explicit form of the process $X$, we will often use the equivalence between the law of $X$ under $\P_x$ and the law of $X^x$ under $\P$.

\subsection{Preliminary results from optimal stopping theory}\label{sec:prelimOS}

Some standard results from optimal stopping theory allow us to get started with our analysis.
We define the stopping region and the continuation region, respectively, as follows
\begin{align}
\begin{aligned}
\cS&\coloneqq \{(x,n,\mu): V^N(x,n,\mu)=\hat{f}(n,\mu)(x-K)\},\\ 
\cC&\coloneqq \{(x,n,\mu): V^N(x,n,\mu)>\hat{f}(n,\mu)(x-K)\}. 
\end{aligned}
\end{align}

The $(n,\mu)$-sections of the stopping and continuation regions read as
\begin{align} \label{SCRegionsV}
\begin{aligned}
\mathcal{S}_{n,\mu}&\coloneqq \{x\in(0,\infty):V^N(x,n,\mu)=\hat{f}(n,\mu)(x-K)\}, \\
\mathcal{C}_{n,\mu}&\coloneqq \{x\in(0,\infty):V^N(x,n,\mu)>\hat{f}(n,\mu)(x-K)\}. 
\end{aligned}
\end{align}
Let 
$\tau^*_{n,\mu}\coloneqq \inf\{t\geq 0: X_t \in \mathcal{S}_{n,\mu}\}$. Continuity of the mapping $x\mapsto V^N(x,n,\mu)$ is shown in Propositions \ref{Prop:ConvN} and \ref{convexw} from first principles.
Then, $\tau^*_{n,\mu}$ is optimal for the problem formulation in \eqref{VF:recursive} by \cite[Ch.\ 3.3, Thm. 3, p. 127]{shiryaev2007optimal}. Moreover, letting $\tau^*\coloneqq \inf \{t\geq 0: (X_t,n_t,\mu_t)\in \cS\}$ 
we have $\tau^*_{n,\mu}\wedge \xi_n=\tau^*\wedge\xi_n$, $\P_{x,n,\mu}$-a.s. 
Finally, we recall that 
\begin{align*}
\begin{aligned}
t\mapsto &\,\e^{-r_n(\mu)(t\wedge \tau)}V^N(X_{t\wedge \tau},n,\mu)+\int_0^{t\wedge \tau}\e^{-r_n(\mu)s}\big[(\alpha+\nu\mu)X_s+\lambda_{n+1}\hat V^N(X_s,n+1;\mu)\big]\ud s,
\end{aligned}
\end{align*}
is a super-martingale for any $\tau\in\cT(\bF^B)$ and it is a martingale for $\tau=\tau^*_{n,\mu}$. That also yields an analytic characterisation of the value function. 
Let us introduce the infinitesimal generator $\cL$ of $(X_t)_{t\geq 0}$, defined as
\begin{align}\label{L}
(\cL u)(x)\coloneqq \tfrac{1}{2} \ \sigma^2 x^2 u_{xx}(x)+(\theta-\alpha)x u_{x}(x), \quad \text{for any} \ u\in C^2(\R),
\end{align}
where $u_x=\partial u/\partial x$ and $u_{xx}=\partial^2 u/\partial x^2$.
If $V^N(\cdot,n,\mu)\in C((0,\infty))$, then $\cC_{n,\mu}$ is an open set and from the martingale condition above $V^N(\cdot,n,\mu)\in C^2(\cC_{n,\mu})$ satisfies
\begin{align}\label{eq:ODE}
\begin{aligned}
\big[\cL V^N-r_n(\mu) V^N\big](x,n,\mu)=-(\alpha+\nu\mu)x-\lambda_{n+1} \hat V^N(x,n+1;\mu), \quad x\in\cC_{n,\mu}.
\end{aligned}
\end{align}
The proof of the above statement is completely standard and therefore omitted (cf.\ \cite[Ch.\ 3.8, Thm. 15, p. 157]{shiryaev2007optimal} and \cite[Ch. 3]{peskir2006optimal}).

For future reference, we recall that 
given $\mu\in[\mu_m,\mu_M]$ and $n\in\N_0$, $n\le N$, the fundamental solutions of the ODE $\cL u=r_n(\mu) u$ are $\psi(x,n,\mu)\coloneqq x^{\gamma_{n}^+(\mu)}$ and $\phi(x,n,\mu)\coloneqq x^{\gamma_{n}^-(\mu)}$, where  
\begin{equation} \label{gammarpm}
\gamma^{\pm}_{n}(\mu)\coloneqq  \frac{1}{2}-\frac{\theta-\alpha}{\sigma^2} \pm \sqrt{\Big(\frac{1}{2}-\frac{\theta-\alpha}{\sigma^2}\Big)^2+\frac{2 r_n(\mu)}{\sigma^2}},
\end{equation}
are such that $\gamma^+_n(\mu)>1$ and $\gamma^-_n(\mu)<0$ (cf.\ \cite[Ch.\ II.1, pp.\ 18-19]{borodin2015handbook}).

\subsection{Optimal annuitization strategies}
It is convenient to work with the \textit{Lagrange formulation} of the problem. For $n \leq N$, define
\begin{align}
\label{ConnectionWandVnjumps}
W^N(x,n,\mu) \coloneqq V^N(x,n,\mu) - \hat{f}(n,\mu)(x - K).
\end{align}
An application of Dynkin's formula to $\e^{-r_n(\mu)\tau}\hat f(n,\mu)(X_\tau-K)$ leads to the representation (cf.\ Section \ref{ConstantForceOfMortality} for details): 
\begin{align}
\label{w}
W^N(x,n,\mu) = \sup_{\tau\in\cT(\bF^B)} \E_x \Big[ \int_0^\tau e^{-r_n(\mu)t} M^N(X_t, n, \mu)\, \ud t \Big],
\end{align}
where 
\begin{align}
\label{Def:M_njumps}
\begin{aligned}
M^N(x,n,\mu) &\coloneqq r_n(\mu) \hat{f}(n,\mu) K \\
&\quad + (\hat{f}(n,\mu) - \beta_n(\mu))(\theta - \alpha - r_n(\mu)) x + \lambda_{n+1} \hat{V}^N(x,n+1;\mu),
\end{aligned}
\end{align}
with
\begin{equation} \label{betanojump} 
\beta_n(\mu) \coloneqq x^{-1}\E_x\Big[\int_0^\infty \e^{-r_n(\mu)s}(\alpha+\nu \mu)X_s \ud s\Big]= \frac{\alpha+\nu \mu}{r_n(\mu)+\alpha-\theta}.
\end{equation}

The next two theorems summarise the regularity of the value function and the structure of the optimal annuitization rule for all parameter choices.
A rigorous proof of the theorem is presented in Section \ref{OptimalAnnuitizationwithJumps}, after a detailed technical study.
We begin by introducing some ingredients required to state the results.
For $\mu\in[\mu_m,\mu_M]$ and $n\le N-1$ let
\begin{align}\label{eq:Lnmu}
\begin{aligned}
        R(n,\mu) &\coloneqq(\hat{f}(n,\mu)-\beta_{n}(\mu))(\theta-\alpha- r_n(\mu))\\
        &\quad +\lambda_{n+1} \int_{\mu_m}^{\mu_M} \Big(\hat f(n+1,z) +  \frac{\big[R(n+1,z)\big]^+}{r_{n+1}(z)+\alpha-\theta}\Big) q(n,\mu,\ud z),
\end{aligned}
\end{align}
with $R(N,\mu)\coloneqq(\hat f(N,\mu)-\beta_N(\mu))(\theta-\alpha-r_{N}(\mu))$.
Given an open set $A\subset(0,\infty)$ we denote $\overline A$ its closure relative to $[0,\infty)$. A function $f$ belongs to $C^1(\overline A)$ if $f_x\in C(A)$ admits a continuous extension to $\overline A$. 
We are now in a position to state the main theorems.
\begin{theorem}[Value function]
\label{maintheorem1}
For $n\le N$ and $\mu\in[\mu_m,\mu_M]$, the mapping $x\mapsto V^N(x,n,\mu)$ is Lipschitz and {\em continuously differentiable} on $(0,\infty)$ (in particular, $V^N(\cdot,n,\mu)$ extends continuously to $[0,\infty)$). Moreover, $V^N$ is classical solution of \eqref{eq:ODE} and
$V^N(\cdot,n,\mu)\in C^2(\overline\cC_{n,\mu}\cap(0,\infty))$. 
\end{theorem}
It is worth noticing that the final statement in the above theorem follows by rewriting \eqref{eq:ODE} as
\begin{equation*}
\begin{aligned}
V^N_{xx}(x,n,\mu)&=\frac{2}{\sigma^2 x^2}\Big(r_n(\mu)V^N(x,n,\mu)-(\theta-\alpha)x V^N_x(x,n,\mu)\\
&\qquad\qquad-(\alpha+\nu\mu)x-\lambda_{n+1} \hat V^N(x,n+1;\mu)\Big),\quad x\in\overline\cC_{n,\mu}\cap(0,\infty),
\end{aligned}
\end{equation*}
and noticing that all terms on the right-hand side are continuous on $(0,\infty)$.

The structure of the optimal annuitization region depends on the behaviour of the function $M^N$. The next theorem is divided into five cases, which cover all possible shapes of the function $M^N$, because $x \mapsto M^N(x,n,\mu)$ is also shown to be convex in the proof of Proposition \ref{convexw} (cf.\ Figure \ref{fig1} for the case $K>0$). 
Notice that case 1.(i) of Theorem \ref{maintheorem} covers the shapes in (A), (C) or (D) of Figure \ref{fig1}. Case 1.(ii) covers instead the shape (B).
Since $V^N(\cdot,n,\mu)$ is extended to $[0,\infty)$, in the next theorem we also consider the extensions to $[0,\infty)$ of the sets $\cC_{n,\mu}$ and $\cS_{n,\mu}$. 

\begin{figure}[h!]
    \centering
    \begin{subfigure}[t]{0.22\textwidth}
        \centering
        \includegraphics[width=\textwidth]{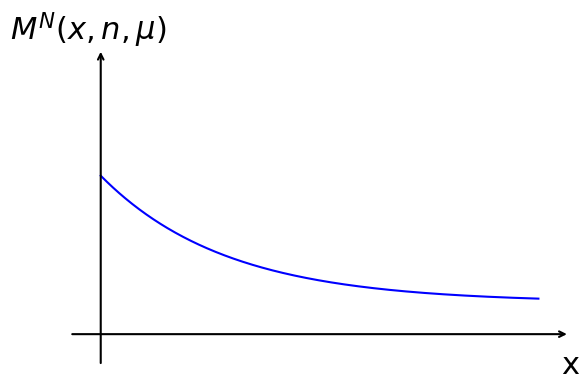}
        \caption{}
    \end{subfigure}
    \hfill
    \begin{subfigure}[t]{0.22\textwidth}
        \centering
        \includegraphics[width=\textwidth]{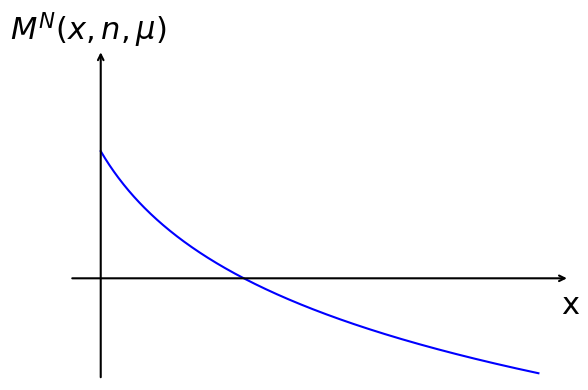}
          \caption{}
    \end{subfigure}
    \hfill
    \begin{subfigure}[t]{0.22\textwidth}
        \centering
        \includegraphics[width=\textwidth]{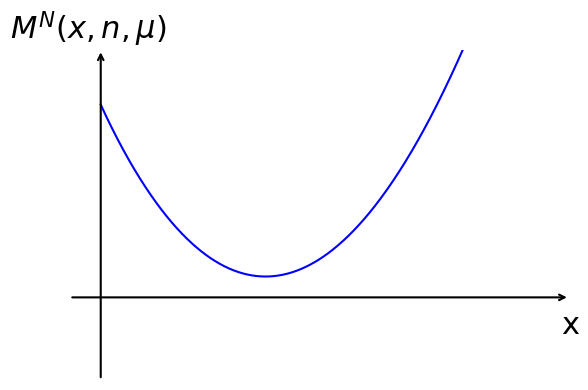}
          \caption{}
    \end{subfigure}
        \hfill
    \begin{subfigure}[t]{0.22\textwidth}
        \centering
        \includegraphics[width=\textwidth]{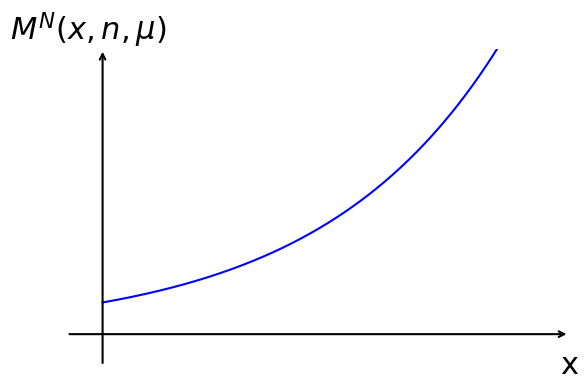}
          \caption{}
    \end{subfigure}
    \caption{Illustration of all possible behaviours of $M^N(\cdot, n, \mu)$ when $K > 0$. 
    (A) $M^N(x, n, \mu)$ converges to $L \in [0, M^N(0, n, \mu)]$ as $x \to \infty$, remaining positive for all $x$; 
    (B) $M^N(x, n, \mu)$ converges to $L \in [-\infty, 0)$ as $x \to \infty$; 
    (C) $M^N(x, n, \mu)$ has a positive minimum and diverges to $+\infty$ as $x \to \infty$; 
    (D) $M^N(x, n, \mu)$ remains positive for all $x$ and diverges to $+\infty$ as $x \to \infty$.}\label{fig1}
\end{figure}

\begin{theorem}[Optimal annuitization region]
\label{maintheorem}
The structure of  $\mathcal{C}_{n,\mu}$ and $\mathcal{S}_{n,\mu}$ is determined as follows:
\begin{itemize}
    \item[\textnormal{(1)}] Let $K > 0$. Then $M^N(0,n,\mu)>0$ and one of the two cases below holds:\smallskip
    
    \begin{itemize}
        \item[(i)] If $\lim_{x\to\infty}M^N(x,n,\mu)=L$ for some $L\in[0,\infty]$ then
       $M^N(x,n,\mu )>0$, $\forall x\in [0,\infty)$. The latter implies $\cC_{n,\mu}=[0,\infty)$ and $\cS_{n,\mu}=\varnothing$.\smallskip
        \item[(ii)] If $\lim_{x\to\infty}M^N(x,n,\mu)=-L$ for some $L\in(0,\infty]$, then there is $b=b(n,\mu)\in(0,\infty)$ such that $\cC_{n,\mu}=[0,b)$ and $\cS_{n,\mu}=[b,\infty)$.
    \end{itemize}
    \smallskip
    
    \item[\textnormal{(2)}] Let $K < 0$. Then $M^N(0,n,\mu)<0$ and one of the two cases below holds:\smallskip
    
    \begin{itemize}
        \item[(iii)] If $\lim_{x\to\infty}M^N(x,n,\mu)=+\infty$, then there is $b=b(n,\mu)\in(0,\infty)$ such that $\cC_{n,\mu}=(b,\infty)$ and $\cS_{n,\mu}=[0,b]$.\smallskip
        
        \item[(iv)] If $\lim_{x\to\infty}M^N(x,n,\mu)=-L$ for some $L\in(0,\infty]$, then $\cC_{n,\mu}=\varnothing$ and $\cS_{n,\mu}=[0,\infty)$.
    \end{itemize}
    
    \item[\textnormal{(3)}] Let $K=0$. Then $M^N(x,n,\mu)=R(n,\mu)x$. Moreover, $\{0\}\subset\cS_{n,\mu}$ whereas $R(n,\mu) \le 0\implies\cC_{n,\mu}=\varnothing$ and $R(n,\mu)>0\implies \cC_{n,\mu}=(0,\infty)$. 
\end{itemize}
\end{theorem}

We briefly comment on the economic interpretation of the above results. More detailed analysis and numerical experiments are presented in the next section. 
The \textit{running reward} $M^N$ represents the instantaneous gain from postponing annuitization (cf. \eqref{w}). When $M^N(X_t,n,\mu)>0$ it is not optimal to annuitize at time $t\ge 0$. In particular, when $K>0$ (i.e., there is an annuitization fee) it may occur that $M^N$ is everywhere positive (see (i)), which means that the annuity is never appealing to the investor. At the opposite end of the spectrum, we find the case with $K<0$ (i.e., there is an annuitization incentive), and it may occur that $M^N$ is negative everywhere. Then, the investor annuitizes immediately, irrespective of the performance of the financial fund (see (iv)). In case (ii), the policyholder does not annuitize when the fund's value is too low (i.e., $x<b$) because the annuity rate would be too low in presence of a fee $K>0$. However, for large values of the financial fund, it is convenient to annuitize. This case may occur if, for example, $\hat f(n,\mu)-\beta_n(\mu)$ is sufficiently large and positive (notice that $\theta-\alpha-r_n(\mu)<0$ as consequence of Assumption \ref{mainass}). The numerator of $\beta_n(\mu)$ accounts for the dividend payments and the bequest motives of the investor. It is then normalised by the effective discount rate $r_n(\mu)$ and the extra-return $\theta-\alpha$ on the fund. Then, large values of $\hat f(n,\mu)-\beta_n(\mu)$ can be interpreted as saying that the money's worth of the annuity is significantly more appealing than the financial returns and bequest motives adjusted by the discounting $r_n(\mu)+\alpha-\theta$. 
In cases (iii) and (iv), we observe the effect of annuitization incentives (i.e., $K<0$). In particular, when the financial performance is poor, the annuity becomes the policyholder's preferred option. Instead, when the financial performance is good, there may be an incentive to postpone annuitization if, for example, the money's worth is lower than the index $\beta_n(\mu)$ (i.e., $\hat f(n,\mu)<\beta_n(\mu)$). 

Before proceeding with the theoretical analysis, we present a numerical illustration to provide a concrete application of our results and to highlight the underlying economic mechanisms at play.

\section{Numerical analysis of annuitization decisions}
\label{numericalanalysis}
We perform a numerical study of the lifetime annuity in a particular case, with the aim of facilitating an economic interpretation of our theoretical results. We consider an individual who experiences a single health shock at a random time $\xi_0$ (in this context $\xi_0=\tau_1$). The individual starts with a baseline mortality force $\mu_0$. After the shock, either the mortality force remains unchanged, with probability $p=0.2$, or it increases to $2\mu_0$ with probability $(1-p)$.

We begin by calibrating the baseline mortality force $\mu_0$. According to the Human Mortality Database \cite{hmd2022}, a 60-year-old Italian male in 2010 had a remaining life expectancy of 22.41 years. We choose $\mu_0$ so that, if the mortality force were constant throughout the individual's remaining lifetime, then the expected lifetime $\tau^0_d$ would match this value. That is, we solve for $\mu_0$ the following equation
\[
\E[\tau^0_d] =\int_0^\infty \P(\tau^0_d>t)\ud t= \int_0^\infty \e^{-\mu_0 t} \ud t = 22.41.
\]
That yields $\mu_0 \approx 0.044623$.
We set the shock intensity to $\lambda_1=0.1$, which means that the health shock is expected to occur on average after 10 years. With one million simulations of the path of $(\mu_t)_{t\ge 0}$, we estimate that the individual's life expectancy in presence of a shock is approximately $\E[\tau_d]\approx16.2162$ years.
The objective mortality force $\hat \mu$ is chosen so that both the individual and the insurer agree on this life expectancy, implying $\hat \mu \approx 0.061667$.
For financial market parameters, we estimate $\theta$ and $\sigma$ using monthly S\&P 500 data from 1980 to 2025, obtaining $\theta \approx 0.087858$ and $\sigma \approx 0.1529518$. From the same period, we used 3-month T-bill data to estimate the subjective discount rate $\rho \approx 0.040400$. The insurer offers a competitive annuity interest rate $\hat{\rho}$, set to $1.5 \rho$ (i.e., $\hat{\rho} \approx 0.060600$), making the annuity investment financially appealing. We choose $\alpha$ to be 70\% of $\theta$, giving $\alpha \approx 0.061500$. Finally, we set $\nu=0.35$ and $K=1500$ USD to enhance data visualisation. 
Table \ref{tab:parameters} summarises the key parameters used in the numerical implementation of the model.
\begin{table}[H]
\centering
\begin{tabular}{cc}
\begin{tabular}{lc}
\toprule
\multicolumn{2}{c}{\textbf{Health Shock Parameters}} \\
\midrule
$\mu_0$            & $0.044623$ \\
$\lambda_1$    & $0.1$ \\
$p$              & $0.2$ \\
\midrule
\multicolumn{2}{c}{\textbf{Annuity Pricing Parameters}} \\
\midrule
$\hat{\mu}$      & $0.061667$ \\
$\hat{\rho}$     & $0.060600$ \\
$K$              & $1500$ \\
\bottomrule
\end{tabular} &

\begin{tabular}{lc}
\toprule
\multicolumn{2}{c}{\textbf{Financial Parameters}} \\
\midrule
$\theta$         & $0.087858$ \\
$\sigma$         & $0.152952$ \\
$\alpha$         & $0.061500$ \\
\midrule
\multicolumn{2}{c}{\textbf{Individual Preferences}} \\
\midrule
$\rho$           & $0.040400$ \\
$\nu$            & $0.35$ \\
\\
\bottomrule
\end{tabular} \\
\end{tabular}
\caption{Summary of Model Parameters}
\label{tab:parameters}
\end{table}

With the calibrated model, we now turn to determining the optimal annuitization strategy.
The function $M^1(x,0,\mu)$ (cf. \eqref{Def:M_njumps}), which is crucial in characterising the geometry of the continuation/stopping regions, depends on the post-shock value function $\hat V^1(x,1;\mu)$, which, using \eqref{Eq:hatV}, takes the form:
\[
\hat V^1(x,1;\mu_0) = p V^1(x,1,\mu_0) + (1 - p) V^1(x,1,2\mu_0).
\]
The function $V^1$ is the value function of the optimal annuitization problem with constant mortality force. This problem is solved in Section \ref{ConstantForceOfMortality}. The explicit expression for $V^1$ with $K>0$ can be obtained by solving \eqref{eq:ODE} with $\cC_{1,\mu}=(0,b(1,\mu))$, using $V^1(x,1,\mu)\le L(1+x)$ from Proposition \ref{vofinitinessl} and imposing continuous-fit and smooth-fit.
It is easy to verify that (cf.\ \cite[Ch.\ 3.4, Eqs.\ (3.146)--(3.148), (3.149)]{mythesis} for details)
\begin{align} \label{eq:valueex}
V^1(x,1,\mu) = 
\begin{cases}  
\begin{aligned}
&\ind_{\{x < b(1,\mu)\}} \big( \beta_1(\mu) x + \zeta_1^*(\mu) x^{\gamma_1^+(\mu)} \big) \\
&+ \ind_{\{x \ge b(1,\mu)\}} \hat f (1,\mu) (x - K)
\end{aligned}  
& \text{if } \hat{f}(1,\mu) > \beta_1(\mu), \\
\beta_1(\mu) x & \text{if } \hat{f}(1,\mu) \leq \beta_1(\mu),
\end{cases}
\end{align}
with 
\[
b(1,\mu) = \frac{\hat f (1,\mu) K \gamma_1^+(\mu)}{(\gamma_1^+(\mu) - 1)(\hat f (1,\mu) - \beta_1(\mu))},
\]
and
\[
\zeta_1^*(\mu) = \Big( \frac{\hat f (1,\mu) K}{\gamma_1^+(\mu) - 1} \Big)^{1 - \gamma_1^+(\mu)} \Big( \frac{\hat f (1,\mu) - \beta_1(\mu)}{\gamma_1^+(\mu)} \Big)^{\gamma_1^+(\mu)},
\]
where explicit formulae for $\hat f (1,\mu)$, $\beta_1(\mu)$ and $\gamma_1^+(\mu)$ are given by \eqref{delta1}, \eqref{betanojump} and \eqref{gammarpm}, respectively. 
After the health shock, the individual's stopping region depends on the realised post-shock mortality. If the mortality remains at the baseline level $\mu_0$, the stopping region is given by $\cS_{1,\mu_0} = [b(1,\mu_0), \infty)$, where $b(1,\mu_0) = 32772.84$ USD. Alternatively, if the mortality is $2\mu_0$, the stopping region becomes $\cS_{1,2\mu_0} = [b(1,2\mu_0), \infty)$, with $b(1,2\mu_0) = 49028.47$ USD.

We next study the optimal annuitization region {\em before} the change in mortality force. It is straightforward to verify numerically that  $ \lim_{x\to\infty}M^1(x,0,\mu_0 )=-\infty$. 
Therefore, by Theorem \ref{maintheorem}-1.(ii), the optimal stopping region is of the form $[b,\infty)$, for some $b>0$. To find explicitly $b$ and the value function, we rely on the geometric approach to optimal stopping problems for one-dimensional diffusions illustrated by Dayanik and Karatzas in \cite{dayanik2003optimal} (cf.\ also It\^o and McKean \cite{itodiffusion} for earlier work). 
First, we introduce the so-called \textit{Mayer reformulation} of the problem. Using the Markovian structure of $X$, it is easy to see that 
\begin{align}
W^1(x,0,\mu_0)=w_1(x,0,\mu_0)+\sup_{\tau\in\cT(\bF^B)} \E[-\e^{-r_0(\mu_0)\tau}w_1(X_\tau^x,0,\mu_0)],
\end{align}
with
\[
w_1(x,0,\mu_0)=\E_x\Big[\int_0^\infty\e^{-r_0(\mu_0)t}M^1(X_t,0,\mu_0)\ud t\Big].
\]
We compute the function $w_1(x,0,\mu_0)$ using the log-normal distribution of $X$. 
Recalling the fundamental increasing and decreasing solutions, $\psi$ and $\phi$ of $(\cL-r_0(\mu_0))u=0$ on $(0,\infty)$, we define the transformation
\begin{align}\label{F}
y=F(x,0,\mu_0)=\frac{\psi(x,0,\mu_0)}{\phi(x,0,\mu_0)}= x^{\gamma_0^+(\mu_0)-\gamma_0^-(\mu_0)},
\end{align}
with inverse $x=F^{-1}(y,0,\mu_0)=y^{\frac{1}{\gamma_0^+(\mu_0)-\gamma_0^-(\mu_0)}}$. We
set
\begin{equation} \hat{w}_1(y,0,\mu_0)\coloneqq 
\begin{cases}0 &\textrm{if} \ y=0, \\
\big(\frac{w_1}{\phi} \ \circ F^{-1} \big) (y,0,\mu_0)&\textrm{if} \ y>0.
\end{cases}
\label{cinv1}
\end{equation}
The problem of calculating the value function $V^1(\cdot,0,\mu_0)=W^1(\cdot,0,\mu_0)+\hat f(0,\mu_0)(\cdot-K)$ reduces to finding the smallest non-negative concave function $U(\cdot,0,\mu_0)$ that dominates $\hat{w}_1(\cdot,0,\mu_0)$. As shown in Figure~\ref{fig:Qandhatw}, the function $\hat{w}_1$ is initially convex and negative, then it becomes concave (and increasing). The function $U$ is linear and increasing between $y=0$ and $y_*\approx 5\times 10^{34}$, which is the tangency point with $\hat w_1$. For $y\ge y_*$ the function $U$ coincides with $\hat w_1$. 
Transforming back into the original coordinates we find 
\[
b_* \coloneqq b(0,\mu_0) =F^{-1}(y_*,0,\mu_0).
\]
For the parameter values from Table \ref{tab:parameters} we have $b_*\approx 20383.66$.
Then 
\[
W^1(x,0,\mu_0) = w_1(x,0,\mu_0) + \phi(x,0,\mu_0)\, U(F(x,0,\mu_0),0,\mu_0),
\]
and $V^1(x,0,\mu_0)=W^1(x,0,\mu_0)+\hat f(0,\mu_0)(x-K)$.
Figure~\ref{fig:W} shows $x \mapsto V^1(x,0,\mu_0)$, with the optimal threshold $b_*$ indicated by a grey point.

\begin{figure}[h!]
    \centering
    \begin{subfigure}[t]{0.49\textwidth}
        \centering
        \includegraphics[width=0.79\textwidth]{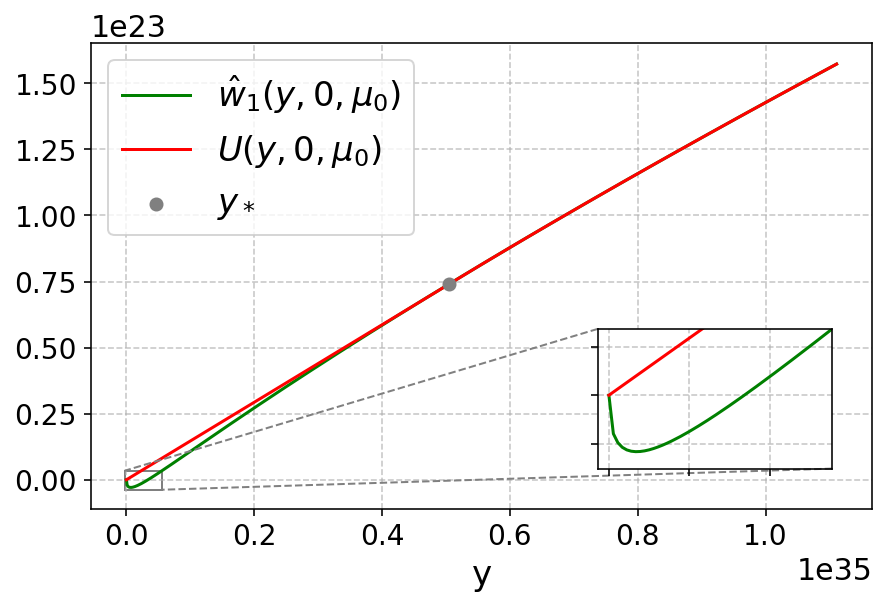}
        \caption{}
        \label{fig:Qandhatw}
    \end{subfigure}
    \hfill
    \begin{subfigure}[t]{0.49\textwidth}
        \centering
        \includegraphics[width=0.81\textwidth]{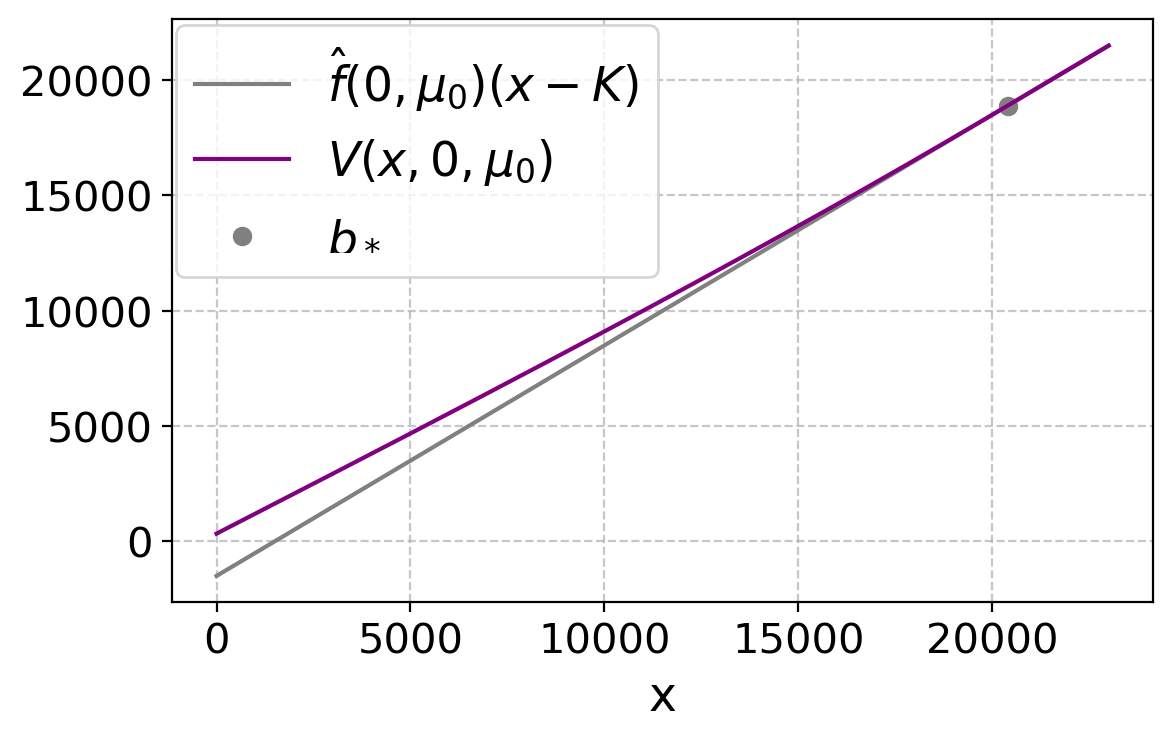}
        \caption{}
        \label{fig:W}
    \end{subfigure}
    \caption{(A) The smallest non-negative concave majorant $U(\cdot,0,\mu_0)$ of $\hat{w}_1(\cdot,0,\mu_0)$. (B) The value function $x \mapsto V^1(x,0,\mu_0)$, with the optimal boundary $b_* \approx 20383.66$ marked by a grey point.}
\end{figure}

When the wealth is below the critical threshold $b_*=20383.66$ USD, the individual prefers to continue investing in the financial market rather than annuitizing. In this region, annuitization is relatively unattractive: the fixed acquisition fee $K=1500$ USD represents a substantial fraction of the individual's wealth, leading to low annuity payments and a loss of investment flexibility. Moreover, the opportunity to earn dividends and the potential for a bequest in the event of premature death provide additional incentives to delay annuitization. However, as the wealth increases, the relative burden of the fee $K$ decreases, and the guaranteed income from the annuity becomes increasingly attractive compared to uncertain investment returns. It becomes optimal for the individual to annuitize once their wealth exceeds the threshold $b_*$. When the wealth is equal to $b_*$, annuitization yields a constant annual income of approximately 2308.84 USD for life.
 
\begin{figure}
    \centering
    \includegraphics[width=0.8\linewidth]{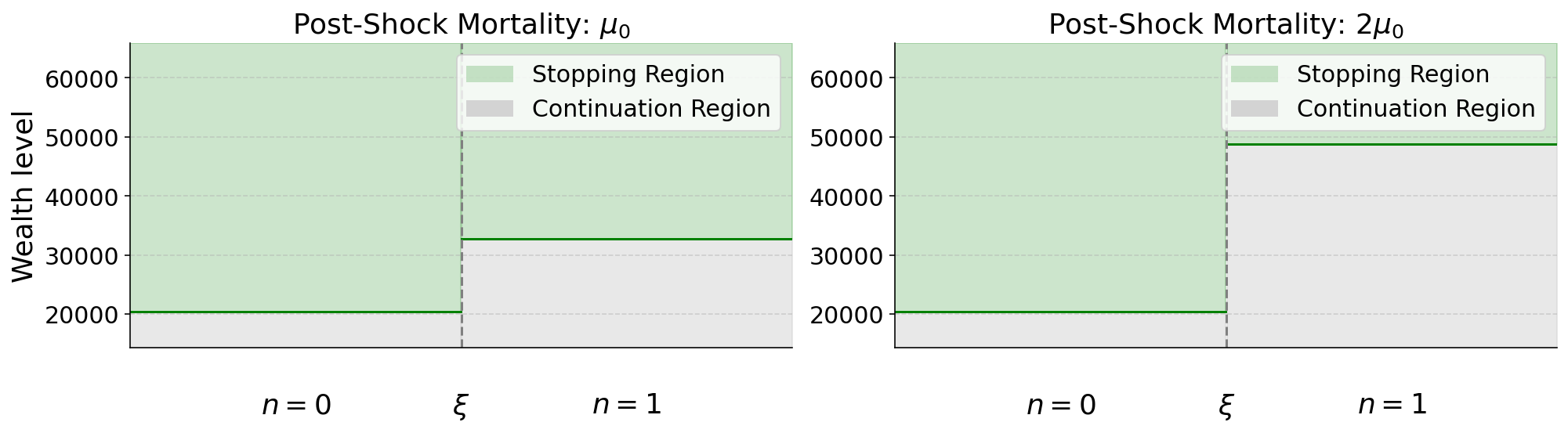}
    \caption{Optimal annuitization threshold before and after a health shock occurring at time $\xi_0$. State $n=0$ corresponds to the pre-shock phase, and $n=1$ to the post-shock phase. The two panels compare scenarios in which the post-shock mortality remains unchanged (i.e., $\bar \mu_1(\omega)=\mu_0$) and the one in which it increases to $\bar \mu_1(\omega)=2\mu_0$.}
    \label{fig:comparison thresholds}
\end{figure}

Figure \ref{fig:comparison thresholds} illustrates how the optimal annuitization threshold adjusts in response to a health shock. On the horizontal axis, we represent the pre-shock period ($n=0$) and the post-shock period ($n=1$), with the vertical dashed line marking the time of the shock $\xi_0$. On the vertical axis, we represent the individual’s wealth level.
\begin{itemize}
    \item Left panel: If the health shock does not affect mortality (i.e., the mortality force remains at the baseline $\mu_0$), the optimal annuitization threshold increases from 20383.66 USD before the shock to 32772.84 USD afterwards. 
The increase is somewhat surprising from a purely demographic perspective, because conditionally on being alive at time $\xi$ of the shock, the individual's life expectancy increases compared to that at time zero (before the shock) as no further shock is expected ($\lambda_2=0$). This should incentivise early annuitization as it increases the money's worth ($\hat f(1,\mu_0)=1.4380$ and $\hat f(0,\mu_0)=1.2239$). However, various competing effects are at play. A decrease of the effective discount rate ($r_0(\mu_0)>r_1(\mu_0)$) reduces the urgency to annuitize and it also contributes to an increasing appeal of the financial market ($\beta_1(\mu_0)>\beta_0(\mu_0)$). Therefore it is hard to make a definite statement on what leads to delaying the annuitization, but it may appear that in this example the increase in the appeal of the financial market and the reduced discount rate outweigh the improved demographic outlook.
    
    \item Right panel: If the shock leads to a doubling of the mortality rate (i.e., $\bar \mu_1(\omega) = 2\mu_0$), the threshold increases from 20383.66 USD to 49028.47 USD. In this case, the individual faces reduced life expectancy and therefore places less value on future annuity payments. As a result, they demand even higher wealth before annuitizing.
\end{itemize}

\begin{figure}[h!]
    \centering
    \begin{subfigure}[t]{0.32\textwidth}
        \centering
        \includegraphics[width=\textwidth]{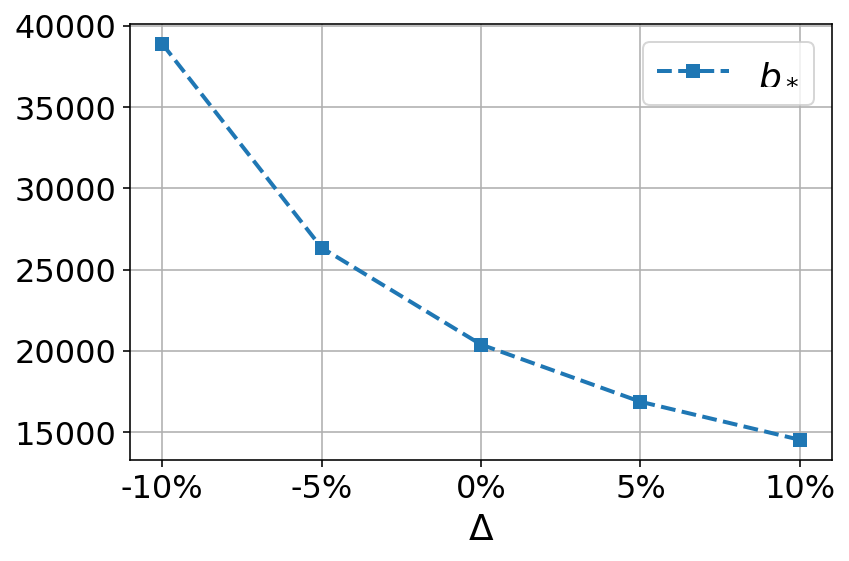}
        \caption{}
        \label{fig:b_vs_hatmu}
    \end{subfigure}
    \hfill
    \begin{subfigure}[t]{0.32\textwidth}
        \centering
        \includegraphics[width=\textwidth]{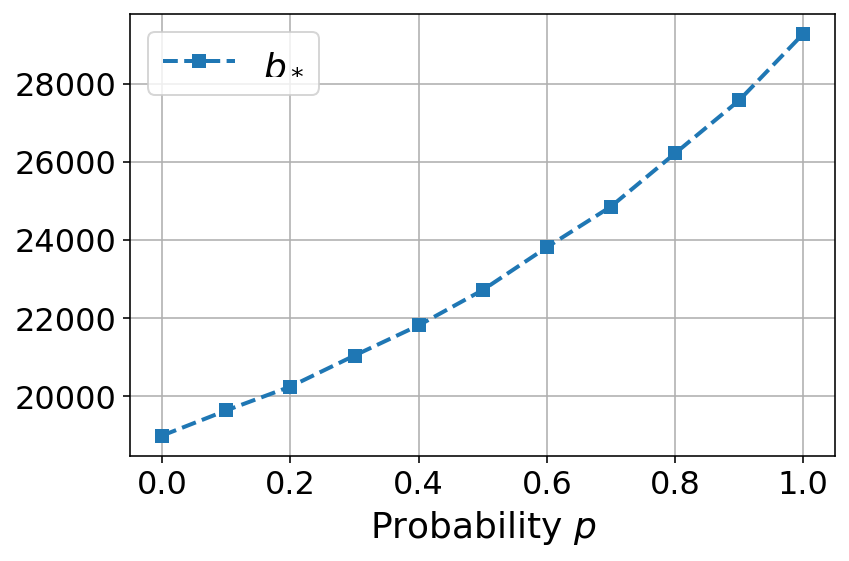}
        \caption{}
        \label{fig:b_vs_pfhat}
    \end{subfigure}
    \hfill
    \begin{subfigure}[t]{0.32\textwidth}
        \centering
        \includegraphics[width=\textwidth]{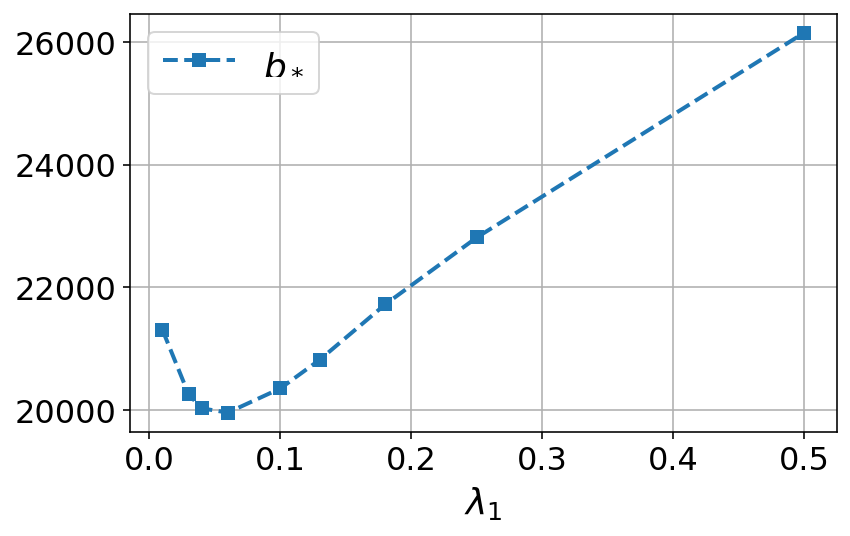}
        \caption{}
        \label{fig:b_vs_lambda}
    \end{subfigure}

    \vspace{0.3cm}

    \begin{subfigure}[t]{0.32\textwidth}
        \centering
        \includegraphics[width=\textwidth]{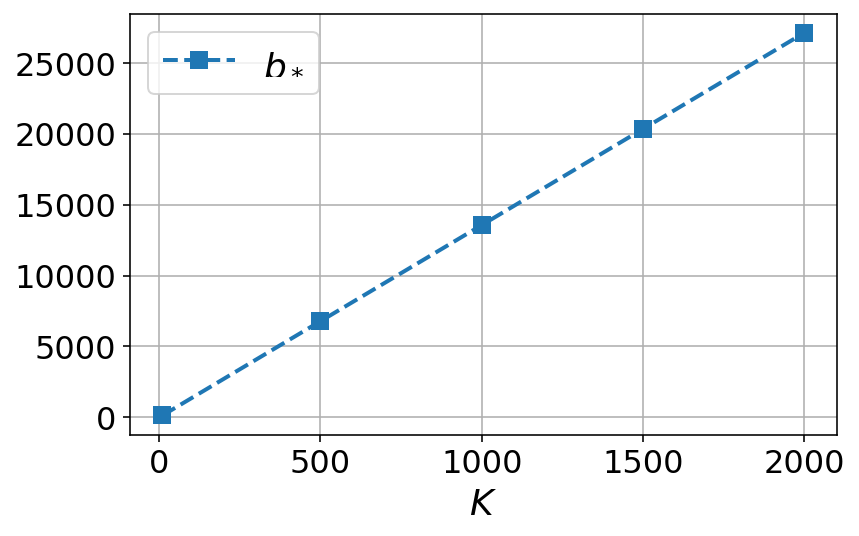}
        \caption{}
        \label{fig:b_vs_K}
    \end{subfigure}
    \hspace{5mm}
    \begin{subfigure}[t]{0.32\textwidth}
        \centering
        \includegraphics[width=\textwidth]{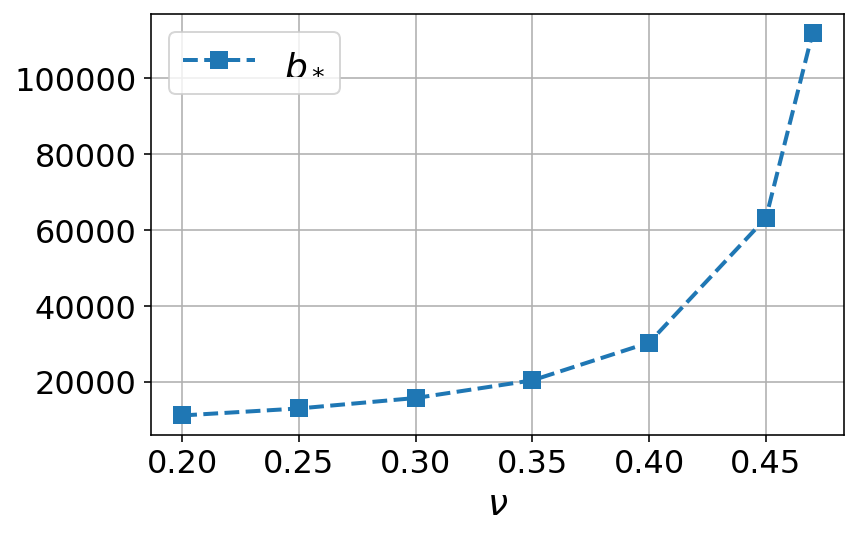}
        \caption{}
        \label{fig:b_vs_nu}
    \end{subfigure}

    \caption{Sensitivity of the optimal annuitization threshold $b_*$ to model parameters: 
    (A) percentage change ($\Delta$) from $\hat{\mu}$; 
    (B) probability $p$; 
    (C) shock intensity $\lambda_1$; 
    (D) acquisition cost $K$;
    (E) bequest motive $\nu$.}
    \label{fig:sensitivity_b}
\end{figure}

We now examine how the optimal annuitization threshold $b_*$ responds to changes in key model parameters.
Since the stopping region is $[b_*,\infty)$, an increase in $b_*$ implies a smaller stopping region, meaning the individual becomes less likely to annuitize. Figure \ref{fig:b_vs_hatmu} illustrates the relationship between $b_*$ and the objective mortality force used by the insurance company to price the annuity. Starting from the baseline (objective mortality) $\hat\mu$ 
which assigns $\hat f(0,\mu_0)\approx 1.2239$, we consider perturbations of the form $\hat{\mu}(1+\Delta)$, where $\Delta$ represents the percentage change from $\hat{\mu}$.
For $\Delta>0$ the insurer assigns a higher probability of death to the individual. 
The money’s worth of the annuity increases, making annuitization more attractive. As a result, the optimal threshold $b_*$ decreases compared to the baseline case with $\hat \mu$. Vice versa, for $\Delta<0$ the insurer assigns a longer individual's lifetime compared to the baseline case. As a result, the money’s worth of the annuity decreases, and the annuitization becomes less appealing.

Figures~\ref{fig:b_vs_pfhat} and \ref{fig:b_vs_lambda} show the relationship between the optimal annuitization threshold $b_*$ and, respectively, the probability $p$ and the shock intensity $\lambda_1$. A higher value of $p$ implies that the individual is more likely to retain the baseline mortality rate $\mu_0$ rather than jumping to $2\mu_0$ after a health shock. The larger $\lambda_1$ is, the earlier the shock is expected to occur --- recall $\E[\tau_1]=1/\lambda_1$. In both cases, the annuitization decision is affected in two opposing ways. On the one hand, a higher $p$ or a lower $\lambda_1$ increases the subjective life expectancy, which reduces the urgency to annuitize. On the other hand, since $\hat\mu$ is fixed, the money’s worth of the annuity increases, encouraging earlier annuitization.
Figure~\ref{fig:b_vs_pfhat} shows that the threshold $b_*$ increases with $p$, indicating that the incentive to delay annuitization is the dominant effect. 
Figure~\ref{fig:b_vs_lambda} shows a U-shaped form of $b_*$ as a function of $\lambda_1$: the line decreases for very low values of $\lambda_1$ and then it increases. 
For small values of $\lambda_1$, even a slight increase in $\lambda_1$ significantly increases the probability of occurrence of a health shock, which in turn lowers life expectancy. Since the individual still has sufficient time to benefit from annuity payments, the urgency to annuitize becomes the dominant effect. Consequently, we observe a decrease in the annuitization threshold $b_*$ as $\lambda_1$ increases but it is small.
However, for sufficiently high values of $\lambda_1$, life expectancy may be substantially reduced, and annuity payments are received over a shorter period. As a result, annuitization becomes less attractive to the individual and $b_*$ increases.
Figure~\ref{fig:b_vs_K} illustrates the effect of the acquisition fee $K$ on the optimal threshold when $K>0$. A higher value of $K$ reduces the annuity payments, making the annuity less attractive relative to continued investment in the financial market. Consequently, the individual requires a higher level of wealth before annuitizing, and $b_*$ increases with $K$. Conversely, when $K$ decreases, annuitization becomes more appealing and the threshold declines. For $K=0$ a direct numerical evaluation of \eqref{eq:Lnmu} yields
${ R}(0,\mu_0) \approx -0.018975<0$.
Therefore, by Theorem \ref{maintheorem}-(3), immediate annuitization is optimal for every initial wealth level.

Next, we study the sensitivity of the optimal annuitization boundary with respect to the bequest motive. We take parameter values as in Table \ref{tab:parameters} and vary $\nu$ in the range $[0.2,0.475]$. Results are illustrated in Figure \ref{fig:b_vs_nu}. As $\nu$ increases, the individual becomes more concerned about leaving a bequest in case of death. This strengthens the incentive to retain liquid wealth rather than converting it into an irreversible annuity. As a result, the annuitization threshold $b_*$ increases with $\nu$, indicating that stronger bequest motives lead to a delay in annuitization.

We now show that both cases 1.(i) and 1.(ii) of Theorem \ref{maintheorem} may occur in practice. For fixed $K=1500$ as in Table \ref{tab:parameters}, we now consider the case in which the insurer prices the annuity using the objective mortality force $\hat\mu = 1.25\mu_0$. Under this specification, the post-shock annuitization regions are different depending on the realised mortality state. If the mortality force remains equal to its baseline level \(\mu_0\), we obtain $\mathcal S_{1,\mu_0}=[71024.32,\infty)$, whereas if the mortality force jumps to \(2\mu_0\), then $\mathcal S_{1,2\mu_0}=\varnothing$.
Therefore, after the shock, annuitization is optimal only in the favourable mortality state and only when the fund value is sufficiently high. Before the mortality shock occurs, a direct numerical verification shows that $M^1(x,0,\mu_0) > 0$ for any $x\geq 0$. Hence, by Theorem \ref{maintheorem} case 1.(i), $\mathcal S_{0,\mu_0}=\varnothing$ and it is never optimal to annuitize.

\begin{figure}[h!]
    \centering
    \begin{subfigure}[t]{0.48\textwidth}
        \centering
        \includegraphics[width=0.79\linewidth]{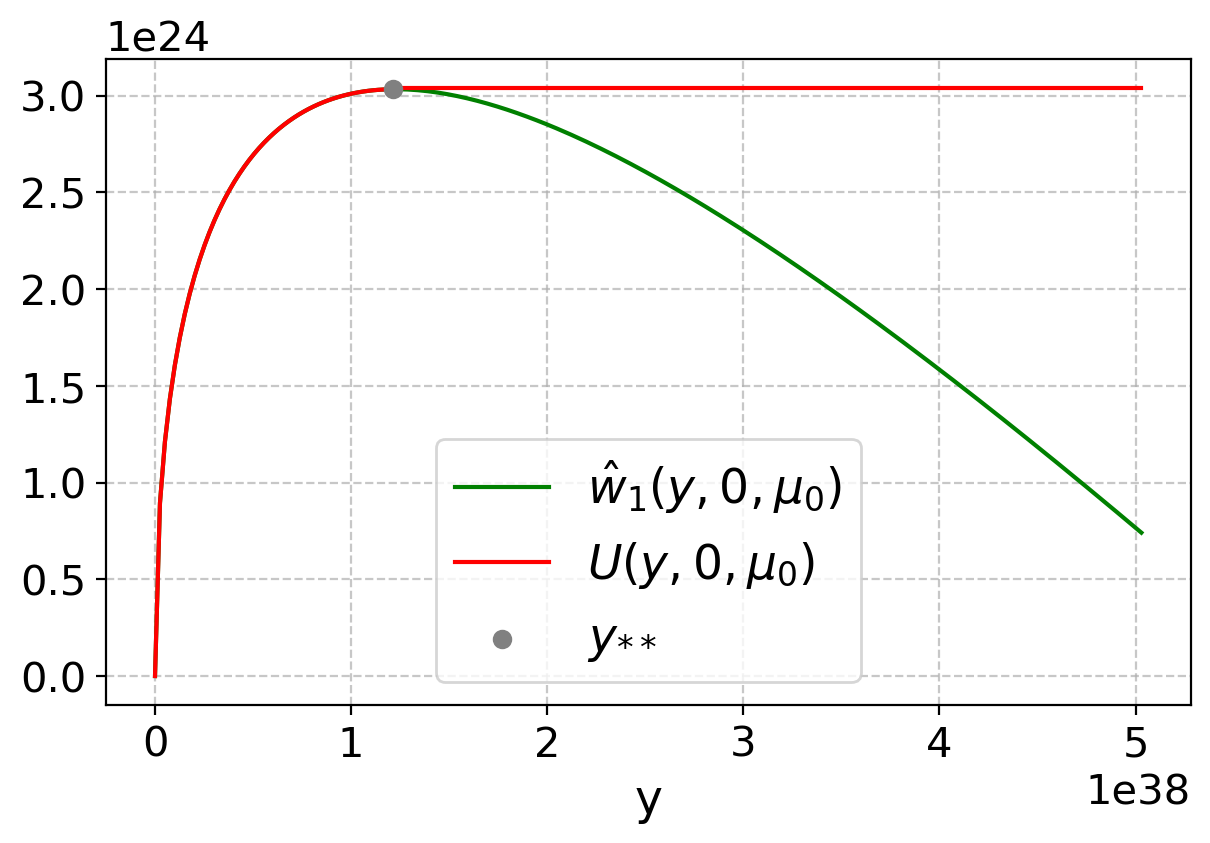}
        \caption{}
        \label{fig:UhatwK}
    \end{subfigure}
    \hfill
    \begin{subfigure}[t]{0.48\textwidth}
        \centering
        \includegraphics[width=0.81\linewidth]{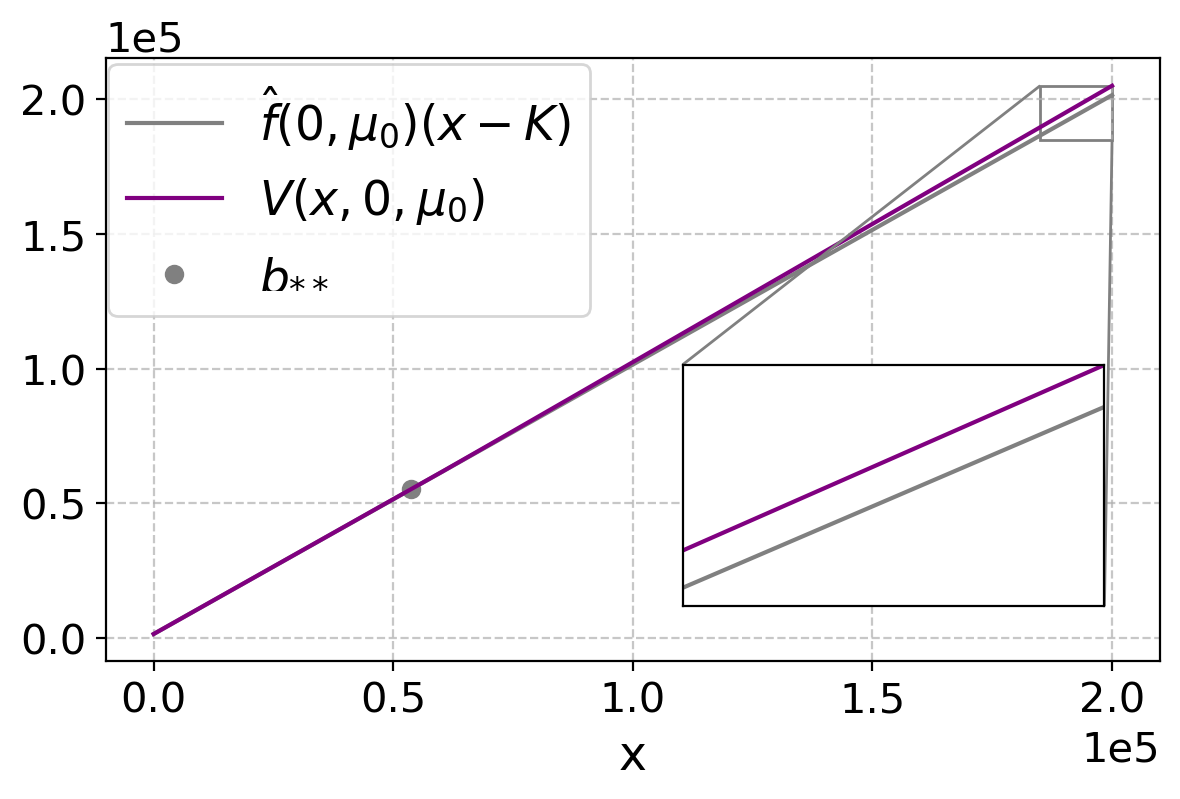}
        \caption{}
        \label{fig:value_func}
    \end{subfigure}
    
    \vspace{0.5cm}
    
    \begin{subfigure}[t]{\textwidth}
        \centering
        \includegraphics[width=0.8\linewidth]{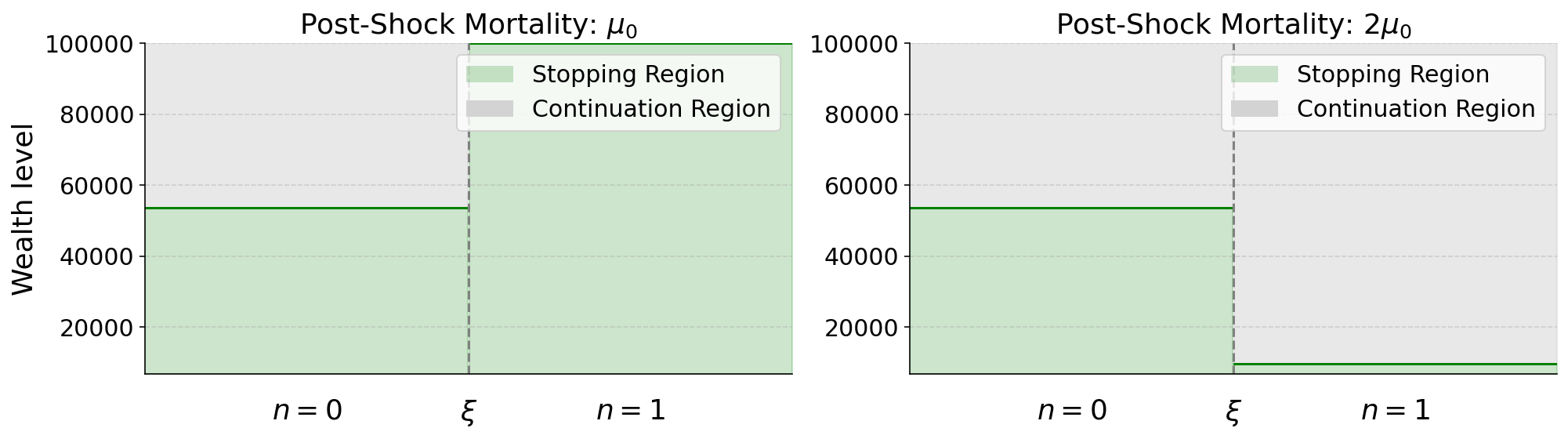}
        \caption{}
        \label{fig:thresholds}
    \end{subfigure}
    
    \caption{
        (A) The smallest non-negative concave majorant $U(\cdot,0,\mu_0)$ of $\hat{w}_1(\cdot,0,\mu_0)$. 
        (B) The value function $x \mapsto V^1(x,0,\mu_0)$, with the optimal boundary $b_{**} \approx 53639.08$ marked by a grey point. 
        (C) Optimal annuitization threshold before and after a health shock occurring at time $\xi_0$. State $n=0$ corresponds to the pre-shock phase, and $n=1$ to the post-shock phase. The two panels compare scenarios where post-shock mortality remains unchanged ($\bar{\mu}_1(\omega)=\mu_0$) versus when it increases to $\bar{\mu}_1(\omega)=2\mu_0$.
    }
    \label{fig:combined_figures}
\end{figure}

We now briefly discuss the scenario $K < 0$. All parameters remain unchanged from the case $K > 0$, except for $\nu = 0.6$ and $K = -1500$ USD, which we adjust to enhance data visualization. The explicit expression for $V^1(x,1,\mu)$ when $K < 0$ is given by:
\begin{align*}  
V^1(x,1,\mu) = 
\begin{cases}  
\begin{aligned}
&\ind_{\{x \leq b(1,\mu)\}} \hat{f}(1,\mu)(x - K) \\
&+ \ind_{\{x > b(1,\mu)\}} \big(\beta_1(\mu) x + \zeta_1^{**}(\mu) x^{\gamma_1^-(\mu)}\big)
\end{aligned}  
& \text{if } \hat{f}(1,\mu) < \beta_1(\mu), \\
\hat{f}(1,\mu)(x - K) & \text{if } \hat{f}(1,\mu) \geq \beta_1(\mu),
\end{cases}
\end{align*}
where  
\begin{align*}
b(1,\mu) = \frac{\hat{f}(1,\mu) K \gamma_1^-(\mu)}{(\gamma_1^-(\mu)-1)(\hat{f}(1,\mu) - \beta_1(\mu))}, 
\end{align*}
and \begin{align*}
\zeta_1^{**}(\mu) = \Big(\frac{\hat{f}(1,\mu) K}{\gamma_1^-(\mu)-1}\Big)^{1-\gamma_1^-(\mu)} \Big(\frac{\hat{f}(1,\mu) - \beta_1(\mu)}{\gamma_1^-(\mu)}\Big)^{\gamma_1^-(\mu)},
\end{align*}
(cf.\ \cite[Ch.\ 3.4, Eqs.\ (3.150),\ (3.155)--(3.158)]{mythesis} for details). Notice that, naturally, the boundary $b(1,\mu)$ is not the same as in the case $K>0$ because we are considering here $K<0$.

It is easy to numerically verify that $\lim_{x\to\infty}M^N(x,0,\mu_0 )=+\infty$.
By Theorem \ref{maintheorem}, the optimal stopping region takes the form $[0, b]$ for some $b > 0$.  
Applying the geometric method (as in the $K > 0$ case), we derive the smallest non-negative concave majorant $U(\cdot, 0, \mu_0)$ of $\hat{w}_1(\cdot, 0, \mu_0)$, illustrated in Figure~\ref{fig:UhatwK}. The corresponding value function is plotted in Figure~\ref{fig:value_func}, with the optimal annuitization threshold $b_{**} \coloneqq b(0,\mu_0) \approx 53639.08$ USD.

Figure \ref{fig:thresholds} illustrates how the optimal annuitization threshold adjusts in response to a health shock. The horizontal axis shows the pre-shock ($n=0$) and post-shock ($n=1$) periods, with a vertical dashed line at $\xi_0$. On the vertical axis, we represent the individual’s wealth level. In the pre-shock period ($n=0$), the individual purchases the annuity when wealth lies in $[0, 53639.08]$. In $n=1$, the decision depends on the realised post-shock mortality. If mortality remains at the baseline level $\mu_0$, the stopping region is $\cS_{1,\mu_0} = [0, \infty)$, indicating immediate annuitization is optimal for any wealth level. If mortality increases to $2\mu_0$, the stopping region becomes $\cS_{1,2\mu_0} = [0, b(1,2\mu_0)]$, where $b(1,2\mu_0) = 9673.02$ USD.
This reveals a drastic shift in annuitization behaviour: a change in mortality state can lead to radically different strategies.

\section{Optimal Annuitization with Constant Mortality Force}
\label{ConstantForceOfMortality}
In this section, we begin the theoretical analysis by studying properties of $V^N(x, N, \mu)$. Later, in Section \ref{OptimalAnnuitizationwithJumps}, we extend those results to  $V^N(x,n, \mu)$ with $n<N$, using induction.
For $n=N$ the mortality force remains constant at $\mu\in[\mu_m,\mu_M]$. The subjective survival probability simplifies to $_z p_{\eta+t}= \e^{-\mu z}$ which leads to the following expression derived from \eqref{fhatcmu},
\begin{align}\label{delta1}
\begin{aligned}
\hat{f}(N,\mu)= \frac{\hat{\rho}+\hat{\mu}}{r_N(\mu)},\quad\text{with}\ r_N(\mu)=\rho+\mu.
\end{aligned}
\end{align} 

First we show that Assumption \ref{mainass} is necessary and sufficient for the well-posedness of the value function $V^N(x,N,\mu)$ for all $(x,\mu)\in(0,\infty)\times[\mu_m,\mu_M]$. 
\begin{proposition} \label{vofinitinessl} 
Set $\hat \theta=\theta-\alpha-\rho$. The following statements hold:
\begin{itemize}
\item[(i)] If $\hat \theta -\mu_m < 0$, there is $L=L(\mu_m)>0$, independent of $(x,N,\mu)$, such that $0\le V^N(x,N,\mu)\le L(1+x)$ for all $x\in(0,\infty)$.
\item[(ii)] If $\hat \theta-\mu_m\ge 0$, we have 
$V^N(x,N,\mu)=+\infty$,
for $x>0$ and 
$\mu\in[\mu_m,\hat \theta\wedge\mu_M]$.
\end{itemize}
\end{proposition}
\begin{proof}
Let us start by proving $(i)$. It is clear that taking $\tau=T$ for any deterministic $T>0$, yields the lower bound
$V^N(x,N,\mu)\ge -\e^{-(\rho+\mu)T}\hat f(N,\mu) K$.
Then, letting $T\to\infty$ we obtain the lower bound $V^N(x,N,\mu)\ge 0$. For the upper bound, notice that $t\mapsto \e^{-(\rho+\mu) t}X_t$ is a \textit{non-negative} supermartingale and, therefore,  
\begin{align}
\begin{aligned}
\label{boundsupermg} \sup_\tau \E_x \big[ \e^{-(\rho+\mu) \tau} X_\tau \big]\leq x 
\end{aligned}
\end{align}
(cf.\ \cite[Ch.\ 1.3, Problem 3.16, p. 18]{karatzas2014brownian} and \cite[Ch.\ 1.3, Thm. 3.22, p. 19]{karatzas2014brownian}). From \eqref{VF:recursive} and \eqref{boundsupermg}, we find
\begin{align} \label{upperboundvaluefunct0}
    \begin{aligned} 
        V^N(x,N,\mu)
        &\leq \sup_{\tau\in\cT(\bF^B)} \E_x \Big[ \int_{0}^{\tau} \e^{-(\rho+\mu)t} ( \alpha+\nu \mu ) X_t \ud t - \e^{-(\rho+\mu)\tau} \hat f(N,\mu) K \Big]+ \hat f(N,\mu) x\\
        &=\sup_{\tau\in\cT(\bF^B)} \E_x \Big[ \int_{0}^{\tau} \e^{-(\rho+\mu)t} \big[( \alpha+\nu \mu ) X_t+(\rho+\mu)\hat f(N,\mu) K\big] \ud t\Big] + \hat f(N,\mu) (x-K),
    \end{aligned}
\end{align}
where the final expression is obtained by integration by parts. 
Using the explicit form of $X_t$ and continuing with the inequalities, we obtain
\begin{align*}
V^N(x,N,\mu)&\le \hat f(N,\mu)(x+|K|)+\int_0^\infty\e^{(\hat \theta-\mu)t}(\alpha+\nu\mu)x\E\Big[\e^{\sigma B_t-\frac12\sigma^2t}\Big]\ud t\\
&\quad+\int_0^\infty\e^{-(\rho+\mu)t}(\rho+\mu)\hat f(N,\mu) |K|\ud t= \hat f(N,\mu)(x+2|K|)+\beta_N(\mu) x,
\end{align*}
where in the final line we used 
$\hat \theta-\mu\le \hat \theta-\mu_m<0$.
Then $(i)$ holds because $\mu\in[\mu_m,\mu_M]$.

For the proof of $(ii)$, it is enough to observe that, for any deterministic time $T>0$, we have
\begin{align*}
V^N(x,N,\mu)&\ge \int_{0}^{T}  \e^{(\hat \theta-\mu)t} ( \alpha+\nu \mu ) x\E\big[\e^{\sigma B_t-\frac12\sigma^2 t}\big] \ud t\\
&\quad + \e^{(\hat \theta-\mu) T} \hat f(N,\mu) x\E\big[\e^{\sigma B_T-\frac12\sigma^2 T}\big]- \e^{-(\rho+\mu)T} \hat f(N,\mu) K.
\end{align*}
For $\mu_M\ge \hat \theta$ and $\mu=\hat \theta$, the expression above reads $(\alpha+\nu \hat \theta)x T+\hat f(N,\hat \theta)(x-K)$, which diverges to $+\infty$ as $T\to\infty$. Instead, for any $\mu\in[\mu_m,\hat \theta\wedge\mu_M)$, we have
\begin{align*}
V^N(x,N,\mu)\ge |\beta_N(\mu)| x\big(\e^{(\hat \theta-\mu)T}-1\big)+\hat f(N,\mu) x\e^{(\hat \theta-\mu)T}- \e^{-(\rho+\mu)T} \hat f(N,\mu) K.
\end{align*}
Then, for $x>0$, the right-hand side diverges to $+\infty$ when $T\to\infty$.
\end{proof}

Another consequence of Assumption \ref{mainass} is given in the following proposition.

\begin{proposition} \label{supgsfinite}
Under Assumption \ref{mainass}, there is a constant $C>0$ such that
\begin{align} \begin{aligned}   \E_x \Big[ \int_{0}^{\infty}  \e^{-r_N(\mu)t}X_t \ud t+ \sup_{s \geq 0}  \e^{-r_N(\mu)s} X_s\Big]\le C x. \end{aligned}\end{align} 
\end{proposition}

\begin{proof}
By Tonelli's theorem, we have
\begin{align*}
&\E_x \Big[ \int_{0}^{\infty}  \e^{-r_N(\mu)t}X_t \ud t\Big]=\int_{0}^{\infty}  \e^{-r_N(\mu)t}\E_x \big[ X_t\big] \ud t=x\int_{0}^{\infty}  \e^{(\theta-\alpha-\rho-\mu)t} \ud t.
\end{align*}
Letting $\bar{\theta}\coloneqq \alpha+\rho+\mu+\frac{\sigma^2}{2}-\theta>0$, we have 
\begin{align}
    \begin{aligned} 
    \E_x \Big[\sup_{s \geq 0}  \e^{-(\rho+\mu)s} X_s\Big] 
    &= x \ \E \Big[ \exp \Big\{ \sup_{s \geq 0} ( -\bar{\theta} s +\sigma B_s) \Big\}\Big]=x\frac{2 \bar{\theta} }{\sigma^2} \int_0^\infty \e^{u-\frac{2\bar \theta}{\sigma^2}u} \ud u<\infty,
    \end{aligned}
\end{align}
where we used 
$\P\big(\sup_{s \geq 0} ( -\bar{\theta} s +\sigma B_s) > u \big)= \exp(-\frac{2\bar{\theta}u}{\sigma^2})$, by \cite[Ch.\ VIII.2, Cor.\ 1, p. 760]{shiryaev1999essentials},
and $\sigma^2-2\bar \theta<0$ under Assumption \ref{mainass}. Combining the two estimates above, the proof is complete.
\end{proof}

By Dynkin's formula applied to $\e^{-r_N(\mu)(\tau\wedge T)}\hat f(N,\mu)(X_{\tau\wedge T}-K)$, for any $T>0$ and any stopping time $\tau$ we obtain
\begin{align}
\begin{aligned}
\label{swerbounded} 
     &\E_x \Big[ \e^{-r_N(\mu)(\tau\wedge T)}\hat f(N,\mu) (X_{\tau\wedge T} - K) \Big]=\hat f(N,\mu) (x - K)\\
     &\quad +\E_x \Big[ \int_0^{\tau\wedge T} \e^{-r_N(\mu)t} \Big((\theta-\alpha-r_N(\mu))\hat f(N,\mu) X_t +(\hat{\rho}+\hat{\mu}) K\Big) \ud t\Big],
\end{aligned}
\end{align} 
where we used the easily verifiable bound
$\E_x\big[\int_0^T X^2_t\ud t\big]<\infty$,
in order to get rid of the stochastic integral under expectation. Letting $T\to\infty$ and using dominated convergence (which holds by Proposition \ref{supgsfinite}) we obtain
\begin{align}\label{eq:Lagr}
\begin{aligned}
     &\E_x \Big[ \e^{-r_N(\mu)\tau}\hat f(N,\mu) (X_{\tau} - K) \Big]=\hat f(N,\mu) (x - K)\\
     &\quad +\E_x \Big[ \int_0^{\tau} \e^{-r_N(\mu)t} \Big((\theta-\alpha-r_N(\mu))\hat f(N,\mu) X_t +(\hat{\rho}+\hat{\mu}) K\Big) \ud t\Big].
\end{aligned}
\end{align} 
Substituting into \eqref{VF:recursive} and recalling \eqref{ConnectionWandVnjumps}, we obtain the Lagrange formulation of our problem: 
\begin{equation} \label{wzero} 
W^N(x,N,\mu) = \sup_{\tau\in\cT(\bF^B)} \E_x \Big[ \int_{0}^{\tau} \e^{-r_N(\mu)t} M^N(X_t,N,\mu) \ud t\Big], 
\end{equation}
with 
\begin{equation} 
\label{M0} 
M^N(x,N,\mu) \coloneqq (\hat f(N,\mu)-\beta_N(\mu))(\theta-\alpha-r_N(\mu))x + (\hat{\rho}+\hat{\mu}) K. 
\end{equation}

In the next subsections, we analyse in detail the properties of the value function $V^N(x,N,\mu)$ or, equivalently, of $W^N(x,N,\mu)$. It is worth providing details in full because the case with no jumps in the mortality force provides crucial insight into the general case with jumps.

\subsection{Properties of the value function} \label{Sec:PropN}
First, we state monotonicity, convexity and Lipschitz continuity of the value function with respect to $x\in(0,\infty)$.
\begin{proposition} \label{Prop:ConvN}
For each $\mu\in[\mu_m,\mu_M]$, the mapping $x\mapsto V^N(x,N,\mu)$ is positive, increasing, convex. Moreover, there is $L>0$ independent of $\mu$ such that
\[
\big|V^N(x,N,\mu)-V^N(y,N,\mu)\big|\le L|x-y|,\quad x,y\in(0,\infty).
\]
\end{proposition}
\begin{proof}
Positivity of the value function follows by Proposition \ref{vofinitinessl}--(i).
In order to show that $x\mapsto V^N(x,N,\mu)$ is increasing, it is enough to notice that $X^x_t=xX^1_t$. Then it is immediate to see that for any fixed stopping time $\tau$, the mapping
\begin{align}
\begin{aligned}
x\mapsto G(x,\tau)\coloneqq \E \Big[ \int_{0}^{\tau} \e^{-r_N(\mu) t} ( \alpha+\nu \mu ) X^x_t \ud t + \e^{-r_N(\mu) \tau} \hat f(N,\mu) (X^x_\tau-K)\Big],
\end{aligned}
\end{align}
is increasing which implies 
$V^N(y,N,\mu)=\sup_{\tau\in\cT(\bF^B)} G(y,\tau)\geq \sup_{\tau\in\cT(\bF^B)} G(x,\tau)=V^N(x,N,\mu)$ for $y>x$.
Convexity of $x\mapsto V^N(x,N,\mu)$ follows by linearity of $x\mapsto G(x,\tau)$, because taking supremum over a family of linear functions yields a convex function.

To prove that $x\mapsto V^N(x,N,\mu)$ is Lipschitz, it is enough to show that $x\mapsto W^N(x,N,\mu)$ is Lipschitz. The function $M^N(x,N,\mu)$ is Lipschitz with a constant $L_M$ independent of $(N,\mu)$. 
Using the inequality $\sup_{\tau} a(\tau)-\sup_{\tau} b(\tau) \le \sup_{\tau}(a(\tau)-b(\tau))$, we obtain
\begin{align*}
& W^N(x,N,\mu)-W^N(y,N,\mu)\\
&\quad \le \sup_{\tau\in\mathcal T(\bF^B)} \E \Big[ \int_0^\tau \e^{- r_N(\mu) t} 
\big( M^N(X_t^x,N,\mu ) - M^N(X_t^y,N,\mu ) \big) \ud t \Big] \\
&\quad \le \sup_{\tau\in\mathcal T(\bF^B)} \E \Big[ \int_0^\tau \e^{- r_N(\mu) t} 
\big| M^N(X_t^x,N,\mu ) - M^N(X_t^y,N,\mu ) \big| \ud t \Big] \\
&\quad \le \E \Big[ \int_0^\infty \e^{- r_N(\mu) t} 
\big| M^N(X_t^x,N,\mu ) - M^N(X_t^y,N,\mu ) \big| \ud t \Big].
\end{align*}
The same bound holds analogously for $W^N(y,N,\mu)-W^N(x,N,\mu)$.
Then
\begin{align}
\begin{aligned}
&|W^N(x,N,\mu)-W^N(y,N,\mu )| \\
&\quad\leq \E \Big[ \int_0^\infty \e^{- r_N(\mu) t}| M^N(X_t^x,N,\mu ) - M^N(X_t^y,N,\mu )|  \ud t \Big]\\
&\quad\leq \E \Big[ \int_0^\infty \e^{- r_N(\mu) t}L_M |X_t^x-X_t^y| \ud t \Big]\\
&\quad= \E \Big[ \int_0^\infty \e^{- r_N(\mu) t}L_M  X_t^1 |x-y| \ud t \Big] \le \frac{L_M}{r_N(\mu_m)+\alpha-\theta}|x-y|,
\end{aligned}
\end{align} 
where we used $r_N(\mu)\ge r_N(\mu_m)$ for the final inequality. 
\end{proof}
Lipschitz continuity of $W^N(\cdot,N,\mu)$ is equivalent to that of $V^N(\cdot,N,\mu)$. Then $V^N(\cdot,N,\mu)$ and $W^N(\cdot,N,\mu)$ admit Lipschitz-continuous extension to $[0,\infty)$ with 
$V^N(0,N,\mu)\coloneqq \lim_{x\to0}V^N(x,N,\mu)$ and $W^N(0,N,\mu)\coloneqq \lim_{x\to0}W^N(x,N,\mu)$.

We are now going to establish an asymptotic behaviour of the value function.
\begin{proposition} \label{prop:asympt}
For $\mu\in[\mu_m,\mu_M]$ it holds
$\lim_{x\to\infty}V^N(x,N,\mu)/x=\max\{\hat f(N,\mu),\beta_N(\mu)\}$.
\end{proposition}
The proof requires a simple technical lemma that we are going to present next. 
\begin{lemma}\label{Lemma:VNmuinfty}
Letting 
\begin{equation}\label{DefVnmuinfty}
V_{N}^\infty(\mu)\coloneqq \sup_{\tau\in\cT(\bF^B)} \E \Big[ \int_{0}^{\tau} \e^{-r_N(\mu)t} ( \alpha+\nu \mu ) X_t^1 \ud t+ \e^{-r_N(\mu)\tau}\hat f(N,\mu) X^1_\tau\Big], 
\end{equation} 
we have $V_{N}^\infty(\mu)=\max\{\beta_N(\mu),\hat f(N,\mu)\}$.
\end{lemma}
\begin{proof}
We use the process $X$ to perform a change of measure and reduce $V^\infty_{N}(\mu)$ to a deterministic optimisation. A little care is needed for the change of measure due to the infinite horizon.

Fix $T>0$ and let
\begin{equation} 
V_{N}^T(\mu)\coloneqq \sup_{\tau\in\cT(\bF^B)} \E \Big[ \int_{0}^{\tau\wedge T} \e^{-r_N(\mu)t} ( \alpha+\nu \mu ) X_t^1 \ud t+ \e^{-r_N(\mu)(\tau\wedge T)}\hat f(N,\mu) X^1_{\tau\wedge T}\Big] 
\end{equation}   
and 
\begin{equation} 
U_{N}^T(\mu)\coloneqq \sup_{t\geq0} \Big( \int_{0}^{t\wedge T} \e^{(\theta-\alpha-r_N(\mu)) s} ( \alpha+\nu \mu )  \ud s+ \e^{(\theta-\alpha-r_N(\mu)) (t\wedge T)}\hat f(N,\mu) \Big). 
\end{equation}  
Let us also introduce $U^\infty_{N}(\mu)$ by simply replacing $T$ with $+\infty$ in the expression above.

For each $T\geq 0$, let 
\begin{equation} \label{changeofmeasure}
\frac{\ud \widehat \P}{\ud \P}\Big|_{\cF_T}\coloneqq X^1_T \ \e^{-(\theta-\alpha)T}, 
\end{equation}
so that a simple change of measure guarantees $V_{N}^T(\mu)=U_{N}^T(\mu)$ for each finite $T>0$.
Indeed
\begin{align}
    \begin{aligned}
    V_{N}^T(\mu)
    &=\sup_{\tau\in\cT(\bF^B)} \widehat \E\Big[ \int_{0}^{\tau\wedge T} \e^{(\theta-\alpha-r_N(\mu)) t} ( \alpha+\nu \mu )  \ud t+ \e^{(\theta-\alpha-r_N(\mu)) (\tau\wedge T)}\hat f(N,\mu) \Big]=U_{N}^T(\mu),
    \end{aligned}
\end{align}   
where the second equality holds because the supremum is attained over deterministic times.

It is obvious that $V_{N}^T(\mu)\leq V^{T'}_{N}(\mu)\leq V_{N}^\infty(\mu)$ and $U_{N}^T(\mu)\leq U^{T'}_{N}(\mu)\le U_{N}^\infty(\mu)$ for $T'>T$. Next, we are going to show that 
\begin{equation}\label{eq:lim} 
\lim_{T\to\infty} V_{N}^T(\mu) = V_{N}^\infty(\mu)\quad\text{and}\quad \lim_{T\to\infty} U_{N}^T(\mu) = U_{N}^\infty(\mu). 
\end{equation}
Then, $V^\infty_N(\mu)=U^\infty_N(\mu)=\max\{\beta_N(\mu),\hat f(N,\mu)\}$, where the second equality is by simple algebra. That concludes the proof of the lemma. 

For the proof of \eqref{eq:lim}, take $\varepsilon$-optimal $\tau_\varepsilon$ in $V_{N}^\infty(\mu)$. Then, 
\begin{align}
\begin{aligned}
V_{N}^\infty(\mu) &\leq  \E \Big[ \int_{0}^{\tau_\varepsilon} \e^{-r_N(\mu)t} ( \alpha+\nu \mu ) X_t^1 \ud t+ \e^{-r_N(\mu)\tau_\varepsilon}\hat f(N,\mu) X^1_{\tau_\varepsilon}\Big] + \varepsilon\\
&=\E \Big[ \liminf_{T\to\infty}\int_{0}^{\tau_\varepsilon\wedge T} \e^{-r_N(\mu)t} ( \alpha+\nu \mu ) X_t^1 \ud t+ \e^{-r_N(\mu)(\tau_\varepsilon\wedge T)}\hat f(N,\mu) X^1_{\tau_\varepsilon\wedge T}\Big] + \varepsilon\\
&\leq \liminf_{T\to\infty} \E \Big[ \int_{0}^{\tau_\varepsilon\wedge T} \e^{-r_N(\mu)t} ( \alpha+\nu \mu ) X_t^1 \ud t+ \e^{-r_N(\mu)(\tau_\varepsilon\wedge T)}\hat f(N,\mu) X^1_{\tau_\varepsilon\wedge T}\Big] + \varepsilon\\
&\leq \liminf_{T\to\infty} V_{N}^T(\mu) + \varepsilon,
\end{aligned}
\end{align}
where the second inequality is by Fatou's lemma. 
Since $V^T_N(\mu)\le V^\infty_N(\mu)$, and $\varepsilon>0$ is arbitrary, 
\begin{equation} 
V_{N}^\infty(\mu) \leq \liminf_{T\to\infty} V_{N}^T(\mu) \leq \limsup_{T\to\infty} V_{N}^T(\mu) \leq V_{N}^\infty(\mu).
\end{equation}
Taking $\varepsilon$-optimal $t_\varepsilon$ in $U_{N}^\infty(\mu)$ and arguing in exactly the same way we also obtain
\begin{equation} 
U_{N}^\infty(\mu) \leq \liminf_{T\to\infty} U_{N}^T(\mu) \leq \limsup_{T\to\infty} U_{N}^T(\mu) \leq U_{N}^\infty(\mu).
\end{equation}
Thus \eqref{eq:lim} holds, concluding the proof of the lemma.
\end{proof}

\begin{proof}[Proof of Proposition \ref{prop:asympt}]
We provide a full proof for $K>0$. The case $K<0$ follows by analogous arguments.
We have
\begin{align}
\begin{aligned}
\frac{V^N(x,N,\mu)}{x} &=  \sup_{\tau\in\cT(\bF^B)} \E \Big[ \int_{0}^{\tau} \e^{-r_N(\mu)t} ( \alpha+\nu \mu ) X_t^1 \ud t+ \e^{-r_N(\mu)\tau}\hat f(N,\mu) \Big(X^1_\tau-\frac{K}{x}\Big)\Big]\\
&\leq \sup_{\tau\in\cT(\bF^B)} \E \Big[ \int_{0}^{\tau} \e^{-r_N(\mu)t} ( \alpha+\nu \mu ) X_t^1 \ud t+ \e^{-r_N(\mu)\tau}\hat f(N,\mu) X^1_\tau\Big]=V_{N}^\infty(\mu).
\end{aligned}
\end{align} 
Letting $x\to\infty$ yields
$\limsup_{x\to\infty}V^N(x,N,\mu)/x \leq V_{N}^\infty(\mu)$. 
Now we prove the reverse inequality.
For any $\varepsilon>0$, taking $x>\frac{K}{\varepsilon}$, we have 
\begin{align}
\begin{aligned}
\frac{V^N(x,N,\mu)}{x} &\geq \sup_{\tau\in\cT(\bF^B)} \E \Big[ \int_{0}^{\tau}\!\! \e^{-r_N(\mu)t} ( \alpha\!+\!\nu \mu ) X_t^1 \ud t\!+\! \e^{-r_N(\mu)\tau}\hat f(N,\mu) (X^1_\tau\!-\!\varepsilon) \Big]
\\
&
\ge V_{N}^\infty(\mu)\!-\!\hat f(N,\mu) \varepsilon.
\end{aligned}
\end{align}
Letting $x\to \infty$ and using that $\varepsilon$ is arbitrary yields
$\liminf_{x\to\infty}V^N(x,N,\mu)/x \geq V_{N}^\infty(\mu)$. 
Then $\lim_{x\to\infty}V^N(x,N,\mu)/x=V^\infty_N(\mu)$ and Lemma \ref{Lemma:VNmuinfty} concludes the proof.
\end{proof}

\subsection{Stopping and Continuation regions}
In this section, we determine the structure of stopping and continuation regions defined in \eqref{SCRegionsV}. Then, we prove the so-called smooth-fit condition to ensure continuous differentiability of the value function at the boundary of the regions. 

It is convenient to split the analysis into two cases, depending on the sign of $K$ (the case $K=0$ is addressed later in Proposition \ref{prop:K=0}). The latter determines the sign of $M^N(0,N,\mu)$ and $W^N(0,N,\mu)$ as shown below. 

\begin{proposition}  \label{w0Kposw}
We have $M^N(0,N,\mu)=(\hat{\rho}+\hat{\mu}) K$. Moreover, the following holds:
\begin{itemize}
\item[(a)] For $K>0$ we have $W^N(0,N,\mu)= K(\hat\rho+\hat\mu)/r_N(\mu)>0$.
\item[(b)] For $K<0$, there is $b>0$ such that $W^N(x,N,\mu)=0\iff x\in[0,b]$ (with the convention $[0,b]=[0,\infty)$ if $b=\infty$).
\end{itemize}
\end{proposition}
Before proving the proposition we state its corollary. The function $x\mapsto W^N(x,N,\mu)$ is non-negative and convex by Proposition \ref{Prop:ConvN} and \eqref{ConnectionWandVnjumps}, because $V^N(x,N,\mu)\ge \hat f(N,\mu)(x-K)$. Then, the set $\cS_{N,\mu}=\{x\ge0: W^N(x,N,\mu)=0\}$ must be either empty or connected. Figure \ref{Fig:casesN(2)} illustrates all possible shapes of the function $x\mapsto W^N(x,N,\mu)$, from which the statements in the corollary are easily deduced. Notice that in describing the sets $\cC_{N,\mu}$ and $\cS_{N,\mu}$ in the corollary, we make use of the extension of $W^N$ to $[0,\infty)$.

\begin{corollary}\label{cor:CS}
For $K>0$, three cases may arise:
\begin{itemize}
\item \textit{Case 1}: $\cC_{N,\mu}=[0,\infty)$ and $\cS_{N,\mu}=\varnothing$.
\item \textit{Case 2}: $\cC_{N,\mu}=[0,b)$ and $\cS_{N,\mu}=[b,\infty)$, for some $b=b(\mu)\in(0,\infty)$.
\item \textit{Case 3}: $\cC_{N,\mu}=[0,b_1)\cup(b_2,\infty)$ and $\cS_{N,\mu}=[b_1,b_2]$, for some $0\!<\!b_1\!=\!b_1(\mu)\!\leq\! b_2(\mu)\!=\!b_2\!<\!\infty$.
\end{itemize}
For $K<0$ we have either 
\begin{itemize}
\item Case 4: $\cC_{N,\mu}=(b,\infty)$ and $\cS_{N,\mu}=[0,b]$ for some $b=b(\mu)\in(0,\infty)$, 
\end{itemize}
or 
\begin{itemize}
\item Case 5: $\cC_{N,\mu}=\varnothing$ and $\cS_{N,\mu}=[0,\infty)$.
\end{itemize}
\end{corollary}

\begin{figure}[h]
\centering
\begin{subfigure}{0.3\textwidth}
\centering
\includegraphics[width=1\textwidth]{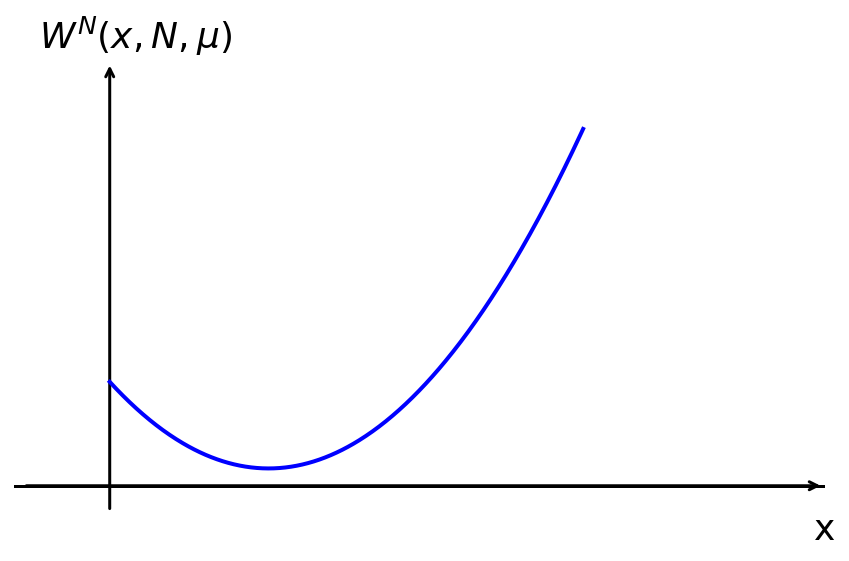}
\caption{Case 1}
\end{subfigure}
\hfill
\begin{subfigure}{0.3\textwidth}
\centering
\includegraphics[width=1\textwidth]{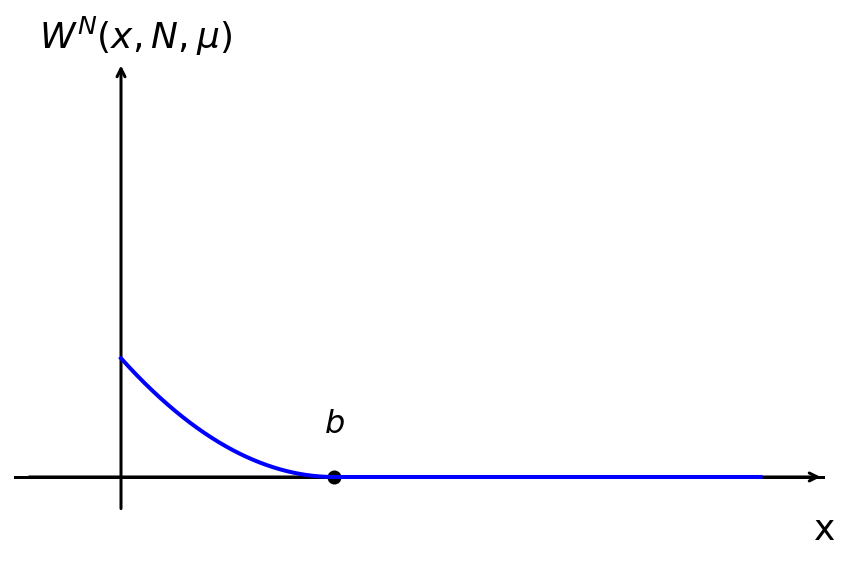}
\caption{Case 2}
\end{subfigure}
\hfill
\begin{subfigure}{0.3\textwidth}
\centering
\includegraphics[width=1\textwidth]{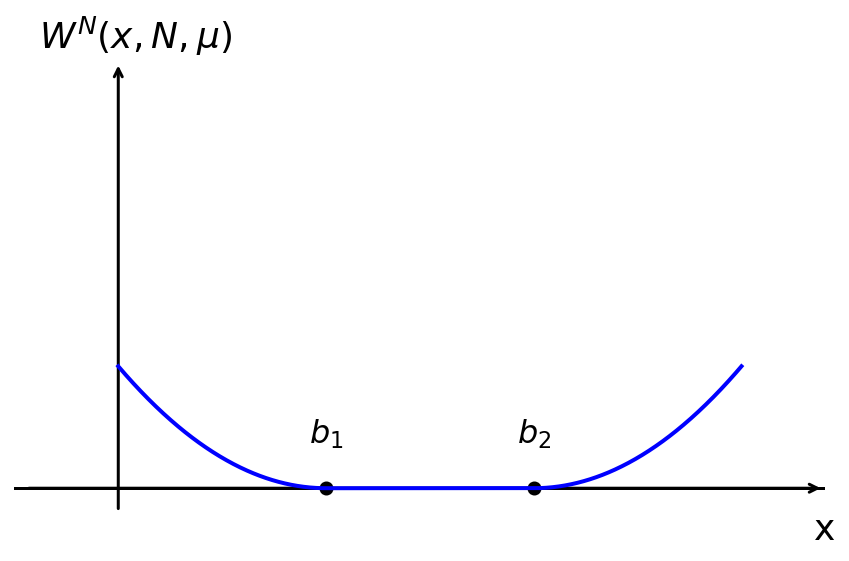}
\caption{Case 3}
\end{subfigure}
\begin{subfigure}{0.3\textwidth}
\centering
\includegraphics[width=1\textwidth]{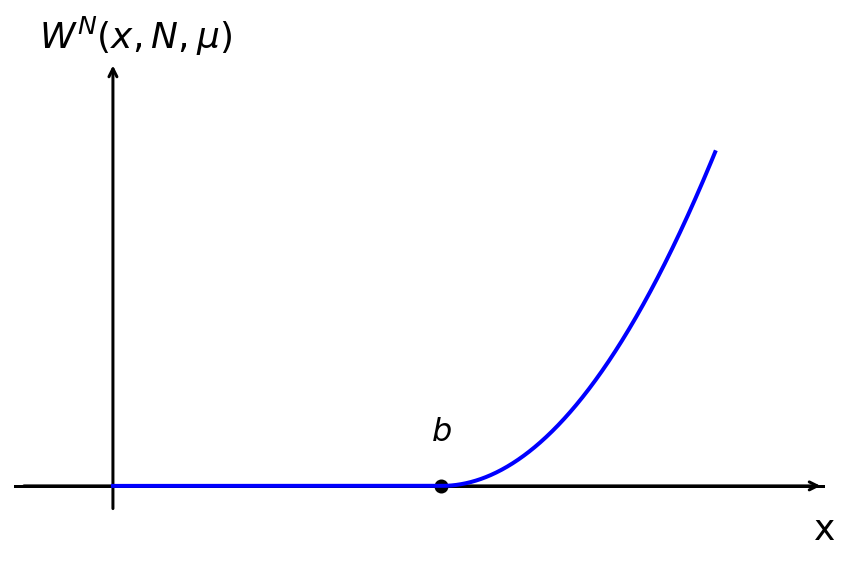}
\caption{Case 4}
\end{subfigure}
\hspace{0.05\textwidth} 
\begin{subfigure}{0.3\textwidth}
\centering
\includegraphics[width=1\textwidth]{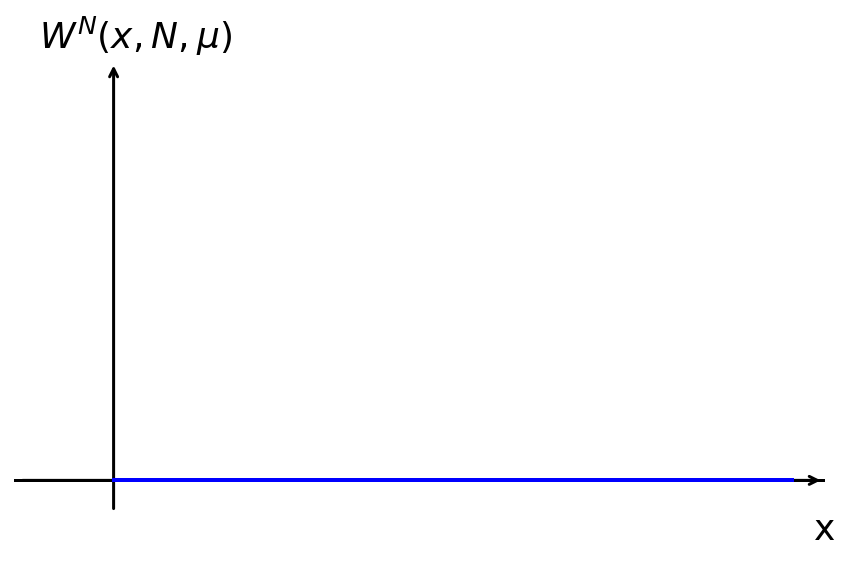}
\caption{Case 5}
\end{subfigure}
\caption{All possible behaviours of $W^N(\cdot,N,\mu)$ for $K\neq 0$.}
\label{Fig:casesN(2)}
\end{figure}
\begin{remark}
A detailed analysis of the occurrence of the five cases above will be performed in Section \ref{sec:proof}, where we complete the proof of Theorem \ref{maintheorem}. 
The explicit affine structure of $M^N(\cdot,N,\mu)$ in \eqref{M0} 
rules out the non-monotonic behaviour needed for Case 3, which is stated only for completeness and for comparison with the general result in Theorem \ref{maintheorem} with $n\le N$. 
\end{remark}

\begin{proof}[Proof of Proposition \ref{w0Kposw}]
The first claim is obvious from the definition of $M^N(x,N,\mu)$ (cf.\ \eqref{M0}).
Next, we prove (a).
By continuity of $M^N$ there exists $\varepsilon>0$ and $x_\varepsilon>0$ such that $M^N(x,N,\mu)\ge (\hat \rho+\hat \mu)K-\varepsilon$ for $x\in[0,x_\varepsilon)$. Let
$\tau_\varepsilon\coloneqq \inf\{t\geq 0 : X_t \geq x_\varepsilon\}$,
with the usual convention $\inf \varnothing=+\infty$. Take $x\in(0,x_\varepsilon)$ then $\P_x(\tau_\varepsilon>0)=1$ and
\begin{align}
\begin{aligned}
W^N(x,N,\mu)&\geq \E_x \Big[ \int_0^{\tau_\varepsilon} \e^{-r_N(\mu)t} M^N(X_t,N,\mu) \ud t\Big]\geq \frac{ (\hat \rho+\hat \mu)K-\varepsilon}{r_N(\mu)}  \E_x \Big[ 1-\e^{-r_N(\mu)\tau_\varepsilon} \Big].
\end{aligned}
\end{align}
From \cite[Ch.\ II.1, N\textsuperscript{\underline{o}} 10, p. 18]{borodin2015handbook}, $\E_x [ \e^{-r_N(\mu)\tau_\varepsilon} ]=(x/x_\varepsilon)^{\gamma^+_{N}(\mu)}$. Then, letting $x\to 0$ yields 
$W^N(0,N,\mu)\geq ( (\hat \rho+\hat \mu)K-\varepsilon)/r_N(\mu)$. Since $\eps>0$ is arbitrary we conclude $W^N(0,N,\mu)\ge M^N(0,N,\mu)/r_N(\mu)>0$. 

For the reverse inequality, we notice that there exists $\eps'>0$ and $x'_\varepsilon>0$ such that $M^N(x,N,\mu)\le (\hat \rho+\hat \mu)K+\varepsilon'$ for $x\in[0,x'_\varepsilon)$. With no loss of generality, thanks to the lower bound obtained above, we may assume $W^N(x,N,\mu)>0$ for $x\in[0,x'_\eps)$. Then $[0,x'_\eps)\cap\cS_{N,\mu}=\varnothing$. Let us now consider $ \tau'_\varepsilon\coloneqq  \inf\{ t\!\geq\! 0\!:\! X_t\!\geq\! x'_\varepsilon \}$ 
so that $\P_x(\tau^*_{N,\mu}>\tau'_\eps)=1$ for $x\in[0,x'_\eps)$. 
Then, for $x\in [0, x'_\varepsilon)$, using optimality of $\tau^*_{N,\mu}$
\begin{align}
\begin{aligned}
W^N(x,N,\mu)
&=\E_x \Big[ \int_0^{\tau'_\varepsilon} \e^{-r_N(\mu)t} M^N(X_t,N,\mu) \ud t+ \int_{\tau'_\varepsilon}^{\tau^*_{N,\mu}} \e^{-r_N(\mu)t} M^N(X_t,N,\mu) \ud t\Big].
\end{aligned}
\end{align}
Denote by $(\theta_t)_{t\ge 0}$ the shift operator (i.e., $\theta_tX_s(\omega)=X_{t+s}(\omega)$). The strong Markov property yields
\begin{align}
\begin{aligned}
&\E_x \Big[ \int_{\tau'_\varepsilon}^{\tau^*_{N,\mu}} \e^{-r_N(\mu)t} M^N(X_t,N,\mu) \ud t\Big]=\\
&\quad=\E_x \Big[ \int_{\tau'_\varepsilon}^{\tau'_\varepsilon+\tau^*_{N,\mu}\circ\theta_{\tau'_\varepsilon}} \e^{-r_N(\mu)t} M^N(X_t,N,\mu) \ud t\Big]\\
&\quad=\E_x \Big[e^{-r_N(\mu)\tau'_\varepsilon} \int_{0}^{\tau^*_{N,\mu}\circ\theta_{\tau'_\varepsilon}} \e^{-r_N(\mu)t} M^N(X_{\tau'_\varepsilon+t},N,\mu) \ud t\Big]\\
&\quad=\E_x \Big[e^{-r_N(\mu)\tau'_\varepsilon} \E_x \Big[\int_{0}^{\tau^*_{N,\mu}\circ\theta_{\tau'_\varepsilon}} \e^{-r_N(\mu)t} M^N(X_{\tau'_\varepsilon+t},N,\mu) \ud t \Given  \cF^B_{\tau'_\varepsilon} \Big]\Big]\\
&\quad=\E_x \Big[e^{-r_N(\mu)\tau'_\varepsilon}  \E_{X_{\tau'_\varepsilon}} \Big[\int_{0}^{\tau^*_{N,\mu}} \e^{-r_N(\mu)t} M^N(X_{t},N,\mu) \ud t\Big] \Big]\\
&\quad=\E_x \Big[\e^{-r_N(\mu)\tau'_\varepsilon}  W^N(X_{\tau'_\varepsilon},N,\mu) \Big]=W^N(x'_\eps,N,\mu)\Big(\frac{x}{x'_\eps}\Big)^{\gamma^+_{N}(\mu)},
\end{aligned}
\end{align}
where we used optimality of $\tau^*_{N,\mu}$ and
$\e^{-r_N(\mu)\tau'_\eps}W^N(X_{\tau'_\eps},N,\mu)1_{\{\tau'_\eps=\infty\}}=0$, 
because $0\le X_t(\omega)\le x'_\eps$ for $t\in[0,\tau'_\eps(\omega)]$, to obtain the penultimate equality.  The final expression holds by \cite[Ch.\ II.1, N\textsuperscript{\underline{o}} 10, p. 18]{borodin2015handbook}.
Now, combining the expressions above, we obtain
\begin{align}
\begin{aligned}\label{eq:cubs}
W^N(x,N,\mu)&=\E_x\Big[\int_0^{\tau'_\eps}\e^{-r_N(\mu)t}M^N(X_t,N,\mu)\ud t\Big]+W^N(x'_\eps,N,\mu)\Big(\frac{x}{x'_\eps}\Big)^{\gamma^+_{N}(\mu)}\\
&\le\frac{K(\hat \rho+\hat\mu)+\eps'}{r_N(\mu)}\E_x\big[1-\e^{-r_N(\mu)\tau'_\eps}\big]+W^N(x'_\eps,N,\mu)\Big(\frac{x}{x'_\eps}\Big)^{\gamma^+_{N}(\mu)}\\
&=\frac{K(\hat \rho+\hat\mu)+\eps'}{r_N(\mu)}\Big[1-\Big(\frac{x}{x'_\eps}\Big)^{\gamma^+_{N}(\mu)}\Big]+W^N(x'_\eps,N,\mu)\Big(\frac{x}{x'_\eps}\Big)^{\gamma^+_{N}(\mu)}.
\end{aligned}
\end{align}
Letting $x\to 0$ first and $\eps'\to 0$ second, yields $W^N(0,N,\mu)\le K(\hat \rho+\hat\mu)/r_N(\mu)$, as needed. 
This concludes the proof of (a).

Next, we prove (b). We split the argument into two parts. First, we show that $W^N(0,N,\mu)=0$ and then we improve the result to show that actually $W^N(\cdot,N,\mu)=0$ on a non-trivial right-neighbourhood of zero. 

({\em Part 1}) Recall that $M^N(0,N,\mu)<0$. We argue by contradiction assuming $W^N(0,N,\mu)>0$. In particular, there is $\eps>0$ such that $W^N(0,N,\mu)>2\eps$. Since $M^N(0,N,\mu)<0$, with no loss of generality we further assume $M^N(0,N,\mu)<-2\eps$. Letting
$x_\varepsilon\coloneqq  \inf\{x>0: W^N(x,N,\mu)<\varepsilon\}$,  
with $\inf \varnothing=+\infty$, by continuity of $W^N$ we have $x_\eps>0$ and $W^N(x,N,\mu)>\varepsilon$ for $x\in [0, x_\varepsilon)$. 
Now, letting  
$x''_\varepsilon\coloneqq  \inf\{x>0: M^N(x,N,\mu)>-\varepsilon\}$, 
we have also $x''_\eps>0$ by continuity of $M^N$. Therefore  $M^N(x,N,\mu)\le-\varepsilon$ and $W^N(x,N,\mu)\ge\varepsilon$ on $[0,\hat{x}_\varepsilon)$, with 
$\hat{x}_\varepsilon\coloneqq  \min\{x_\varepsilon, x''_\varepsilon\}>0$.
In particular, it must be 
$\cS_{N,\mu}\cap[0, \hat x_\varepsilon)=\varnothing$. The latter implies $\P_x(\tau_\varepsilon<\tau^*_{N,\mu})=1$ for $x\in [0, \hat{x}_\varepsilon)$ for
$ \tau_\varepsilon\coloneqq  \inf\{ t\!\geq\! 0\!:\! X_t\!\geq\! \hat{x}_\varepsilon \}$. 
If $\hat{x}_\varepsilon=+\infty$, then $\tau_\varepsilon=+\infty$, yielding the contradiction
\begin{align}
\begin{aligned}
\label{eq:wN0}
0\le W^N(x,N,\mu)=w_N(x,N,\mu)\coloneqq \E_x \Big[ \int_0^{\infty} \e^{-r_N(\mu)t} M^N(X_t,N,\mu) \ud t\Big] \leq -\frac{\varepsilon}{r_N(\mu)}.
\end{aligned}
\end{align}
If $\hat{x}_\varepsilon<+\infty$, then for $x\in [0, \hat{x}_\varepsilon)$, identical calculations as those leading to \eqref{eq:cubs} (but replacing $\tau'_\eps$, $x'_\eps$ with $\tau_\eps$, $\hat x_\eps$) yield
\begin{align}
\begin{aligned}
0\le W^N(x,N,\mu)
&\le - \frac{\varepsilon}{r_N(\mu)} \Big[1-\Big(\frac{x}{\hat x_\eps}\Big)^{\gamma^+_{N}(\mu)}\Big] 
+ W^N(\hat{x}_\varepsilon,N,\mu) \Big(\frac{x}{\hat x_\eps}\Big)^{\gamma^+_{N}(\mu)}.
\end{aligned}
\end{align}
Letting $x\to 0$ we reach a contradiction. Therefore, it must be $W^N(0,N,\mu)=0$. 

({\em Part 2}) Since $x\mapsto W^N(x,N,\mu)$ is convex, with $W^N(0,N,\mu)=0$, then it must be $\cS_{N,\mu}=[0,y]$ for some $y\ge 0$ (possibly $\cS_{N,\mu}=[0,\infty)$). We want to show that indeed $y>0$. Arguing by contradiction, suppose $y=0$. Since for $x>0$ we have $X^x_t>0$ for all $t\ge 0$, then $\tau^*_{N,\mu}=\infty$ and the equality in \eqref{eq:wN0} holds.
However, letting $x\to 0$, we get $W^N(0,N,\mu)=\lim_{x\to 0}w_N(x,N,\mu)=M^N(0,N,\mu)/r_N(\mu)<0$, where the second equality is by dominated convergence. 
Then we have reached a contradiction.
This proves (b).
\end{proof}

We show that smooth-fit condition holds at boundary points of $\cC_{N,\mu}$ which lie in $(0,\infty)$.
In what follows,
we use the symbols $f(x+)=\lim_{\eps\downarrow 0}f(x+\eps)$ and $f(x-)=\lim_{\eps\downarrow 0}f(x-\eps)$ for the right- and left-limit of a function $f$ at a point $x$. 

\begin{proposition} \label{smoothfitv1} The mapping $x\mapsto W^N(x,N,\mu)$ is continuously differentiable at $\partial \cC_{N,\mu}\cap(0,\infty)$. Hence, $W^N(\cdot,N,\mu)\in C^1((0,\infty))$.
\end{proposition}

\begin{proof}
Recall that continuity of $W^N(\cdot,N,\mu)$ implies $W^N(\cdot,N,\mu)\in C^2(\cC_{N,\mu})$ (cf.\ Section \ref{sec:prelimOS}). Then, the second claim in the proposition follows from the first claim and the fact that $W^N(x,N,\mu)=0$ for $x\in\cS_{N,\mu}$. 
It remains to prove the first claim. We only prove it for Case 4 of Corollary \ref{cor:CS}, that is $\cC_{N,\mu}=(b,\infty)$ and $\cS_{N,\mu}=[0,b]$, with $b>0$. All remaining cases can be worked out by analogous arguments. The proof is based on ideas developed in \cite[Ch.\ IV.9]{peskir2006optimal}.

Since $W^N(\cdot,N,\mu)$ is convex, and $C^1$ on $(0,b)$ and $(b,\infty)$, then it is continuously differentiable at $b$ 
provided that right- and left-derivative coincide. It is clear that $W^N_x(x,N,\mu)=0$ for $x\in(0,b)$ and $W^N_x(b-,N,\mu)=0$. Let us now consider the right-derivative in $x=b$.
Fix $\eps>0$. Let $\tau_\varepsilon^*\coloneqq  \inf \{t\geq 0: X^{b+\varepsilon}_t\leq b\}$ be the optimal stopping time for $W^N(b+\varepsilon,N,\mu)$. Then
\begin{align}
\begin{aligned}
0&\le \frac{W^N(b+\varepsilon,N,\mu)-W^N(b,N,\mu)}{\varepsilon}\\
&\leq \frac{1}{\varepsilon} \E \Big[ \int_0^{\tau_\varepsilon^*} \e^{-r_N(\mu)t} (M^N(X^{b+\varepsilon}_t,N,\mu))-M^N(X^{b}_t,N,\mu)) \ud t\Big]\\
&=(\hat f(N,\mu)\!-\!\beta_N(\mu))(\theta\!-\!\alpha\!-r_N(\mu))\E \Big[ \int_0^{\tau_\varepsilon^*} \e^{-r_N(\mu)t} X_t^1 \ud t\Big], \end{aligned}\end{align}
where the first inequality holds because $W(b+\eps,N,\mu)\geq 0$ and $W^N(b,N,\mu)=0$, and the equality holds because
$M^N(X^{b+\varepsilon}_t,N,\mu)-M^N(X^{b}_t,N,\mu)
= (\hat f(N,\mu)-\beta_N(\mu))(\theta-\alpha-r_N(\mu))  X_t^1 \varepsilon$.
Thanks to the law of iterated logarithm, it is easy to show that $\tau^*_\varepsilon \downarrow 0$, $\P$-a.s.\ as $\varepsilon \downarrow 0$. 
Indeed, since $\tau^*_\varepsilon$ is decreasing in $\varepsilon$, the limit $\tau^*_0:=\lim_{\varepsilon\downarrow0}\tau^*_\varepsilon$ exists. We claim that $\tau^*_0=0$, $\P$-a.s. Indeed, suppose by contradiction, there is $\omega\in\Omega$ such that $\tau^*_0(\omega)>0$, then there is  $\delta(\omega)>0$ such that $\tau^*_0(\omega)>\delta(\omega)$. From monotonicity, $\tau^*_\varepsilon(\omega)>\delta(\omega)$ for all $\varepsilon>0$. This implies that for every $\varepsilon>0$, $X_t^{b+\varepsilon}(\omega)>b$ for $t\in[0,\delta(\omega)]$. Now, taking the limit as $\varepsilon$ goes to zero we have $X_t^{b}(\omega)\geq b$ for every $t\in[0,\delta(\omega)]$,
or equivalently,
$B_t(\omega)\ge -\big((2(\theta-\alpha)-\sigma^2)/2\sigma\big)t$, for every $t\in[0,\delta(\omega)]$. The latter contradicts the law of the iterated logarithm, see \cite[Ch. 2.9.E, Thm 9.23]{karatzas2014brownian}. Hence $\tau^*_\varepsilon\downarrow0$, $\P$-a.s. as $\varepsilon\downarrow0$.

Then, letting $\eps\to 0$ in the equation above yields $W^N_x(b+,N,\mu)=0=W^N_x(b-,N,\mu)$ by Monotone Convergence theorem. This concludes the proof.
\end{proof}

\section{Optimal Annuitization with Jumps in the Mortality Force} \label{OptimalAnnuitizationwithJumps}
In this section, we use mostly induction arguments to extend our analysis from Section \ref{ConstantForceOfMortality} to the value functions 
$V^N(x,n, \mu)$ for any $0\le n \leq N$. 

\begin{proposition} \label{Vnfiniteness}
It holds $0\le V^N(x,n,\mu)\leq L(1+x)$, for $0\le n\le N$, $\mu\in[\mu_m,\mu_M]$ and $x\in(0,\infty)$, with $L>0$ independent of $(x,n,\mu)$.
\end{proposition}
\begin{proof}
The statement holds for $n=N$ by Proposition \ref{vofinitinessl}. Arguing by induction, let us now assume $0\le V^N(x,n+1,\mu)\le L(1+x)$ for all $(x,\mu)$. Since $0\le \hat f(n,\mu)\le \rho^{-1}(\hat \rho+\hat \mu)$, taking $\tau=T$ in \eqref{VF:recursive} and then letting $T\to \infty$ we deduce $V^N(x,n,\mu)\ge 0$ for all $(x,\mu)$. Next, we prove the upper bound.

From \eqref{VF:recursive} and \eqref{boundsupermg}
    \begin{align}
        \begin{aligned} \label{vf2bound1}
            V^N(x,n,\mu)
            &\leq \hat{f}(n,\mu) (x +|K|) \\ 
            &\quad +\sup_{\tau\in\cT(\bF^B)} \E_{x} \Big[\int_0^\tau \e^{- r_n(\mu) t} \big[( \alpha+\nu \mu  ) X_t +\lambda_{n+1} \hat{V}^N(X_t,n+1;\mu) \big] \ud t \Big].
    \end{aligned}
    \end{align}
By the induction hypothesis $\hat{V}^N(x,n+1;\mu) \leq L(1+x) \int_{\mu_m}^{\mu_M} q(n,\mu,\ud z)=L (1+x)$.

Then, continuing from \eqref{vf2bound1}, setting $\bar\lambda=\max_{n\in\N_0,n\le N}\lambda_{n+1}$ and using $r_n(\mu)\ge \rho+\mu_m$, yields
    \begin{align}
        \begin{aligned}
            V^N(x,n,\mu) &\le \hat{f}(n,\mu) (x\!+\!|K|)\! +\! \E_{x} \Big[\int_0^\infty \e^{- (\rho+\mu_m) t} \big[( \alpha\!+\!\nu \mu_M \! +\! \bar\lambda L) X_t\! +\!\bar\lambda L\big] \ud t \Big]\\
            &\le\frac{\hat \rho+\hat \mu}{\rho} (x\!+\!|K|)\!+\! \frac{\alpha+\nu \mu_M  + \bar \lambda L}{\rho+\mu_m +\alpha-\theta} x + \frac{\bar \lambda L}{\rho+\mu_m }\le\hat{L}(1+x),
    \end{aligned}
    \end{align}
for suitable $\hat L>0$ independent of $(n,\mu)$, where we used again $\hat f(n,\mu)\le \rho^{-1}(\hat \rho+\hat \mu)$ and the explicit formula for $\E_x[X_t]$ for the second inequality.
Relabelling $\hat L$ as $L$, we obtain $V^N(x,n,\mu)\leq L (1+x)$ as needed, because $n$ takes only finitely many values.
\end{proof}

\begin{proposition} \label{convexw}
For any $n\le N$ and $\mu\in[\mu_m,\mu_M]$ the mapping $x\mapsto V^N(x,n,\mu)$ is convex. Moreover, there is $L>0$ independent of $(n,\mu)$ such that
\[
\big|V^N(x,n,\mu)-V^N(y,n,\mu)\big|\le L|x-y|,\quad x,y,\in(0,\infty).
\]
\end{proposition}
\begin{proof}
It suffices to prove all claims for $W^N(x,n,\mu)$ as introduced in \eqref{ConnectionWandVnjumps}.
Convexity of $x\mapsto W^N(x,N,\mu)$ is established in Proposition \ref{Prop:ConvN}. 
To prove convexity for any $n\le N$, we proceed by induction. 
Assume that $x\mapsto W^N(x,n+1,\mu)$ is convex. 
From \eqref{ConnectionWandVnjumps} it follows that the mapping $x\mapsto V^N(x,n+1,\mu)$ is convex, which in turn implies convexity of $x\mapsto \hat{V}^N(x, n+1; \mu)$. 
Then, it is easy to check that $x\mapsto M^N(X^x_t(\omega),n,\mu)=M^N(xX^1_t(\omega),n,\mu)$ is convex for each $\omega\in\Omega$.
Using the latter, for $\alpha\in(0,1)$ we have
\begin{align}
\begin{aligned}
W^N\big(\alpha x\! +\! (1\!-\!\alpha)y,n,\mu \big)
&\leq \sup_{\tau\in\cT(\bF^B)}\E\Big[\! \int_0^\tau \!\e^{- r_n(\mu) t}   \big(\alpha M^N(X^x_t,n,\mu)\!+\!(1-\alpha)M^N(X^y_t,n,\mu)\big) \ud t \Big]\\
&\le \alpha W^N(x,n,\mu ) +(1-\alpha) W^N(y,n,\mu ),
\end{aligned}
\end{align}
which proves convexity of $W^N(\cdot,n,\mu)$.

We also prove Lipschitz continuity by induction, because it holds with $n=N$ (cf.\ Proposition \ref{Prop:ConvN}). 
Suppose that $|W^N(x,n+1,\mu)-W^N(y,n+1,\mu)|\le L|x-y|$. 
It is easy to check that the induction hypothesis implies $|\hat{V}^N(x, n+1; \mu)-\hat{V}^N(y,n+1;\mu)|\le L'|x-y|$, with $L'=L+\sup_{n,\mu}|\hat f(n,\mu)|$.
Thus, 
 $| M^N(x,n,\mu ) - M^N(y,n,\mu )|\le L''|x-y|$,
with 
$L''\coloneqq \sup_{n,\mu}|(\hat{f}(n,\mu)-\beta_{n}(\mu))(\theta-\alpha- r_n(\mu))| +L'\bar \lambda$ with $\bar\lambda=\max_{n\in\N_0,n\le N}\lambda_{n+1} $.
Finally, we have
\begin{align}
\begin{aligned}
|W^N(x,n,\mu )-W^N(y,n,\mu )| &\leq \E \Big[ \int_0^\infty \e^{- r_n(\mu) t}| M^N(X_t^x,n,\mu ) - M^N(X_t^y,n,\mu )|  \ud t \Big]\\
&\le |x-y| \E \Big[ \int_0^\infty \e^{- r_n(\mu) t}L'' X^1_t  \ud t \Big]. 
\end{aligned}\end{align} 
Since the expectation is bounded by a constant independent of $(n,\mu)$ and $n$ takes finitely many values, the proof is complete.
\end{proof}

\begin{remark}
Thanks to the above proposition we can continuously extend functions $M^N(\cdot,n,\mu)$ and $V^N(\cdot,n,\mu)$ to $[0,\infty)$.  
\end{remark}

We now analyse the structure of the continuation and stopping sets. The case $K=0$ is completely addressed in the next proposition, which also proves Theorem \ref{maintheorem}--(3). The cases $K<0$ and $K>0$ require a more detailed analysis: possible geometries for the continuation/stopping sets are described in Corollary \ref{cor:CSn}; the occurrence of each case is detailed in the proofs of (1) and (2) of Theorem \ref{maintheorem}, later in this section.
\begin{proposition}\label{prop:K=0}
   Let $K=0$. For each $(n,\mu)\in\N_0\times[\mu_m,\mu_M]$, $n\le N$, we have
\begin{align}\label{eq:VNK0}
V^N(x,n,\mu)=\Big(\hat f(n,\mu)+\frac{\big[R(n,\mu)\big]^+}{ r_{n}(\mu) + \alpha - \theta }\Big) x,
\end{align}
with $R(n,\mu)$ defined in \eqref{eq:Lnmu}.
Then, $\{0\}\subset\cS_{n,\mu}$. Moreover, $R(n,\mu)\le 0\implies\cC_{n,\mu}=\varnothing$ and $R(n,\mu)>0\implies \cC_{n,\mu}=(0,\infty)$.
\end{proposition}
\begin{proof}
First, notice that $M^N(x,N,\mu)= R (N,\mu)x$. Then, using the change of measure \eqref{changeofmeasure}, we get
\begin{align} 
\begin{aligned} \label{WK0}
W^N(x,N,\mu)
&=x\sup_{T\ge 0}\int_0^T\e^{(\theta-\alpha-r_{N}(\mu))t} R(N,\mu) \ud t=x\frac{\big[ R (N,\mu) \big]^+}{r_{N}(\mu) +\alpha -\theta}.
\end{aligned}
\end{align}
All claims in the proposition hold for $n=N$. 
Proceeding by induction, assume the proposition holds for $n$, then from \eqref{Eq:hatV}, we get 
\begin{align}
    \hat V^N(x,n;\mu) = x \ \int_{\mu_m}^{\mu_M} \Big(\hat f(n,z) +  \frac{\big[R(n,z) \big]^+}{r_{n}(z)+\alpha-\theta}\Big) q(n-1,\mu,\ud z).
\end{align}
By combining the latter expression with \eqref{Def:M_njumps} and recalling the definition \eqref{eq:Lnmu} we get $M^N(x,n-1,\mu)=R(n-1,\mu) x$.
By the same change of measure as in \eqref{WK0}, we obtain 
\begin{align*}
    W^N(x,n-1,\mu)= \frac{\big[R(n-1,\mu) \big]^+}{r_{n-1}(\mu)+\alpha-\theta} x.
\end{align*}
Recalling $W^N(x,n,\mu)= V^N(x,n,\mu)-\hat f(n,\mu)x$ we prove \eqref{eq:VNK0}, concluding the induction step.
\end{proof}

In the proposition below, we 
extend the results of Proposition \ref{w0Kposw} to the general setting $n\le N$.

\begin{proposition} \label{Kposw}
The following holds:
\begin{itemize}
    \item[(a)] For $K>0$, we have  $W^N(0,n,\mu)>0$.
    \item[(b)] For $K<0$, there is $b>0$ such that $W^N(x,n,\mu)=0\iff x\in[0,b]$ (with the convention $[0,b]=[0,\infty)$ if $b=\infty$).
\end{itemize}
\end{proposition}

\begin{proof}
We prove (a). First, we show that $M^N(0,n,\mu)>0$ implies $W^N(0,n,\mu)>0$. By continuity of $M^N$ there exists $\varepsilon>0$ and $x_\varepsilon>0$ such that $M^N(x,n,\mu)> \varepsilon$ for $x\in[0,x_\varepsilon)$. Then, by the same argument as in the first paragraph of the proof of Proposition \ref{w0Kposw} we obtain $W^N(0,n,\mu)>0$. 

Now we show that $K>0$ implies $M^N(0,n,\mu)>0$. The claim holds for $n=N$ by Proposition \ref{w0Kposw}. Arguing by induction, we assume $M^N(0,n+1,\mu)>0$.
From the previous paragraph, it follows $W^N(0,n+1,\mu)>0$. Then, \eqref{ConnectionWandVnjumps} yields  $V^N(0,n+1,\mu)>-\hat{f}(n+1,\mu)K$ and, plugging the latter, into the definition of $M^N$ in \eqref{Def:M_njumps}, it follows that
\begin{align}
    \begin{aligned}
        M^N(0,n,\mu )&> K \Big( r_n(\mu) \hat{f}(n,\mu) -\lambda_{n+1} \int_{\mu_m}^{\mu_M}\hat{f}(n+1,z) q(n,\mu,\ud z)\Big)> 0,
    \end{aligned}
\end{align}
where the last inequality is derived in Lemma \ref{lem:mw} in Appendix. This concludes the induction step and the proof of (a).

We prove (b). We repeat the exact same arguments as in Part 1 and Part 2 of the proof of Proposition \ref{w0Kposw}-(b), but replacing $M^N(x,N,\mu)$, $W^N(x,N,\mu)$, $\gamma^+_{N}(\mu)$ and $r_N(\mu)$ with $M^N(x,n,\mu)$, $W^N(x,n,\mu)$, $\gamma_{n}^+(\mu)$ and $r_n(\mu)$, respectively. That allows us to show the following: if $M^N(0,n,\mu)<0$ then there exists $b>0$ such that $W^N(x,n,\mu)=0\iff x\in[0,b]$. Thus, it remains to show that $K<0$ implies $M^N(0,n,\mu )<0$. 
The latter holds for $n=N$ by Proposition \ref{w0Kposw}. Arguing by induction, let us assume $M^N(0,n+1,\mu)<0$ so that $W^N(0,n+1,\mu)=0$. Using \eqref{ConnectionWandVnjumps}, we have $V^N(0,n+1,\mu)=-\hat{f}(n+1,\mu)K$. Substituting into \eqref{Def:M_njumps} yields
\begin{align}
    \begin{aligned}
       M^N(0,n,\mu )&=-|K| \Big(r_n(\mu) \hat{f}(n,\mu) -\lambda_{n+1} \int_{\mu_m}^{\mu_M} \hat{f}(n+1,z) q(n,\mu,\ud z)\Big)<0,
    \end{aligned}
\end{align}
where the inequality is by Lemma \ref{lem:mw}. That concludes the induction step and proves (b).
\end{proof}

Analogously to Corollary \ref{cor:CS}, we use Proposition \ref{Kposw} along with positivity and convexity of $x\mapsto W^N(x,n,\mu)$ to describe the structure of the set $\cS_{n,\mu}=\{x\ge0: W^N(x,n,\mu)=0\}$, which must either be empty or connected. All possible shapes of the function $x\mapsto W^N(x,n,\mu)$ are the same as in Figure \ref{Fig:casesN(2)}. That easily justifies the statements in the corollary.
\begin{corollary}\label{cor:CSn}
For $K>0$, three cases may arise:
\begin{itemize}
\item \textit{Case 1}: $\cC_{n,\mu}=[0,\infty)$ and $\cS_{n,\mu}=\varnothing$.
\item \textit{Case 2}: $\cC_{n,\mu}=[0,b)$ and $\cS_{n,\mu}=[b,\infty)$, for some $b=b(n,\mu)\in(0,\infty)$.
\item \textit{Case 3}: $\cC_{n,\mu}\!=\![0,b_1)\cup(b_2,\infty)$\! and $\cS_{n,\mu}=[b_1,b_2]$,\! for some $0\!<\!b_1\!=\!b_1(n,\mu)\!\leq\! b_2(n,\mu)\!=\!b_2\!<\!\!\infty$. 
\end{itemize}
For $K<0$ we have, either 
\begin{itemize}
\item Case 4: $\cC_{n,\mu}=(b,\infty)$ and $\cS_{n,\mu}=[0,b]$ for some $b=b(n,\mu)\in(0,\infty)$, 
\end{itemize}
or 
\begin{itemize}
\item Case 5: $\cC_{n,\mu}=\varnothing$ and $\cS_{n,\mu}=[0,\infty)$.
\end{itemize}
\end{corollary}

In Corollary \ref{cor:CSn}, we identified all possible geometries for the continuation and stopping regions based only on simple considerations around convexity of the map $x\mapsto W^N(x,n,\mu)$. However, the next proposition gives us a deeper insight by showing that indeed Case 3 cannot occur. Notice that $\cS_{n,\mu}\subset\{x\ge 0: M^N(x,n,\mu)\le 0\}$ and recall that for $K>0$, $M^N(0,n,\mu)>0$. Therefore $\{0\}\subset\cC_{n,\mu}$. Combining the two observations, we also deduce that for Case 3 to hold it is necessary that $x\mapsto M^N(x,n,\mu)$ admits a negative minimum in $(0,\infty)$ and, by convexity, $\lim_{x\to\infty}M^N(x,n,\mu)=+\infty$.
\begin{proposition}\label{prop:matteo}
Assume $K>0$. Fix $n\leq N$ and
$\mu\in[\mu_m,\mu_M]$. If
$\lim_{x\to\infty}M^N(x,n,\mu)=+\infty$, then $\inf_{x\geq0}M^N(x,n,\mu)>0$. Therefore, Case 3 of Corollary \ref{cor:CSn} cannot occur.
\end{proposition}

\begin{proof}
The second claim is direct consequence of the first one and of the discussion preceding the statement of the proposition. The rest of the proof is devoted to showing the first claim.

We first prove that for $x>0$, $n\leq N$ and $\mu\in[\mu_m,\mu_M]$, the following holds 
\begin{equation}
\label{upperb-conv}
xD_x^+V^N(x,n,\mu)-V^N(x,n,\mu) \leq K\hat f(n,\mu),    
\end{equation}
where
\begin{equation}
D_x^+V^N(x,n,\mu) \coloneqq \lim_{h\downarrow0} \frac{V^N(x+h,n,\mu)-V^N(x,n,\mu)}{h}.    
\end{equation}
The right derivative exists because $x\mapsto V^N(x,n,\mu)$ is convex.

Let us introduce 
\begin{equation}
    A(\tau) \coloneqq \E_{n,\mu} \Big[ \int_0^\tau \e^{-\int_0^t (\rho+\mu_s) \ud s} (\alpha+\nu\mu_t) X^1_t \ud t + \e^{-\int_0^\tau (\rho+\mu_s) \ud s} \hat f(n_\tau,\mu_\tau)X^1_\tau \Big],
\end{equation}
and
\begin{equation}
B(\tau) \coloneqq \E_{n,\mu} \Big[ \e^{-\int_0^\tau (\rho+\mu_s) \ud s} \hat f(n_\tau,\mu_\tau) \Big].
\end{equation}
Set $f(n_t,\mu_t)=\hat{f}(n_t,\mu_t)/(\hat \rho+\hat \mu)$. It is shown in the appendix (cf.\ \eqref{fchimartingale}) that
\begin{equation}
f(n,\mu)= \E \Big[ {}_\tau p_{\eta} \e^{ -\rho\tau}f(n_{\tau},\mu_\tau) + \int_0^\tau \e^{-\rho u}{}_u p_{\eta}\ud u  \Big].
\end{equation}
Then, using that the integral term is non-negative and recalling that ${}_t p_\eta=\exp(-\int_0^t\mu_s\ud s)$ we obtain $B(\tau)\leq \hat f(n,\mu)$.

Using that $X_t^x=xX_t^1$ in the definition of the value function (cf.\ \eqref{Eq:Value_f_mrkv}), we have 
\[
V^N(x,n,\mu)=\sup_\tau (xA(\tau)-KB(\tau)) \ge xA(\tau)-KB(\tau),
\]
where the inequality holds for every admissible $\tau$. Therefore
\begin{equation}
xA (\tau) \leq V^N(x,n,\mu)+KB(\tau) \leq V^N(x,n,\mu)+K\hat f(n,\mu).
\end{equation}
For $h>0$, the above inequality yields 
\begin{equation}
\begin{aligned}
V^N(x+h,n,\mu)&\le \sup_\tau (x A(\tau)-KB(\tau))+\sup_\tau hA(\tau)\\
& \leq V^N(x,n,\mu) + h\frac{V^N(x,n,\mu)+K\hat f(n,\mu)}{x}.
    \end{aligned}
\end{equation}
Hence
\begin{equation}
\frac{V^N(x+h,n,\mu)-V^N(x,n,\mu)}{h} \leq \frac{V^N(x,n,\mu)+K\hat f(n,\mu)}{x}.
\end{equation}
Letting $h\downarrow0$, we obtain \eqref{upperb-conv}.

Recall from the proof of Proposition \ref{Kposw} that $K>0$ implies $M^N(0,n,\mu)>0$ and recall that
$x\mapsto M^N(x,n,\mu)$ is convex. Since by assumption $\lim_{x\to\infty}M^N(x,n,\mu)=+\infty$, the function $M^N(\cdot,n,\mu)$ admits a global minimizer on $[0,\infty)$. Let $x^*$ be such a minimizer.
If $x^*=0$, then $\inf_{x\geq0}M^N(x,n,\mu)=M^N(0,n,\mu)>0$, and the proof is complete. 
Instead, let us assume $x^*>0$. We can rewrite
\begin{equation}
M^N(x,n,\mu) = r_n(\mu)\hat f(n,\mu)K + C_n(\mu)x + \lambda_{n+1} \int_{\mu_m}^{\mu_M} V^N(x,n+1,z) q(n,\mu,dz),
\end{equation}
with
\begin{equation}
C_n(\mu)\coloneqq \big(\hat f(n,\mu)-\beta_n(\mu)\big) \big(\theta-\alpha-r_n(\mu)\big).
\end{equation}
Since $x^*>0$ is a minimizer and $M^N(\cdot,n,\mu)$ is 
convex, we have $D_x^+ M^N(x^*,n,\mu)\ge 0$. That reads as 
\begin{equation}
\label{foc}
C_n(\mu) + \lambda_{n+1} \int_{\mu_m}^{\mu_M} D_x^+V^N (x^*,n+1,z) q(n,\mu,dz) \ge 0.    
\end{equation}
Then, we have 
\begin{align}
\begin{aligned}
\label{upperboundMxVx}
    M^N(x^*,n,\mu)&=r_n(\mu)\hat f(n,\mu)K + C_n(\mu)x^*+ \lambda_{n+1} \int_{\mu_m}^{\mu_M} V^N(x^*,n+1,z)q(n,\mu,dz)\\
    &\ge r_n(\mu)\hat f(n,\mu)K  + \lambda_{n+1} \int_{\mu_m}^{\mu_M} [V^N(x^*,n+1,z)-x^*D_x^+V^N (x^*,n+1,z)]q(n,\mu,dz)\\
    &\ge K \Big(r_n(\mu)\hat f(n,\mu) - \lambda_{n+1} \int_{\mu_m}^{\mu_M} \hat f(n+1,z)q(n,\mu,dz) \Big)>0
\end{aligned}
\end{align}
where the first inequality holds thanks to \eqref{foc}, the second one holds thanks to \eqref{upperb-conv} and the final inequality is the result of Lemma \ref{lem:mw} in the appendix.
Then, $M^N(x^*,n,\mu)>0$.

Since $x^*$ is a global minimizer of $M^N(\cdot,n,\mu)$, the proof of the first claim is complete.
\end{proof}

Next, we establish the smooth-fit condition and show that the value function and the optimal stopping boundary form a solution to a suitable free-boundary problem.

\begin{proposition} 
For $n\le N$, the mapping $x\mapsto W^N(x,n,\mu)$ is continuously differentiable at $\partial \cC_{n,\mu}\cap(0,\infty)$. Hence, $W^N(\cdot,n,\mu)\in C^1((0,\infty))$.
\end{proposition}

\begin{proof}
The claim holds for $n=N$ by Proposition \ref{smoothfitv1}. Arguing by induction, we assume that $x\mapsto W^N(x,n+1,\mu)$ is continuously differentiable at $\partial \cC_{n+1,\mu}\cap(0,\infty)$. Since it is also Lipschitz (cf.\ Proposition \ref{convexw}), there exists $L>0$ such that $|V^N_x(x,n+1,\mu)|\leq L$ for $x\in(0,\infty)$. By dominated convergence
$|\hat V^N_x(x,n+1;\mu)|  
\le\int_{\mu_m}^{\mu_M} |V^N_x(x,n+1,\vartheta)| q(n,\mu, \ud \vartheta)\le L$.
Using \eqref{Def:M_njumps}, it is now immediate that $
|M^N_x(X^{x}_t,n,\mu )| \leq L$, possibly with a different constant $L>0$. For some $\varrho \in (0,1)$,
$|M^N(X^{b+\varepsilon}_t,n,\mu )-M^N(X^{b}_t,n,\mu )|= |M^N_x(X^{b+\varrho\varepsilon}_t,n,\mu )| X^{1}_t \varepsilon \le L X^{1}_t \varepsilon$, 
by the mean value Theorem.
The rest of the proof follows the same arguments as in the proof of Proposition \ref{smoothfitv1} and it is therefore omitted. 
\end{proof}

Thanks to the asserted regularity of $W^N(\cdot,n,\mu)$ and to \eqref{eq:ODE}, we obtain the next result. 
\begin{proposition}
The value function $V^N(\cdot,n,\mu)\in C^1((0,\infty))\cap C([0,\infty))$ and the set $\mathcal{C}_{n,\mu}$ satisfy the following free-boundary problem: for any $n\le N$ and $\mu\in[\mu_m,\mu_M]$, $V^N\in C^2(\overline\cC_{n,\mu}\cap(0,\infty))$ and
\begin{align}
\begin{alignedat}{2} \label{freebp}
    (\cL- r_n(\mu)) V^N(x,n,\mu )&=-(\alpha+\nu\mu)x-\lambda_{n+1}\hat V^N(x,n+1;\mu), \quad && x\in\mathcal{C}_{n,\mu },\\
    V^N(x,n,\mu)&= \hat f(n,\mu)(x-K), \quad && x\in\partial \mathcal{C}_{n,\mu}\cap(0,\infty),\\
    V^N_x(x,n,\mu)&= \hat f(n,\mu), \quad && x\in\partial \mathcal{C}_{n,\mu}\cap(0,\infty).
\end{alignedat}
\end{align}    
\end{proposition}

\subsection{Proof of Theorem \ref{maintheorem}}\label{sec:proof}

Thanks to Corollary \ref{cor:CSn} we are in a position to prove Theorem \ref{maintheorem}. We recall that Theorem \ref{maintheorem}--(3) was already addressed in Proposition \ref{prop:K=0}. 
\begin{proof}[Proof of Theorem \ref{maintheorem}--(1)]
Let us start from (i). Recall that $M^N(0,n,\mu)>0$ thanks to Proposition \ref{Kposw}. By convexity of $x\mapsto M^N(x,n,\mu)$ it is clear that when $L\in[0,\infty)$ then $M^N(\cdot,n,\mu)$ is decreasing. If $L\in(0,\infty)$ it immediately follows that $M^N(\cdot,n,\mu)$ is strictly positive on $[0,\infty)$. If instead $L=0$, then we suppose by contradiction that $M^N(x^*,n,\mu)=0$ for some $0<x^*<\infty$. Hence $M^N(x,n,\mu)=0$ for $x\ge x^*$ and, recalling the notation for the right derivative in Proposition \ref{prop:matteo}, $D_x^+M^N(x^*,n,\mu)= 0$. The argument in \eqref{upperboundMxVx} can be repeated up to replacing the first inequality with an equality and, again, we reach a contradiction. Therefore, $M^N(x,n,\mu)>0$ for $x\in[0,\infty)$. Finally, if $L=+\infty$, Proposition \ref{prop:matteo} guarantees that $M^N(\cdot,n,\mu)$ is again strictly positive on $[0,\infty)$. Thus, in all cases covered by (i), for all $x\in[0,\infty)$ we have
\begin{equation}\label{wninfty}
W^N(x,n,\mu)\ge w_N(x,n,\mu)\coloneqq \E_x \Big[ \int_0^\infty \e^{-r_n(\mu) t} M^N(X_t,n,\mu)\, \ud t \Big]>0.
\end{equation}
The latter implies the claim in (i). Notice that $w_N(x,n,\mu)$ can be interpreted as the value the individual obtains if they never annuitize.

In order to prove (ii) we first show that $\cS_{n,\mu}\neq\varnothing$.
Arguing by contradiction, suppose that $\mathcal{C}_{n,\mu }=[0,\infty)$, then 
$W^N(x,n,\mu )
=w_N(x,n,\mu )>0$ for all $x\in[0,\infty)$.
Convexity of the mapping $x\mapsto M^N(x,n,\mu)$ and the assumption in (ii) imply that $M^N$ is monotonically decreasing for all $x\in[0,\infty)$. Then, by the monotone convergence theorem
\[
\lim_{x\to\infty}W^N(x,n,\mu )=\lim_{x\to\infty}w_N(x,n,\mu )= \E\Big[ \int_0^\infty \e^{-r_n(\mu) t} \lim_{x\to\infty}M^N(X_t^x,n,\mu ) \ud t\Big]= -\frac{L}{r_n(\mu)}<0,
\]
showing that $W^N(x,n,\mu)<0$ for large values of $x$. Hence, we reach a contradiction and $\cS_{n,\mu}\neq\varnothing$. 
By linearity of $x\mapsto X^x_t=xX^1_t$, for $y>x$ it holds $M^N(X_t^y,n,\mu )=M^N(yX_t^1,n,\mu )\le M^N(xX_t^1,n,\mu )$, implying $x\mapsto W^N(x,n,\mu)$ decreasing on $[0,\infty)$.
Then, $x\in\cS_{n,\mu}\implies y\in\cS_{n,\mu}$ for $y>x$  and $\cC_{n,\mu}=[0,b)$ for $b\in(0,\infty)$.
\end{proof}

\begin{proof}[Proof of Theorem \ref{maintheorem}--(2)]
    By convexity of $M^N(\cdot,n,\mu)$ and $M^N(0,n,\mu)<0$, under the conditions in (iv) it must be $M^N(\cdot,n,\mu)<0$. Then $\E[ \int_0^\tau \e^{-r_n(\mu) t} M^N(X_t^x,n,\mu ) \ud t]\leq 0$, for any $\tau\in\cT(\bF^B)$,
    which implies $W^N(x,n,\mu)=0$, for all $x \in[0,\infty)$, i.e. $\mathcal{S}_{n,\mu }=[0,\infty)$.

    Under the conditions in (iii), it must be $\Gamma\coloneqq \{x\in [0,\infty): M^N(x,n,\mu )>0\}\neq\varnothing$ and open. Since $\mathcal{C}_{n,\mu }\supset\Gamma\neq\varnothing$, we must be in Case 4 from Corollary \ref{cor:CSn}. 

  Again recalling convexity of $M^N(\cdot,n,\mu)$ and $M^N(0,n,\mu)<0$ it is clear that no other case may arise when $K<0$.
\end{proof}

\subsection{Asymptotic properties of $M^N$}

Theorem \ref{maintheorem} highlights how the structure of $\cC_{n,\mu}$ 
depends on the sign and asymptotic behaviour of the function $M^N(\cdot,n,\mu)$. In order to gain some insight on conditions that determine the value of $\lim_{x\to\infty}M^N(x,n,\mu)$ we now prove two auxiliary results. In particular, Corollary \ref{Vnatinftpropcor} yields sufficient conditions for $\lim_{x\to\infty}M^N(x,n,\mu)=\pm\infty$.

Let us start by introducing, for $n\in\N_0$, $n\le N-1$, 
\begin{align} 
\begin{aligned} \label{constantvcninfty} 
V_{n}^\infty(\mu)\coloneqq \sup_{\tau\in\cT(\bF^B)} \E \Big[ \e^{- r_n(\mu) \tau}  \hat{f}(n,\mu) X^1_\tau+\int_0^\tau \e^{- r_n(\mu) t} X^1_t\big( \alpha+\nu \mu +\lambda_{n+1} \hat V_{n+1}^\infty(\mu)\big) \ud t \Big],
\end{aligned}
\end{align} 
with $\hat V^\infty_{n+1}(\mu)\coloneqq \int_{\mu_m}^{\mu_M} V^\infty_{n+1}(z)q(n,\mu,\ud z)$.
This generalises $V_{N}^\infty(\mu)$ defined in \eqref{DefVnmuinfty}.
Using the same change of measure \eqref{changeofmeasure} as in Lemma \ref{Lemma:VNmuinfty} we obtain
\begin{equation}\label{eq:Vninfty}
\begin{aligned}
V_{n}^\infty(\mu)&=\sup_{t\ge 0}\Big(\e^{(\theta-\alpha-r_n(\mu))t}\hat f(n,\mu)+\int_0^t\e^{(\theta-\alpha-r_n(\mu))s}\big( \alpha+\nu \mu +\lambda_{n+1} \hat V_{n+1}^\infty(\mu)\big) \ud s\Big)\\
&=\max\Big\{\hat f(n,\mu),\frac{\alpha+\nu \mu +\lambda_{n+1} \hat V_{n+1}^\infty(\mu)}{r_n(\mu)+\alpha-\theta}\Big\},
\end{aligned}
\end{equation}
where we recall that $r_n(\mu)+\alpha-\theta>0$. Because $V^\infty_{N}(\mu)$ is known explicitly from Lemma \ref{Lemma:VNmuinfty}, then $V^\infty_n(\mu)$ can be computed explicitly for all $0\le n\le N-1$ by simple recursive formulae.

\begin{proposition} \label{Vnatinftprop} 

For $(n,\mu)\in\N_0\times[\mu_m,\mu_M]$, $n\le N$, we have
$\lim_{x\to\infty}V^N(x,n,\mu)/x= V_{n}^\infty(\mu)$.
\end{proposition}
\begin{proof}
The proof is by induction. The claim holds for $V^N(x,N,\mu)$, thanks to Proposition \ref{prop:asympt}.
Let us now suppose that  
 $\lim_{x\to\infty} V^N(x,n+1,\mu)/x= V_{n+1}^\infty(\mu)$.
Then, thanks to Proposition \ref{Vnfiniteness}, we can use dominated convergence to obtain
\begin{align}
\begin{aligned} \label{vnmeno1}
\lim_{x\to\infty}\frac{\hat V^N(x,n+1;\mu)}{x}
&=\int_{\mu_m}^{\mu_M} V^\infty_{n+1}(z) q(n,\mu,\ud z)=\hat V^\infty_{n+1}(\mu).
\end{aligned}
\end{align}
For $K>0$, we can drop the fee payment in $V^N(x,n,\mu)$ and obtain the upper bound
\begin{align}
    \begin{aligned} \label{vnpremaggio0}
    \frac{V^N(x,n,\mu)}{x}
    &\leq  \sup_{\tau\in\cT(\bF^B)} \E \Big[ \e^{- r_n(\mu) \tau}  \hat{f}(n,\mu) X^1_\tau\\
    &\qquad\qquad +\int_0^\tau \e^{- r_n(\mu) t} \Big( ( \alpha+\nu \mu  ) X^1_t +\lambda_{n+1} \frac{\hat V^N(X^x_t,n+1;\mu)}{x} \Big) \ud t \Big]. 
    \end{aligned}
\end{align}
We focus on the last term inside the integral. For fixed $\mu\in[\mu_m,\mu_M]$, \eqref{vnmeno1} implies that for any $\varepsilon>0$, we can find $S=S_{n,\mu}(\eps)>0$ such that with $x>S$,
\begin{align} 
\begin{aligned} \label{vcpiu1nlimit}
\Big|\frac{\hat V^N(x,n+1;\mu)}{x}- \hat V_{n+1}^\infty(\mu)\Big|<\varepsilon, 
\end{aligned}
\end{align}  
yielding
$\ind_{\{X_t^x > S\}} X^1_t [\hat V^N(X^x_t,n+1;\mu)/X^x_t] \leq \ind_{\{X_t^x > S\}} X^1_t (\hat V_{n+1}^\infty(\mu)+\varepsilon)$.
Hence
\begin{align}
    \begin{aligned} \label{vnpremaggio2}
    \frac{\hat V^N(X^x_t,n+1;\mu)}{x}
    &=X^1_t \frac{\hat V^N(X^x_t,n+1;\mu)}{X^x_t} \Big(\ind_{\{X_t^x\leq S\}}+\ind_{\{X_t^x > S\}}\Big)\\
    &\leq \ind_{\{X_t^x\leq S\}} X^1_t \frac{\hat V^N(X^x_t,n+1;\mu)}{X^x_t}+ \ind_{\{X_t^x > S\}} X^1_t (\hat V_{n+1}^\infty(\mu)+\varepsilon)\\
    &\leq \ind_{\{X_t^x\leq S\}}  \frac{\hat V^N(X^x_t,n+1;\mu)}{x} + X^1_t (\hat V_{n+1}^\infty(\mu)+\varepsilon)\\
    &\leq x^{-1}\sup_{0\leq y \leq S}|\hat V^N(y,n+1;\mu)|+  X^1_t (\hat V_{n+1}^\infty(\mu)+\varepsilon),
    \end{aligned}
\end{align}
where 
in the second inequality we remove the indicator because $\hat V_{n+1}^\infty(\mu)+\varepsilon$ is positive. 
Combining \eqref{vnpremaggio0} and \eqref{vnpremaggio2} yields
\begin{align*}
    \begin{aligned}
    \frac{V^N(x,n,\mu)}{x}
    &\leq  \sup_{\tau\in\cT(\bF^B)} \E \Big[ \e^{- r_n(\mu) \tau}  \hat{f}(n,\mu) X^1_\tau\!+\!\int_0^\tau\!\! \e^{- r_n(\mu) t} X^1_t\Big( \alpha\!+\!\nu \mu \!+\!\lambda_{n+1} \hat V_{n+1}^\infty(\mu)\!+\!\lambda_{n+1}\eps\Big) \ud t \\
    &\qquad\qquad +\int_0^\tau\!\! \e^{- r_n(\mu) t} \lambda_{n+1} \sup_{0\leq y \leq S}|\hat V^N(y,n+1;\mu)|/x  \ud t\Big]\\
    &\le V^\infty_{n}(\mu)\!+\!\int_0^\infty\e^{-r_n(\mu)t}\E\big[X^1_t\big]\lambda_{n+1}\eps\ud t\!+\!\sup_{0\leq y \leq S}|\hat V^N(y,n+1;\mu)|/x \int_0^\infty\e^{-r_n(\mu)t}\lambda_{n+1}\ud t.
    \end{aligned}
\end{align*}
The first integral is finite because $\E[X^1_t]=\e^{(\theta-\alpha)t}$ and $r_n(\mu)>\theta-\alpha$. Letting $x\to\infty$, the last term in the above expression vanishes because $\hat V^N(\cdot,n+1;\mu)$ is bounded on compacts. Afterwards we let $\eps\to 0$, so that 
$\limsup_{x\to\infty}V^N(x,n,\mu)/x \leq V_{n}^\infty(\mu)$.

Next, we prove the reverse inequality. 
From \eqref{vcpiu1nlimit},
$\hat V^N(X^x_t,n+1;\mu)/X^x_t \geq (\hat V_{n+1}^\infty(\mu)-\varepsilon)$ on the event $\{X_t^x > S\}$.
Hence,
\begin{align}
    \begin{aligned}
    \frac{\hat V^N(X^x_t,n+1;\mu)}{x}&\geq \ind_{\{X_t^x\leq S\}} X^1_t \frac{\hat V^N(X^x_t,n+1;\mu)}{X^x_t} + \ind_{\{X_t^x > S\}} X^1_t (\hat V_{n+1}^\infty(\mu)-\varepsilon)\\
    &\geq \ind_{\{X_t^x\leq S\}}  \frac{\hat V^N(X^x_t,n+1;\mu)}{x} - \ind_{\{X_t^x\leq S\}}X_t^1 \hat V_{n+1}^\infty(\mu) +  X^1_t (\hat V_{n+1}^\infty(\mu)-\varepsilon)\\
    &\geq -x^{-1}\sup_{0\leq y \leq S}|V^N(y,n+1,\mu)| - \frac{S}{x}\hat V_{n+1}^\infty(\mu) +  X^1_t (\hat V_{n+1}^\infty(\mu)-\varepsilon).
    \end{aligned}
\end{align}
For large $x$ we have $\frac{K}{x}<\varepsilon$ and therefore 
\begin{align}
    \begin{aligned} 
    \frac{V^N(x,n,\mu)}{x}&= \sup_{\tau\in\cT(\bF^B)} \E \Big[ \e^{- r_n(\mu) \tau}  \hat{f}(n,\mu) \Big(X^1_\tau-\frac{K}{x}\Big)\\
    &\qquad\qquad +\int_0^\tau \e^{- r_n(\mu) t} \Big( ( \alpha+\nu \mu  ) X^1_t +\lambda_{n+1} \frac{\hat V^N(X^x_t,n+1;\mu)}{x} \Big) \ud t \Big] \\
    &\geq  \sup_{\tau\in\cT(\bF^B)} \E \Big[ \e^{- r_n(\mu) \tau}  \hat{f}(n,\mu) (X^1_\tau-\varepsilon) + \int_0^\tau  \e^{- r_n(\mu) t} X^1_t\Big(  \alpha+\nu \mu   +\lambda_{n+1} (\hat V_{n+1}^\infty(\mu)-\varepsilon ) \Big) \ud t\Big]\\
    &\quad  -\frac{S \hat V_{n+1}^\infty(\mu)+\sup_{0\leq y \leq S}|\hat V^N(y,n+1;\mu)|}{x}\int_0^\infty \e^{- r_n(\mu) t} \lambda_{n+1}   \ud t.
    \end{aligned}
\end{align}
Letting $x\to\infty$, the last term in the expression above vanishes. Then, letting also $\eps \to 0$ we get
$\liminf_{x\to\infty} V^N(x,n,\mu)/x\geq V_{n}^\infty(\mu)$. 

We have shown $\lim_{x\to\infty} V^N(x,n,\mu)/x = V_{n}^\infty(\mu)$ for $K>0$. The case $K<0$ is analogous.
\end{proof}

The above proposition implies that $\hat{V}^N(x, n+1; \mu)$ behaves like $x\cdot\hat V_{n+1}^\infty (\mu)$ when $x\to \infty$ (cf.\ \eqref{vnmeno1}). 
Therefore, the asymptotic behaviour of $M^N$ can be summarised in the next corollary.
\begin{corollary}
\label{Vnatinftpropcor}
Let $P(n,\mu)\coloneqq (\hat{f}(n,\mu)-\beta_{n}(\mu))(\theta-\alpha- r_n(\mu)) +\lambda_{n+1} \hat V_{n+1}^\infty(\mu)$
and recall \eqref{Def:M_njumps}. Then, 
\begin{equation*}
\begin{aligned}
&P(n,\mu)>0\implies \lim_{x\to\infty} M^N(x,n,\mu )=+\infty,\\
&P(n,\mu)<0\implies \lim_{x\to\infty} M^N(x,n,\mu )=-\infty.
\end{aligned}
\end{equation*}
Finally, if $P(n,\mu)=0$, then $\lim_{x\to \infty}M^N(x,n,\mu)/x=0$ and, using the convexity of $M^N(\cdot,n,\mu)$, either $M^N$ diverges to $-\infty$ or it converges to a finite value $L\le M^N(0,n,\mu)$. \end{corollary}

It is worth noticing that the corollary above, combined with the conditions in Theorem \ref{maintheorem}, gives a simple method to determine a-priori the geometry of the continuation/stopping set without calculating the value functions $V^N(x,n,\mu)$ for $0\le n\le N$. Indeed, the values of $V^\infty_n(\mu)$ are given explicitly in \eqref{eq:Vninfty} in terms of the model parameters only.

\appendix
\section{Markovian structure of the problem} \label{app:markovian}
In this appendix, we present results concerning the Markovian structure of the problem.
First, we show that $a_{\eta+t}$, defined in \eqref{Eq: a_eta+t}, is an $\mathcal{E}$-measurable function of $(n_t,\mu_t)$.
In particular, it holds
\begin{align}\label{annuitymarkov}
\begin{aligned}  
a_{\eta+t}&=\E \Big[\int_{0}^{\infty} \e^{-\rho u} {}_u p_{\eta+t}\ud u  \Given \cF_t \Big] = \int_{0}^{\infty} \e^{-\rho u} \E \Big[ \e^{ -\int_{t}^{t+u}\mu_s\ud s} \Given \cH_t \Big]\ud u\\
&=\int_{0}^{\infty} \e^{-\rho u} \E \Big[ \e^{ -\int_{t}^{t+u}\mu_s\ud s} \Given \sigma(n_t,\mu_t) \Big]\ud u, 
\end{aligned}
\end{align}
where the second equality is by independence of $\cF^B_\infty$ and $\cH_\infty$ and Fubini's theorem, the third equality is by Markov property of $\chi_t=(n_t,\mu_t)$.
By \cite[Ch.\ 1.3, Lemma 1.1, p.5]{baldi2017stochastic}, for any fixed $u$ there exists a $\mathcal{E}$ measurable function $\tilde{f}(u,\cdot)$ such that 
$\E [ \e^{ -\int_{t}^{t+u}\mu_s\ud s} | \sigma(n_t,\mu_t) ]= \tilde{f}(u,n_t,\mu_t)$. 
Moreover, $u\mapsto \tilde{f}(u,n,\mu)$ is continuous uniformly in $(n,\mu)$, because $\e^{-x}$ is 1-Lipschitz for $x\ge 0$, and therefore 
\begin{equation*}
\begin{aligned}
        |\tilde{f}(u,n,\mu)-\tilde{f}(v,n,\mu)|&\leq \E\Big[\big|\e^{ -\int_{t}^{t+u}\mu_s\ud s} -\e^{ -\int_{t}^{t+v}\mu_s\ud s} \big|\Big|\sigma(n_t,\mu_t)\Big]\\
        &\le \E \Big[\int_{t+u}^{t+v}\mu_s \ud s \Big| \sigma(n_t,\mu_t) \Big] \leq \mu_M|v-u|.
\end{aligned}
\end{equation*}
Then, $(u,n,\mu)\mapsto \tilde{f}(u,n,\mu)$ is $\mathcal{B}(\R_+)\otimes \mathcal{E}$ measurable, by \cite[Ch.\ 4.10, Lemma 4.51, p.\ 153]{Aliprantis06infinitedimenstional}.
Hence, continuing \eqref{annuitymarkov},
\begin{align} \begin{aligned} \label{annuitymarkov2}
    a_{\eta+t}&=\int_{0}^{\infty} \e^{-\rho u} \tilde{f}(u,n_t,\mu_t)\ud u\eqqcolon f(n_t,\mu_t),
    \end{aligned}\end{align}
where $f$ is $\mathcal{E}$ measurable, by Fubini's Theorem, see \cite[Ch.\ 1.3, Thm 1.2, p. 8-9]{baldi2017stochastic}.

Next, we embed the individual optimisation problem \eqref{Def:Value_f_nonmrkv} into a Markovian framework. 
With no loss of generality, by independence of $\Theta$ from $\cF_\infty$ we enforce throughout the paper a product-space structure:
\begin{equation}\label{eq:probabspace}
(\Omega,\cG,\P)=(\Omega^1\times\Omega^2,\cG^1\otimes\cG^2,\P^1\otimes\P^2).
\end{equation}
For $\omega=(\omega_1,\omega_2)\in\Omega^1\times\Omega^2$, we have
$((X_t(\omega),\chi_t(\omega))_{t\ge 0},\Theta(\omega))=((X_t(\omega_1),\chi_t(\omega_1))_{t\ge 0},\Theta(\omega_2))$.
The expectations under $\P^i$ are indicated by $\E^i[\cdot]$ for $i=1,2$. 
\begin{proposition} \label{prop:markovian}
For $\P_{x,n,\mu}(\cdot)=\P(\cdot|X_0=x,n_0=n,\mu_0=\mu)$, letting     
\begin{align*}
    \begin{aligned} 
    V(x,n,\mu)\coloneqq \sup_{\tau \in\cT(\bF)} \E_{x,n,\mu} \Big[& \int_{0}^{\tau}\e^{-\int_0^t (\rho+\mu_s)\ud s} ( \alpha+\nu \mu_t ) X_t\ud t+ \e^{-\int_0^\tau (\rho+\mu_s)\ud s} \hat{f}(n_\tau,\mu_\tau)(X_\tau-K)\Big],
    \end{aligned}
\end{align*}
we have $V(X_0,0,\mu_0)=V_0$ (cf.\ \eqref{Def:Value_f_nonmrkv}).
\end{proposition}
\begin{proof}
By definition of $\tau_d$ and the product space structure in \eqref{eq:probabspace} we have
\[
\tau_d(\omega)=\tau_d(\omega_1,\omega_2)=\inf\{t\ge 0:\Lambda_t(\omega_1)\ge \Theta(\omega_2)\}.
\]
Then we can decompose $\tau_d(\omega)=\hat \tau(\omega_1,\Theta(\omega_2))$ with 
\[
\hat \tau(\omega_1,\vartheta)\coloneqq\inf\{t\ge 0:\Lambda_t(\omega_1)\ge \vartheta\}\quad\text{for $\vartheta\in[0,\infty)$}.
\]
Sometimes we omit the dependence on $\omega_1$ and simply write $\hat \tau(\vartheta)$ instead of $\hat \tau(\omega_1,\vartheta)$.
Moreover, $\tau\in\cT(\bF)$ implies $\tau(\omega)=\tau(\omega_1)$.

Thanks to the product space structure, for any bounded measurable function $F(\omega_1,\vartheta)$ by Fubini's theorem, we obtain
\[
\E\big[F(\cdot,\Theta(\cdot))\big]=\int_{\Omega_1}\int_{\Omega_2}F(\omega_1,\Theta(\omega_2))\ud \P^2(\omega_2)\ud\P^1(\omega_1)=\int_{\Omega_1}\int_0^\infty\e^{-\vartheta}F(\omega_1,\vartheta)\ud \vartheta\ud \P^1(\omega_1),
\]
where in the second equality we used that the law $\P^2\circ\Theta^{-1}$ is an exponential of parameter one. Now, using this observation we can rewrite
\begin{align*}
&\E \Big[  \int_{0}^{\tau_d \wedge \tau }\e^{-\rho t}\alpha X_t\ud t  +\ind_{\{\tau_d \leq \tau \}}\e^{-\rho \tau_d} \nu X_{\tau_d} +P_{\tau}\hspace{-3pt}\int_{\tau_d\wedge \tau}^{\tau_d}\e^{-\rho t}\ud t \Big]\\
&=\E^1 \Big[\int_0^\infty \!\Big(\int_{0}^{\hat \tau(\vartheta) \wedge \tau }\!\!\e^{-\rho t}\alpha X_t\ud t\Big)\e^{-\vartheta} \ud \vartheta \\
&\quad \quad \quad +\int_0^\infty\!\!\ind_{\{\hat \tau(\vartheta) \leq \tau \}}\e^{-\rho \hat \tau(\vartheta)} \nu X_{\hat \tau(\vartheta)}\e^{-\vartheta}\ud \vartheta +P_{\tau}\hspace{-3pt}\int_0^\infty\Big(\int_{\hat \tau(\vartheta)\wedge \tau}^{\hat \tau(\vartheta)}\!\!\e^{-\rho t}\ud t\Big)\e^{-\vartheta}\ud \vartheta \Big]\\
&=\E^1 \Big[\int_0^\infty \!\Big(\int_{0}^{\infty}\!\!\ind_{\{t\le \hat \tau(\vartheta) \}}\e^{-\vartheta} \ud \vartheta\Big)\e^{-\rho t}\alpha X_t 1_{\{t\le \tau\}}\ud t \\
&\quad \quad \quad +\int_0^\infty\!\!\ind_{\{\hat \tau(\vartheta) \leq \tau \}}\e^{-\rho \hat \tau(\vartheta)} \nu X_{\hat \tau(\vartheta)}\e^{-\vartheta}\ud \vartheta +P_{\tau}\hspace{-3pt}\int_0^\infty\Big(\int_0^\infty \ind_{\{t\le \hat \tau(\vartheta)\}}\e^{-\vartheta}\ud \vartheta\Big)\e^{-\rho t}1_{\{t\ge \tau\}}\ud t \Big].
\end{align*}
Since $t\mapsto\Lambda_t$ is strictly increasing, its inverse is continuous. Continuing from the expression above and using \cite[Ch.\ 0.4, pp. 7-9]{revuz2013continuous} we obtain
\begin{align*}
& \E^1 \Big[\int_0^\infty \Big(\int_{t}^{\infty }\e^{-\Lambda_s} \ud \Lambda_s\Big)\! \ind_{\{t\leq \tau\}} \e^{-\rho t}\alpha X_t \ud t \\
&\quad +\int_0^\tau \e^{-\rho t} \nu X_{t}\e^{-\Lambda_t} \ud\Lambda_t +P_{\tau}\hspace{-3pt}\int_{0}^{\infty }\!\!\Big(\int_{t}^{\infty }\e^{-\Lambda_s} \ud \Lambda_s\Big)\ind_{\{t\ge \tau\}}\e^{-\rho t}\ud t \Big]\\
&=\E^1 \Big[ \int_{0}^{\tau}\e^{-\rho t} {}_tp_{\eta} ( \alpha+\nu \mu_t ) X_t\ud t+P_\tau \int_{\tau}^\infty\e^{-\rho t} {}_tp_{\eta} \ud t\Big],
\end{align*}
where the equality holds upon observing that
$\int_{t}^{\infty }\e^{-\Lambda_s} \ud \Lambda_s=\e^{-\Lambda_t}= {_t} p_\eta$ and $\ud \Lambda_t=\mu_t\ud t$.

Using \eqref{Eq: a_eta+t}, it is not difficult to check that 
\begin{align}\label{eq:mr1}
\E^1 \Big[ \int_{\tau}^{\infty}\e^{-\rho t} \ {}_tp_{\eta} \ud t \Given \cF_\tau\Big]={}_\tau p_{\eta} \ \e^{-\rho \tau} \ a_{\eta+\tau}.
\end{align}
Finally, we express everything under the measure $\P$ for notational simplicity and upon observing that for any $Y:\Omega_1\to\R$, $\E^1[Y]=\E[Y]$. Expanding the expression of $P_\tau$ we reach our problem formulation
\begin{align*}
&\sup_{\tau\in\cT(\bF)}\E \Big[ \int_{0}^{\tau}\e^{-\rho t} {}_tp_{\eta} ( \alpha+\nu \mu_t ) X_t\ud t+\e^{-\rho \tau} {}_\tau p_{\eta} \frac{a_{\eta+\tau}}{\hat{a}_{\eta+\tau}}(X_\tau-K)\Big]\\
&=\sup_{\tau\in\cT(\bF)}\E \Big[ \int_{0}^{\tau}\e^{-\rho t} {}_tp_{\eta} ( \alpha+\nu \mu_t ) X_t\ud t+\e^{-\rho \tau} {}_\tau p_{\eta} \hat f(n_\tau,\mu_\tau)(X_\tau-K)\Big],
\end{align*}
where, recalling $f(n_t,\mu_t)$ from \eqref{annuitymarkov2}, we have set  
$a_{\eta+t}/\hat{a}_{\eta+t}=(\hat \rho+\hat \mu)f(n_t,\mu_t)\eqqcolon\hat{f}(n_t,\mu_t)$.

Since the process $(X_t,n_t,\mu_t)_{t\ge 0}$ is a Markov process, we embed the optimization problem into a Markovian framework by evaluating the objective function above starting from any initial condition $(x,n,\mu)\in (0,\infty)\times\N_0\times[\mu_m,\mu_M]$. That yields the measurable function $V(x,n,\mu)$ and the equivalence $V(X_0,0,\mu_0)=V_0$ is by construction.
\end{proof}
Next we state a property of the money's worth, which we use in the proof of Proposition \ref{Kposw}.
\begin{lemma}\label{lem:mw}
For all $(n,\mu)$ it holds
   $(\rho+\mu) \hat{f}(n,\mu) >\lambda_{n+1} \big(\int_{\mu_m}^{\mu_M}\hat{f}(n+1,z) q(n,\mu,\ud z)-\hat{f}(n,\mu)\big)$.
\end{lemma}
\begin{proof}
Recall \eqref{annuitymarkov2} and set $f(n_t,\mu_t)=f(\chi_t)$ with a slight abuse of notation. From \eqref{Eq: a_eta+t} we obtain 
\[
_{s} p_{\eta+t}\e^{-\rho s}f(\chi_{t+s})+\int_0^s\e^{-\rho u} {_{u}p_{\eta+t}}\ud u=\E\Big[\int_0^\infty\e^{-\rho u} {_u p_{\eta+t}}\ud u\Big|\cF_{t+s}\Big],\quad \text{for $s\ge 0$},
\]
showing that the process
$s\mapsto {_{s} p_{\eta+t}}\e^{-\rho s}f(\chi_{t+s})+\int_0^s\e^{-\rho u} {_{u}p_{\eta+t}}\ud u$
is indeed a bounded $\bF$-martingale (and a $\bH$-martingale). Hence, for any $\bF$ stopping time $\tau\ge 0$, by optional sampling
\begin{equation}
\label{fchimartingale}
f(\chi_t)=\E\Big[{_{\tau} p_{\eta+t}}\e^{-\rho \tau}f(\chi_{t+\tau})+\int_0^\tau\e^{-\rho u} {_{u}p_{\eta+t}}\ud u\Big|\cF_t\Big].
\end{equation}

In particular, when $\chi_t=(n,\mu)$ and $\tau=\xi_n$, that implies
\begin{align*}
f(n,\mu)&>\E_{n,\mu}\Big[\e^{-(\rho+\mu) \xi_n}f(n+1,\bar\mu_{n+1})\Big]\\
&=\frac{\lambda_{n+1}}{\rho+\mu+\lambda_{n+1}}\E_{n,\mu}\Big[f(n+1,\bar\mu_{n+1}\big)\Big]=\frac{\lambda_{n+1}}{\rho+\mu+\lambda_{n+1}}\int_{\mu_m}^{\mu_M}f(n+1,z)q(n,\mu,\ud z),
\end{align*}
where the first equality is obtained by integrating out the random time $\xi_n$ (recall independence of $\xi_n$ and $\bar\mu_{n+1}$) and the second one uses the law of $\bar\mu_{n+1}$ (cf.\ Remark \ref{rem:q}). The claim in the lemma follows.
\end{proof}

\begin{proof}[{\bf Proof of Proposition \ref{Prop:recursive}}]
\label{app.proof} For $n=N$, the problem reduces to a one-dimensional optimal stopping problem for the process $(X_t)_{t\ge 0}$. Hence, we can restrict our attention to $\cT(\bF^B)$-stopping times.
In particular, \eqref{VF:recursive} holds for $V^N(x,N,\mu)$.
Next, we prove \eqref{VF:recursive} for $n< N$. Recall $\xi_n\coloneqq \inf\{t>0: n_t\neq n\}$. 
The Dynamic Programming Principle (DPP)  yields\footnote{The argument is standard, but interested readers can find further details in \cite[Ch. 4, Prop. 4.1.4, p. 55]{mythesis}.}:
    \begin{align}
        \begin{aligned}  \label{Eq:DPP}
            V^N(x,n,\mu)=\sup_{\tau \in \cT(\bF)} \E_{x,n,\mu}\Big[& \int_{0}^{\tau\wedge \xi_n}\!\!\e^{-(\rho+\mu) t} ( \alpha\!+\!\nu \mu ) X_t\ud t\!+\!\ind_{\{\tau<\xi_n\}} \e^{-(\rho+\mu) \tau}   \hat{f}(n,\mu) (X_\tau\!-\!K) \\
            &+ \ind_{\{\tau\geq\xi_n\}}\e^{-(\rho+\mu)\xi_n}V^N(X_{\xi_n},n+1,\bar\mu_{n+1}) \Big].
        \end{aligned}
    \end{align} 
Thanks to independence of $\xi_n$ and $\cF^B_\infty$, in \eqref{Eq:DPP} we can restrict our attention to stopping times from $\cT(\bF^B)$ by the same argument used to obtain \eqref{Def:Value_f_nonmrkv}. Integrating with respect to the density of $\xi_n$ we further simplify the problem. Indeed, with no loss of generality we may assume that $(\Omega^1,\cG^1,\P^1)$ in \eqref{eq:probabspace} is itself a product space, that accommodates mutually independent random variables and processes $\bar \mu_{n+1}$, $\xi_n$ and $(X_t)_{t\ge 0}$. 
For any $\tau\in\cT(\bF^B)$, by Fubini's theorem we deduce
\begin{align}
    \begin{aligned}  
&\E_{x,n,\mu}\Big[\int_{0}^{\tau\wedge \xi_n}\e^{-(\rho+\mu) t} ( \alpha+\nu \mu ) X_t\ud t+\ind_{\{\tau<\xi_n\}} \e^{-(\rho+\mu) \tau}   \hat{f}(n,\mu) (X_\tau-K)\Big]\\
&=\E_{x,n,\mu}\Big[\int_0^\infty\Big(\int_{0}^{\tau\wedge s}\e^{-(\rho+\mu) t} ( \alpha+\nu \mu ) X_t\ud t\Big)\lambda_{n+1} \e^{-\lambda_{n+1} s}\ud s+\e^{-\lambda_{n+1} \tau} \e^{-(\rho+\mu) \tau}   \hat{f}(n,\mu) (X_\tau-K)\Big]\\
&=\E_{x,n,\mu}\Big[\int_{0}^{\tau}\e^{-r_n(\mu) t} ( \alpha+\nu \mu ) X_t\ud t+ \e^{-r_n(\mu) \tau}   \hat{f}(n,\mu) (X_\tau-K)\Big].
\end{aligned}
\end{align}
Similarly,
\begin{align}
    \begin{aligned}  
&\E_{x,n,\mu}\Big[\ind_{\{\tau\geq\xi_n\}}\e^{-(\rho+\mu)\xi_n}V^N(X_{\xi_n},n+1,\bar\mu_{n+1}) \Big]\\
&=\E_{x,n,\mu}\Big[\int_0^\tau \lambda_{n+1}\e^{-\lambda_{n+1} t}\e^{-(\rho+\mu)t}V^N(X_{t},n+1,\bar\mu_{n+1})\ud t\Big]\\
&=\E_{x,n,\mu}\Big[\int_0^\tau \lambda_{n+1}\e^{-r_n(\mu)t}\hat{V}^N(X_t ,n+1;\mu)\ud t\Big],
\end{aligned}
\end{align}
with $\hat{V}^N(x,n+1;\mu)$ as in \eqref{Eq:hatV}.
Combining \eqref{Eq:DPP} with the expressions above completes the proof.
\end{proof}
\medskip
\noindent{\bf Funding}: G. Stabile received financial support from EU -- Next Generation EU -- PRIN2022 (2022FWZ2CR) CUP:
B53D23009970006. G. Stabile was also partially supported by Sapienza University of Rome, research project ``\emph{On the role of bequest motives in the annuitization decision}'', grant no.~RP124190EA4DDA9F.
\medskip

T.\ De Angelis received financial support from -- Next Generation EU -- PRIN2022 (2022BEMMLZ) CUP: D53D23005780006 and PRIN-PNRR2022 (P20224TM7Z) CUP: D53D23018780001.

\medskip

M. Buttarazzi carried out part of this work during his PhD at Sapienza University of Rome and part while he was a postdoctoral researcher at the University of Turin. He gratefully acknowledges financial support from both institutions, including funding from Sapienza University of Rome through the research project  ``\emph{Optimal annuitization under piecewise deterministic mortality force}'', grant no.~ AR2241906EB2D3B9.

\bibliographystyle{plain}
\bibliography{biblio}

\end{document}